\documentclass[
  reprint,
  superscriptaddress,
  amsmath,amssymb,
  aps,
  prl,
  secnumarabic
]{revtex4-2}

\usepackage{etoolbox}
\usepackage{needspace}

\makeatletter

\newcommand{\appendixtableofcontents}{%
  \section*{Appendix Contents}%
  \@starttoc{atoc}%
}

\let\latexappendix\appendix

\renewcommand{\appendix}{%
  \latexappendix
  \let\oldsection\section
  \let\oldsubsection\subsection

  \renewcommand{\section}{%
    \@ifstar{\oldsection*}{\appendixsection}%
  }%
  \renewcommand{\subsection}{%
    \@ifstar{\oldsubsection*}{\appendixsubsection}%
  }%
}

\newcommand{\appendixsection}[1]{%
  \oldsection{#1}%
  \addcontentsline{atoc}{section}{\protect\numberline{\thesection}#1}%
}

\newcommand{\appendixsubsection}[1]{%
  \oldsubsection{#1}%
  \addcontentsline{atoc}{subsection}{\protect\numberline{\thesubsection}#1}%
}

\makeatother

\def\@seccntformat#1{\csname the#1\endcsname\quad}

\usepackage{amsmath,amssymb,amsfonts,amsthm}
\usepackage{mathtools,mathrsfs}

\DeclarePairedDelimiter\floor{\lfloor}{\rfloor}
\usepackage{bbm}
\usepackage{dsfont}
\usepackage{graphicx}   % figures
\usepackage{bm}    % bold math
\usepackage{braket}     % Dirac notation
\usepackage{physics}   % optional, use sparingly
\usepackage{hyperref}
\hypersetup{
  colorlinks=true,
  linkcolor=black,
  citecolor=black,
  urlcolor=black
}
\usepackage{stmaryrd}
\usepackage{pgf,tikz}

\usetikzlibrary{positioning,fit,calc,arrows.meta}

\newtheorem{theorem}{Theorem}[section]
\newtheorem{lemma}[theorem]{Lemma}

\newtheorem{proposition}[theorem]{Proposition}

\newtheorem{remark}[theorem]{Remark}

\newcommand{\NA}{N_{\mathcal{A}}}
\newcommand{\id}{\mathbb{I}}
\newcommand{\cE}{\mathcal{E}}
\newcommand{\hatfM}{\widehat{f}_{\mathscr{M}}}
\newcommand{\nn}{\nonumber}
\newcommand{\1}{\mathbbm{1}}

\newcommand{\kb}[1]{|#1\rangle\!\langle#1|}

\newcommand{\R}{\mathbb{R}}
\newcommand{\C}{\mathbb{C}}
\newcommand{\N}{\mathbb{N}}
\newcommand{\cA}{\mathcal{A}}
\newcommand{\cH}{\mathcal{H}}

\newcommand{\cL}{\mathcal{L}}
\newcommand{\cN}{\mathcal{N}}

\newcommand{\spec}{\operatorname{Sp}}

\newcommand{\eps}{\varepsilon}
\renewcommand{\epsilon}{\varepsilon}
\renewcommand{\phi}{\varphi}

\newcommand{\Ntot}{N_{\operatorname{tot}}}
\newcommand{\bin}[1]{}

\newcounter{BHSchurFootnote}

\usepackage{xcolor}
\begin{document}

% Optional: additional authors/affiliations
% \author{Second Author}
% \affiliation{...}

%\date{}

\title{Simulating Thermal Properties of Bose--Hubbard Models on a Quantum Computer}

\author{Simon Becker}
\affiliation{Bocconi University, 20136 Milan, Italy}
\email{simon.becker@unibocconi.it}

\author{Cambyse Rouz\'e}
\affiliation{Inria, T\'el\'ecom Paris -- LTCI, Institut Polytechnique de Paris, 91120 Palaiseau, France}
\email{cambyse.rouze@inria.fr}

\author{Robert Salzmann}
\affiliation{Institute for Quantum Information, RWTH Aachen University, Germany} \affiliation{Inria, ENS de Lyon, Université Claude Bernard Lyon 1, LIP, 69342, Lyon cedex 07, France}
\email{robert.salzmann@rwth-aachen.de}

\begin{abstract}
While recent advances have established efficient quantum algorithms for preparing Gibbs states of finite-dimensional systems, comparable complexity results for bosonic and other infinite-dimensional models remain unexplored. We introduce the first general rigorous Gibbs sampling framework for bosonic many-body systems, showing that physically relevant bosonic models admit gapped dissipative generators, enabling efficient preparation of thermal states, provided that the spectral gap scales favourably with the number of modes. Although our results hold for broad classes of models, we illustrate them using Bose--Hubbard Hamiltonians, both within and beyond the mean-field regime. In both cases, we show that the associated dissipative generators maintain a positive spectral gap, thereby implying exponential convergence to the thermal state. For the Gibbs sampler corresponding to the full Bose-Hubbard model, we obtain a spectral-gap lower bound $Ce^{-cn}$, where $n$ is the number of modes. We apply our results to provide a Gibbs-state preparation algorithm on qubit hardware, with runtime polynomial in the number of modes and the inverse spectral gap, and thereby obtain a quantum algorithm to compute thermal properties of the model. This provides the first mathematically controlled route to Gibbs sampling in infinite-dimensional systems, with implications for quantum simulation, thermalization, and many-body complexity, where quantum advantages may arise.
\end{abstract}

\maketitle

\section{Introduction}

\noindent
The simulation of ground and thermal states of many-body quantum systems is widely regarded as one of the most promising routes toward near-term quantum advantage. Recent progress has occurred on several fronts: advances in the complexity-theoretic understanding of quantum many-body systems~\cite{bravyi2021complexity}, the development of quantum algorithms for estimating physically relevant observables~\cite{Anschuetz2025,gilyen2021sub,leng2025sub,basso2024optimizing}, and increasingly sophisticated experimental demonstrations of quantum simulations~\cite{Mi2024}.

Among these directions, the preparation of thermal states has recently seen major progress. A series of works introduced quantum Gibbs samplers based on Lindbladian dynamics~\cite{chen2023quantum,gilyen2024quantum,chen2023efficient,ding2025efficient}, and subsequent results established efficient convergence in a variety of physically relevant regimes~\cite{rouze2411optimal,rouze2024efficient,vsmid2025polynomial,vsmid2025rapid,tong2025fast,ding2025end,hahn2025efficient,bakshi2025dobrushin,bergamaschi2025quantum,zhan2025rapid}. However, in essentially all regimes where quantum algorithms are known to be efficient, closely related classical algorithms can also efficiently approximate thermal expectation values and partition functions~\cite{bakshi2024high,mann2021efficient,helmuth2023efficient,Mann2024,chen2025convergence}. This raises the question of whether genuine quantum advantages can emerge in alternative physical settings.

While most existing results focus on locally finite-dimensional systems, such as spin models or fermionic lattices, bosonic systems have received far less attention in the computational complexity literature. Historically, discrete systems have dominated quantum algorithmic studies because they parallel the classical digital-computation paradigm, while bosonic systems were mainly explored in quantum communication and information theory~\cite{braunstein2005quantum,cerf2007quantum,weedbrook2012gaussian}.

Recently, however, a rigorous complexity-theoretic framework for bosonic quantum systems has begun to emerge~\cite{chabaud2024bosonic,chabaud2025energy}, revealing phenomena that differ substantially from those of finite-dimensional systems. At the same time, bosonic systems arise naturally in many areas of physics, including condensed matter \cite{cazalilla2011one,griffin1996bose,dalfovo1999theory,torma2022superconductivity}, quantum optics \cite{fabre2020modes,brennecke2007cavity,mivehvar2021cavity}, atomic physics and quantum chemistry~\cite{gersch1963quantum,guo2012critical,del2018tensor,kloss2019multiset,macridin2018digital,macridin2018electron,reinhard2019density,sandhoefer2016density,schroder2016simulating,woods2015simulating}. Their dynamics have therefore been extensively studied from the perspective of Hamiltonian simulation~\cite{kuwahara2024effective} and mathematical physics~\cite{schuch2006quantum,nachtergaele2007lieb}. However, existing quantum runtime guarantees for Gibbs sampling in the bosonic setting are largely restricted to Gaussian systems, for which quantum advantage is not generally expected~\cite{bartlett2002efficient}.

From a computational standpoint, bosonic systems present unique challenges. Classical approximation techniques that are highly effective for spin systems often fail in infinite-dimensional settings. For example, semidefinite-programming relaxations, which provide powerful approximation methods for ground-state problems~\cite{Goemans1995,Hastings2022}, do not readily extend to bosonic systems~\cite{Navascus2013}. Likewise, cluster-expansion techniques that yield efficient classical algorithms for partition functions of finite-dimensional systems~\cite{mann2021efficient,helmuth2023efficient,Mann2024} encounter major obstacles in bosonic models due to the unboundedness of the Hamiltonians and observables~\cite{kuwahara25}. Moreover, Gibbs states of bosonic systems can remain entangled at arbitrarily high temperatures \cite{han2025entropic}, in sharp contrast with discrete variables~\cite{bakshi2024high}.

These challenges suggest that bosonic systems may offer natural regimes where quantum algorithms outperform classical approaches. Bosonic systems already play a central role in early proposals for quantum advantage~\cite{Aaronson2011,tillmann2013experimental,spring2013boson,hamilton2017gaussian}, motivating a systematic study of quantum Gibbs sampling in continuous-variable many-body systems.

In this Letter, we initiate such a program. We study a family of efficiently implementable quantum Gibbs samplers for infinite-dimensional systems and establish positive spectral-gap results, and hence exponential convergence to equilibrium. We focus 
on the \textit{Bose-Hubbard model}. This model is central to many-body bosonic physics because of its fundamental theoretical significance \cite{cazalilla2011one,bloch2008many,kollath2007quench,Fisher1989,freericks1996strong} and its direct experimental relevance \cite{jaksch1998cold,greiner2002quantum,stoferle2004transition,bakr2009quantum,fallani2007ultracold,bakr2010probing}.
Our results provide a rigorous framework for thermal-state preparation and for quantum algorithms estimating thermodynamic observables in bosonic many-body systems.

\vspace{-1em}

\section{Quantum Gibbs sampling for continuous-variable systems}

\noindent The primary task considered in this letter is that of the preparation of the Gibbs state $\sigma_\beta(H):=e^{-\beta H}/\Tr(e^{-\beta H})$ of the Hamiltonian $H$ of an infinite dimensional quantum system at inverse temperature $\beta>0$. In order for $\sigma_\beta(H)$ to be well-defined, we further assume the Gibbs hypothesis, namely that $\Tr(e^{-\lambda H})<\infty$ for all $\lambda >0$. To achieve this, we consider a family of dissipative quantum Gibbs samplers \cite{ding2025efficient,gilyen2024quantum} extended in~\cite{BeckerRouzeSalzmannToAppearcmp} to infinite-dimensional systems. In a nutshell, the jumps associated with the generator $\mathcal{L}_{\smash{\sigma_E,\widehat{f},H}}$ of the dissipation formally consist of the following dressing of a family of bare jump operators $\{A^\alpha\}_{\alpha\in\mathcal A}=\{(A^\alpha)^\dagger\}_{\alpha\in\mathcal A}$:
\begin{align}
\label{eq:IntegralJump}
L^\alpha(H)
\!:=\!
\int_{\mathbb R} e^{itH}\! A^\alpha e^{-itH} f(t)\,dt,
\end{align}
where the filter function $f$ is chosen so that the Gibbs state is stationary:
\begin{align}
\mathcal{L}_{\sigma_E,\widehat f,H}(\sigma_\beta(H))=0.
\end{align}
In fact, the so-called KMS condition $\overline{\widehat f(\nu)}\!=\!\widehat f(-\nu)e^{-\beta\nu/2}$ on the Fourier transform $\widehat{f}(\nu)=\int f(t)e^{it\nu}dt$ ensures that the generator is self-adjoint with respect to a well-chosen inner product $\langle.,.\rangle_{\sigma_\beta(H)}$ (see Appendix \ref{sec:CVGS}). The parameter $\sigma_E>0$ controls both the Hamiltonian simulation time required to implement the dissipation on a quantum computer and the spectral gap, $\operatorname{gap}(L_{\smash{\sigma_E,\widehat{f},H}}),$ of the equivalent self-adjoint generator on Hilbert-Schmidt space. When the stationary space is one-dimensional, the latter controls the mixing time of the dynamics when initialized in a well-chosen state $\rho_{\operatorname{ini}}$ of the system with $\rho_{\operatorname{ini}}\le \mathfrak{c}\,\sigma_\beta(H)$:
\begin{align}
t_{\operatorname{mix}}(\varepsilon)&:=\inf\big\{t\ge 0:\,\big\|e^{t\mathcal{L}_{\sigma_E,\widehat{f},H}}(\,\rho_{\operatorname{ini}}\,)-\sigma_\beta(H)\big\|_1\le \varepsilon\big\}\nonumber\\
& \le  \frac{2\log\big({\mathfrak{c}}/{\varepsilon}\big)}{\operatorname{gap}(L_{\sigma_E,\widehat{f},H})}\,,\qquad\text{ given } \varepsilon\in (0,1)\,.\label{mixingtimeintro}
\end{align}
Intuitively, as $\sigma_{E}\to 0$, the amount of resources required to implement the evolution up to time $t_{\operatorname{mix}}(\varepsilon)$ blows up, while the spectral gap decreases as $\sigma_E$ increases. In our companion paper \cite{BeckerRouzeSalzmannToAppearcmp}, we also identified single-mode settings where a bad choice of filter $f$ may lead to a vanishing gap, resulting in infinite mixing times.
% or exponentially increases the runtime of the dynamics' algorithmic implementation \rob{Do we actually see the latter case?}. 
Building on these insights, our main goal is to prove that, in the limit $\sigma_E\to\infty$, the resulting generator $\mathcal{L}_{\smash{\widehat f,H}}$ remains gapped for a broad class of physically relevant bosonic systems, and to develop end-to-end schemes for computing their key physical properties.

In Section \ref{sec.BHmain} below, we detail our findings concerning the positivity of the spectral gap of the Gibbs samplers corresponding to Bose-Hubbard Hamiltonians. For the full Bose-Hubbard model, we prove a quantitative lower bound on the spectral gap of the corresponding Gibbs sampler directly through estimates on the Dirichlet form in particle-number sectors, see Theorem~\ref{thm:BH-gap-main}. For the mean-field model, we use a perturbative comparison with a number-diagonal reference Hamiltonian. We then apply the finite-dimensional approximation scheme from \cite{BeckerRouzeSalzmannToAppearcmp} to obtain qubit-based implementations and preparation costs in terms of the spectral gap, see Section~\ref{sec:EndToEnd}.

Applications of the same framework to harmonically confined Coulomb systems and free-energy estimation are developed in~%
\cite{BeckerRouzeSalzmannSchroedinger}\hskip0pt%
. Relative-entropy decay for Bose--Hubbard samplers with additional quadratic or cubic jumps is studied in~%
\cite{BeckerRouzeSalzmannHypercontractivity2026}\hskip0pt%
, complementing the spectral-gap analysis of the linear-jump dynamics considered here.

\vspace{-0.4cm}

\section{Bose--Hubbard Hamiltonians}\label{sec.BHmain}

\noindent A broad class of examples is provided by the Bose--Hubbard Hamiltonians: on a finite lattice $\Lambda=\llbracket 1,L\rrbracket^D$ with $n=L^D$ modes, its Hamiltonian is
\begin{align*}
H_{\operatorname{BH}}
\!=\!
-J\sum_{\langle i,j\rangle}(a_i^\dagger a_j+\mathrm{h.c.})
\!+\!\frac{U}{2}\sum_i(N_i^2-N_i)
\!-\!\mu\sum_i N_i.
\end{align*}

\noindent Here $a_i$ and $a_i^\dagger$ are bosonic annihilation and creation operators, $N_i=a_i^{\smash{\dagger}} a_i$, $J$ is the hopping amplitude, $U>0$ the on-site repulsion, and $\mu$ the chemical potential. For such Hamiltonians, the ground-state problem restricted to finite-particle sectors is known to be computationally hard~\cite{childs2014bose}, while finite-temperature quantities necessarily involve all particle-number sectors and therefore the full infinite-dimensional Fock space. This infinite-dimensionality creates a serious obstacle for classical algorithms. In particular, existing classical partition-function algorithms typically rely on truncating local boson numbers with a cutoff that grows with system size, leading at best to quasi-polynomial runtime in relevant regimes~\cite{kuwahara25}. This motivates the search for quantum algorithms that avoid such classical truncation barriers.

\vspace{-0.5em}

\subsection{Mean-field regime}
\noindent For the mean-field decoupling of the Bose--Hubbard Hamiltonian, the effective single-site model is
\begin{align*}
H_{\mathrm{MF}}
=
-\mu N+\frac{U}{2}N(N-1)-\bar\psi\,a-\psi\,a^\dagger+|\psi|^2.
\end{align*}
We perform a perturbative spectral analysis in the superfluid order parameter $\psi$, which corresponds to a genuinely unbounded perturbation of the Hamiltonian. For sufficiently small $|\psi|$, the spectrum remains discrete, high-energy eigenvalues stay isolated, and the corresponding eigenvectors remain close to the Fock basis. Comparing the resulting Dirichlet form with that of a diagonal reference model yields a positive spectral gap for the associated self-adjoint Gibbs sampler (see Appendix~\ref{sec:meanfieldBH} for details):
\begin{theorem}
\label{thm:BH-mean-field-main}
The mean-field Bose--Hubbard sampler $L_{\hatfM,H_{\operatorname{MF}}}$ associated to the filter function $\hatfM(\nu) = \exp\big(-{\sqrt{1+(\beta\nu)^2}+ \beta\nu}/{4}\big)$ and bare jumps $\{a,a^\dagger\}$ is gapped for sufficiently small $|\psi|$.
\end{theorem}
\noindent 

\vspace{-0.5em}

\section{An exponential spectral gap for the full Bose-Hubbard Gibbs sampler}
\label{sec:BH-gap-main}

\noindent We now turn to the full interacting Bose-Hubbard model. For that, we first introduce the filter functions
\begin{align}
 \widehat f_\eta(\nu):=\hatfM(\nu)(1+\beta^2\nu^2)^{-\eta/2},
 \label{eq:BH-filter-main}
\end{align}
where $\eta\in[0,1/2]$ and $\hatfM$ is the Metropolis-type filter defined in Theorem~\ref{thm:BH-mean-field-main}. For $\eta=0$, we recover $\hatfM$, while $\eta>0$ introduces polynomial decay at negative Bohr frequencies and preserves the KMS condition. Theorem~\ref{thm:BH-gap-main} gives an exponential lower bound on the spectral gap of the Gibbs samplers associated with this family of filter functions.

\begin{theorem}
\label{thm:BH-gap-main}
Let $H_{\operatorname{BH}}$ be the Bose-Hubbard Hamiltonian from Section~\ref{sec.BHmain} with $U>0$ and $\mu,J\in\R$, on a $D$-dimensional lattice with $n$ modes and open boundaries. For $\beta>0$, $\eta\in[0,1/2]$ and bare jumps $\{a_i,a_i^\dagger\}_{i=1}^n$, we have
\begin{align}
 \operatorname{gap}(L_{\widehat f_\eta,H_{\operatorname{BH}}})\ge Ce^{-cn}>0,
 \label{eq:BH-gap-main}
\end{align}
where $C,c>0$ denote constants depending only on $\beta,U,\mu,J,D$. Moreover, the Gibbs state $\sigma_\beta(H_{\operatorname{BH}})$ is the unique invariant state of the corresponding quantum Markov semigroup, i.e. $\ker\mathcal L_{\widehat f_\eta,H_{\operatorname{BH}}}=\C\sigma_\beta(H_{\operatorname{BH}})$.
\end{theorem}

The bound holds for arbitrary fixed hopping strength $J$, without a perturbative smallness assumption. It gives a general exponential guarantee, and in particular establishes positivity of the gap of the unbounded generator. This worst-case bound need not describe the actual dependence on $n$. Obtaining inverse-polynomial lower bounds or bounds independent of the number of modes, for example for small $|J|$ through perturbative arguments, remains an important direction for efficient Gibbs-state preparation.

\subsection{Proof strategy}

\noindent For finite-dimensional spin systems, an exponential lower bound on the spectral gap is a natural worst-case guarantee for Gibbs samplers with local bare jumps and suitable KMS-symmetric filters \footnote{ For $\beta>0$ and $\eta\in[0,1/2]$, the finite-dimensional Gibbs sampler with Hamiltonian $H$ on $n$ qubits, the $3n$ single-site Pauli bare jumps $A^\alpha$ and the filters $\protect\widehat f_\eta$ from \eqref{eq:BH-filter-main} satisfies $\operatorname{gap}(L_{\protect\widehat f_\eta,H})\ge Cn^{-c}e^{-c\beta W}$, where $W=\max\operatorname{Sp}(H)-\min\operatorname{Sp}(H)$ and $C,c>0$ denote constants independent of $n,\beta,H,\eta$. This includes $\protect\hatfM$ at $\eta=0$ and gives an exponential lower bound in $n$ when $W=\mathcal O(n)$ at fixed $\beta$. To see this, denote $\sigma_\beta=e^{-\beta H}/\Tr(e^{-\beta H})$ and $w_\beta(t)=[\beta\cosh(2\pi t/\beta)]^{-1}$, and define $B^\alpha(t)=\sum_\nu e^{it\nu+\beta\nu/4}\protect\widehat f_\eta(\nu)A^\alpha_\nu$. For $X=\sigma_\beta^{1/4}Y\sigma_\beta^{1/4}$, the Dirichlet form is
\begin{align*}
 \cE_{\protect\widehat f_\eta,H}(X)=\sum_\alpha\int_\R w_\beta(t)
 \left\|\sigma_\beta^{1/4}[B^\alpha(t),Y]\sigma_\beta^{1/4}\right\|_2^2\,dt.
\end{align*}
The multiplier $e^{x/4}\protect\widehat f_\eta(x/\beta)$ is a positive mixture of Gaussians of total weight $e^{-1/4}$, giving $\left\|B^\alpha(t)\right\|_\infty\le1$. Its reciprocal is analytic for $|\operatorname{Im}x|<1$ and grows only exponentially. To recover the bare jumps, we multiply this reciprocal by $(1+e^{-b}\cosh x)^{-1}$, with $b=\beta W+\log(16\sqrt{3n})$. After substituting $x=\beta\nu$, Fourier inversion, a contour shift and truncation give a kernel $k$ supported on $[-T,T]$ such that
\begin{align*}
 \left\|A^\alpha-\int_{-T}^{T}k(t)B^\alpha(t)\,dt\right\|_\infty&\le\frac{1}{4\sqrt{3n}},\\
 K^2:=\int_{-T}^{T}\frac{|k(t)|^2}{w_\beta(t)}\,dt&\le Cn^c e^{c\beta W},
\end{align*}
where $T=\mathcal O(\beta[1+\beta W+\log n])$. Subtracting a scalar from $Y$ leaves the form unchanged, so assume $\Tr Y=0$. The depolarising gap $4$, the triangle inequality and Cauchy-Schwarz give
\begin{align*}
 \sqrt8\left\|Y\right\|_2
 &\le K\left\|\sigma_\beta^{-1}\right\|_\infty^{1/2}\cE_{\protect\widehat f_\eta,H}(X)^{1/2}
 +\tfrac12\left\|Y\right\|_2.
\end{align*}
Absorbing the last term and using $\min_z\left\|X-z\sqrt{\sigma_\beta}\right\|_2^2\le\left\|\sigma_\beta\right\|_\infty\left\|Y\right\|_2^2$ yields
\begin{align*}
 \operatorname{gap}(L_{\protect\widehat f_\eta,H})
 \ge 4K^{-2}e^{-\beta W}\ge Cn^{-c}e^{-c\beta W},
\end{align*}
after renaming $C,c$, since $\left\|\sigma_\beta\right\|_\infty\left\|\sigma_\beta^{-1}\right\|_\infty=e^{\beta W}$}. Such bounds can be obtained by comparing the Gibbs generator with a local depolarising generator at $\beta=0$, whose invariant state is maximally mixed. Such a comparison does not apply directly here, since there is no maximally mixed state on the infinite-dimensional Hilbert space and $H_{\operatorname{BH}}$ is unbounded. We instead use that $[H_{\operatorname{BH}},\Ntot]=0$, with $\Ntot=\sum_iN_i$, so $H_{\operatorname{BH}}$ is block diagonal in the sectors of fixed total particle number. The bare jumps $a_i$ and $a_i^\dagger$ lower and raise this number by one, respectively, allowing us to prove the gap bound through estimates on these blocks.

We denote $\sigma_\beta=\sigma_\beta(H_{\operatorname{BH}})$. The self-adjoint representation $L_{\widehat f_\eta,H_{\operatorname{BH}}}$ is formally related to the Schr\"odinger generator by
\begin{align*}
 &\mathcal L_{\widehat f_\eta,H_{\operatorname{BH}}}\big(\sigma_\beta^{1/4}X\sigma_\beta^{1/4}\big)=\sigma_\beta^{1/4}L_{\widehat f_\eta,H_{\operatorname{BH}}}(X)\sigma_\beta^{1/4},
\end{align*}
as explained in Appendix~\ref{sec:CVGS}. From this and invariance of the Gibbs state, we see that $L_{\widehat f_\eta,H_{\operatorname{BH}}}(\sqrt{\sigma_\beta})=0$. For $X$ supported in finitely many particle sectors, the associated Dirichlet form is
\begin{align}
 \cE_{\widehat f_\eta,H_{\operatorname{BH}}}(X)
 =-\langle X,L_{\widehat f_\eta,H_{\operatorname{BH}}}(X)\rangle_2,
 \label{eq:BH-Dirichlet-main}
\end{align}
where $\langle X,Y\rangle_2=\Tr(X^\dagger Y)$. We prove
\begin{align}
 \cE_{\widehat f_\eta,H_{\operatorname{BH}}}(X)
 \ge Ce^{-cn}\left\|X-\langle\sqrt{\sigma_\beta},X\rangle_2\sqrt{\sigma_\beta}\right\|_2^2
 \label{eq:BH-Poincare-main}
\end{align}
on the form domain. This gives the gap bound and shows that $\ker L_{\widehat f_\eta,H_{\operatorname{BH}}}=\C\sqrt{\sigma_\beta}$, so $\sigma_\beta$ is the unique invariant state of the quantum Markov semigroup generated by $\mathcal L_{\widehat f_\eta,H_{\operatorname{BH}}}$.

To prove \eqref{eq:BH-Poincare-main}, we first restrict the Hamiltonian and bare jumps to the sectors of total particle number at most $R$, suppressing the cutoff dependence in the notation. We then decompose $X$ into blocks $X_{p,q}=\Pi_pX\Pi_q$, where $\Pi_p$ is the spectral projection of $\Ntot$ corresponding to the total particle number $p$. Since annihilation lowers particle number by one, the filtered jumps relate neighbouring blocks. Estimates on these operators give, for $p\ge q\ge1$,
\begin{align*}
 \left\|X_{p,q}\right\|_2
 \le C(n+1)e_{p,q}+K\sqrt{n/p}\,\left\|X_{p-1,q-1}\right\|_2,
\end{align*}
where $C,K>0$ are independent of $n,p,q$. Here $e_{p,q}\ge0$ measures, roughly, how far these blocks are from satisfying the relations imposed by the filtered annihilation operators, which hold exactly for $\sqrt{\sigma_\beta}$. Thus each block is controlled by this deviation and a block with fewer particles. Crucially, $K\sqrt{n/p}<1$ for sufficiently large total particle numbers $p$, so these steps of the recurrence do not amplify the bound on the neighbouring block norm.

For $X$ such that $X_{0,0}=0$, we iterate this bound towards $q=0$, where the blocks are controlled directly. The products of recurrence coefficients are bounded by $(K\sqrt n)^j/\sqrt{j!}$ after $j$ steps. Summing over the blocks and using the corresponding creation estimates for $p<q$ gives
\begin{align}
 \left\|X\right\|_2
 \le\sqrt2 C(n+1)e^{K^2n}\left(\sum_{p,q\ge0}e_{p,q}^2\right)^{1/2},
 \label{eq:BH-block-norm-main}
\end{align}
where we used $\sum_{j\ge0}(K\sqrt n)^j/\sqrt{j!}\le\sqrt2 e^{K^2n}$, which accounts for the exponential dependence on the number of modes.

For general $X$, subtracting a suitable multiple of $\sqrt{\sigma_\beta}$ makes the vacuum block zero without changing these contributions. The same bound therefore controls
\begin{align*}
 \operatorname{dist}\big(X,\C\sqrt{\sigma_\beta}\big)
 =\left\|X-\langle\sqrt{\sigma_\beta},X\rangle_2\sqrt{\sigma_\beta}\right\|_2.
\end{align*}
The comparison estimates in Appendix~\ref{BHexp:section} relate the sum involving the $e_{p,q}^2$ terms in \eqref{eq:BH-block-norm-main} to the Dirichlet form, yielding \eqref{eq:BH-Poincare-main}.

The estimates are uniform in $R$, so the bound passes to the full model by taking $R\to\infty$ and using density in the form norm. The full proof is given in Appendix~\ref{BHexp:section}.

\section{End-to-end runtimes \\for physical properties}
\label{sec:EndToEnd}
\noindent The above findings open the door to rigorous runtime analysis of quantum algorithms based on implementing our constructed Gibbs samplers for the estimation of thermal properties of bosonic quantum systems. In Appendix~\ref{sec:CircuitImplementationBH}, we combine Theorem~\ref{thm:BH-gap-main} with \hskip0pt\relax%
\cite[Theorem~4.23 and Corollary~4.25]{BeckerRouzeSalzmannToAppearcmp}\hskip0pt%
 to show that the dissipative dynamics of the full Bose-Hubbard model, starting from the vacuum, can be simulated with resources polynomial in the evolution time on a qubit-based quantum computer. For that, we use the filter $\widehat f_\eta$ from \eqref{eq:BH-filter-main} for fixed $\eta\in(0,1/2]$. This filter satisfies the decay and regularity assumptions of the circuit implementation result in %
\cite[Theorem~4.23]{BeckerRouzeSalzmannToAppearcmp}\hskip0pt%
, as reviewed in Appendix~\ref{sec:CircuitImplementationBH}. In particular, its time-domain representation is integrable, $f_\eta\in L^1(\R)$, by %
\cite[Lemmas~3.5 and 4.19]{BeckerRouzeSalzmannToAppearcmp}\hskip0pt%
. Using gap antimonotonicity and the mixing-time bound \eqref{mixingtimeintro}, we obtain Theorem~\ref{thm:BH-preparation-main} below.
\begin{theorem}
\label{thm:BH-preparation-main}

Let $\beta,U>0$, $\mu,J\in\R$, $\eps\in(0,1)$, and fix $\eta\in(0,1/2]$. Under the coherent-access assumptions and workspace convention of Theorem~\ref{thm:CircuitGibbsSF}, the Gibbs state $\sigma_\beta(H_{\operatorname{BH}})$ can be prepared within $\eps$-trace distance on a quantum computer with $\mathcal O\left(n\log\left(2+n+\log\left((1+\lambda_2^{-1})/\eps\right)\right)\right)$ system qubits and $\widetilde{\mathcal O}\left(\log^2\left(2+n(1+\lambda_2^{-1})/\eps\right)\right)$ ancilla qubits, with total Hamiltonian simulation time of order
$\widetilde{\mathcal O}\left(\left(1+\frac{1}{\lambda_2}\right)\operatorname{poly}\left(n,\log(1/\eps)\right)\right)$,
where $\lambda_2:=\operatorname{gap}(L_{\widehat f_\eta,H_{\operatorname{BH}}})\ge Ce^{-cn}>0$ by Theorem~\ref{thm:BH-gap-main}.
\end{theorem}

Here and below, the constants depend only on the fixed physical and filter parameters, and $\widetilde{\mathcal O}$ suppresses additional polylogarithmic factors in the inverse gap, $n$ and $1/\eps$.

The above Gibbs-state preparation results enable the estimation of various thermal observables, including local densities, compressibility, and two-point correlation functions. We illustrate this fact with the problem of estimating the free energy of the Bose--Hubbard model at a given inverse temperature $\beta$ (see 
Appendix~\ref{sec:CircuitImplementationBH}.\ref{appendixforfreeenergy} for details), defined by
\begin{align*}
 F(\beta,H_{\operatorname{BH}}):=-\beta^{-1}\log\Tr(e^{-\beta H_{\operatorname{BH}}}).
\end{align*}
For this, we consider the path $H(s):=H_0+s  V $, for $s\in[0,1]$, with $H_0:=\frac{U}{2}\sum_iN_i(N_i-1)-\mu\Ntot$ and $V:=-J\sum_{\langle i,j\rangle}(a_i^\dagger a_j+a_j^\dagger a_i)$,
so that $H(0)=H_0$ and $H(1)=H_{\operatorname{BH}}$. Along this path, $U>0$ stays fixed and the hopping strength is $sJ$, so Theorem~\ref{thm:BH-gap-main} gives a uniform lower bound on the spectral gap of the Gibbs sampler corresponding to $H(s)$, as shown in Appendix~\ref{sec:CircuitImplementationBH}.\ref{appendixforfreeenergy}. A standard use of the fundamental theorem of calculus yields 
\begin{align*}
F(\beta,H_{\operatorname{BH}})=F(\beta,H_{0})
+\int_0^1\Tr\Big(\sigma_\beta(H(s))V\Big) ds,
\end{align*}
where we denote the Gibbs state of $H$ at inverse temperature $\beta$ by $\sigma_\beta(H)$. Therefore, since $H_0$ is a sum of independent single-site Hamiltonians, its free energy is explicitly given by \eqref{eq:BH-reference-free-energy} and can be evaluated efficiently classically and it hence suffices to estimate the average of the hopping observable $V$ in the Gibbs state at different values of $s$ in order to obtain a good estimate of $F(\beta,H_{\operatorname{BH}})$. A finite-dimensional approximation of the sampler and observable, combined with the discretisation of the integral over $[0,1]$ and use of our sampler to prepare the appropriate Gibbs states along the discretised path, leads to the following:

% Keep the theorem together without stretching the preceding displays.
\Needspace{16\baselineskip}
\begin{samepage}
\begin{theorem}
Under the same assumptions as Theorem~\ref{thm:BH-preparation-main}, the  free energy $ F(\beta,H_{\operatorname{BH}})$ can be estimated with accuracy $\eps\in(0,1)$ and probability of failure at most $\delta\in(0,1)$ on a quantum computer with the system-qubit and ancilla counts of Theorem~\ref{thm:BH-preparation-main}, with $\lambda_2$ replaced by $\lambda_2^{\min}$, and total Hamiltonian simulation time of order
\begin{align*}
    \widetilde{ \mathcal{O}}\left(\frac{1+(\lambda^{\min}_{2})^{-1}}{\eps^2}\log\left(2/\delta\right)\operatorname{poly}\left(n,\log(1/\eps)\right)\right),
\end{align*}
where $\lambda^{\min}_{2}:=\inf_{s\in[0,1]}\operatorname{gap}(L_{\widehat f_\eta,H(s)})\ge Ce^{-cn}>0$. Inserting this lower bound gives total Hamiltonian simulation time $\mathcal O\left(e^{cn}\eps^{-2}\log(2/\delta)\operatorname{poly}(n,\log(1/\eps))\right)$.
\end{theorem}
\end{samepage}

\vspace{-1.5em}

\section{Outlook}

\noindent In this Letter, we introduced a general framework for quantum Gibbs sampling in continuous-variable many-body systems and proved positive spectral-gap results for the Gibbs samplers corresponding to the paradigmatic Bose-Hubbard model. These results provide a first rigorous foundation for analyzing the runtime of quantum algorithms that prepare Gibbs states of interacting, infinite-dimensional quantum systems.

We give two independent proofs of spectral-gap control for these Gibbs samplers: a perturbative comparison of Dirichlet forms for the mean-field Bose-Hubbard model and a nonperturbative Dirichlet-form estimate based on a recurrence between particle-number blocks for the full model.

An important next challenge is to improve the exponential lower bound on the spectral gap of the Gibbs sampler corresponding to the full Bose-Hubbard model, for example to an inverse-polynomial bound or a bound independent of the number of modes in a suitable parameter regime, and to obtain quantitative gap bounds for Gibbs samplers corresponding to other continuous-variable many-body systems.

More broadly, our work opens several directions. On the algorithmic side, it yields rigorous convergence guarantees for quantum algorithms aimed at estimating free energies, response coefficients, and other thermal observables in bosonic systems. On the complexity side, it points to continuous-variable models as promising candidates for quantum advantages in thermal-state simulation \cite{kuwahara25}. On the physical side, it applies to experimentally relevant platforms such as optical-lattice bosons. Altogether, these results mark a first step toward a complexity theory of thermal-state preparation for interacting continuous-variable quantum matter.

\onecolumngrid
\clearpage

 \subsection*{Acknowledgement}
\noindent SB acknowledges support from the SNF Grant PZ00P2\_216019. 
CR is supported by France 2030 under the
French National Research Agency award number ''ANR-22-EXES-0013''. RS acknowledges support by the European Research Council (ERC Grant Agreement No.~948139 and ERC Grant AlgoQIP, Agreement No. 851716), from the Excellence Cluster Matter and Light for Quantum Computing (ML4Q-2), from the QuantERA II Programme of the
European Union’s Horizon 2020 research and innovation programme under Grant Agreement No
101017733 (VERIqTAS) as well as the government grant managed by the Agence Nationale de la
Recherche under the Plan France 2030 with the reference ANR-22-PETQ-0007. 

\paragraph*{AI statement}
AI-assisted tools were used only for light language editing in the first version of this article posted on arXiv. In preparing the revised version, we used these tools for a detailed critical review of the original manuscript, which helped us strengthen several statements and streamline the presentation. In particular, the spectral-gap lower bound for the full Bose--Hubbard model in the revised version was obtained with the help of ChatGPT. The authors checked and verified the AI-assisted revisions and are responsible for the mathematical content and the final text.

\appendix

\appendixtableofcontents

\section{Continuous-variable Gibbs sampling}
\label{sec:CVGS}

\subsection{General framework}

\noindent Given a lower bounded Hamiltonian $H\ge -h_0I$ for some $h_0\ge 0$ over a separable Hilbert space $\mathcal{H}$ with spectral decomposition $H=\sum_{E\in\operatorname{Sp}(H)}EP_E$ and a set $\{A^\alpha\}_{\alpha\in\mathcal{A}}$ of bare jumps over $\mathcal{H}$, the projected Davies jump operators are
\begin{align*}
A^\alpha_\nu
:= \sum_{\substack{E,E'\in\operatorname{Sp}(H)\\ E-E'=\nu}}
    P_{E} A^\alpha P_{E'} .
\end{align*}
Throughout, we assume the Gibbs hypothesis $\Tr(e^{-bH})<\infty$ for every $b>0$ and a finite family of closed, densely defined bare jumps that is invariant under taking adjoints, with a common test domain containing the energy eigenvectors, as in \cite[Section~2]{BeckerRouzeSalzmannToAppearcmp}. The rate function $\Upsilon:\mathbb{R}\to\mathbb{R}_+$ determines the jump frequencies and satisfies the KMS condition
$\Upsilon(-\nu)=e^{\beta\nu}\Upsilon(\nu)$, ensuring that the Gibbs state
$\sigma_\beta=\frac{e^{-\beta H}}{\Tr(e^{-\beta H})}$ is stationary. A well-known drawback, already in finite dimensions, is the dependence of $A^\alpha_\nu$ on the generally unknown spectral decomposition of $H$, which complicates both circuit implementations and convergence analyses.

Recent alternatives avoid explicit spectral information \cite{chen2023quantum,gilyen2024quantum,chen2023efficient,ding2025efficient}. In \cite{ding2025efficient}, the jump operators are formally defined by
\begin{align}
L^\alpha
&:= \int e^{itH} A^\alpha e^{-itH}\, f(t)\, dt
 = \sum_{E,E'\in\operatorname{Sp}(H)} \widehat{f}(E{-}E')\, P_E A^\alpha P_{E'}
 = \sum_{\nu\in B(H)}\widehat{f}(\nu)\, A^\alpha_\nu ,
\label{def:Lalpha}
\end{align}
where $f\in L^1(\mathbb{R})$ is a \emph{filter function}, with Fourier transform
$\widehat{f}(\nu)=\int_{\mathbb R} f(t)e^{i\nu t}\,dt$ as in \cite{ding2025efficient} 
which satisfies the symmetry condition
\begin{align}\label{eq:symmetryintro}
\overline{\widehat{f}(\nu)}=\widehat{f}(-\nu)\,e^{-\beta\nu/2}.
\end{align}
The generation theory also permits filters without an $L^1(\mathbb R)$ time-domain representation, including the Metropolis-type filter already considered in \cite{ding2025efficient},
\begin{align}
\label{eq:filterFunction}
    \widehat f_{\mathscr M}(\nu) = \exp\left(-\frac{\sqrt{1+(\beta\nu)^2}+ \beta\nu}{4}\right).
\end{align}
For implementation, we instead use $\widehat f_\eta$ from \eqref{eq:BH-filter-main} with $\eta>0$, whose time-domain representation satisfies $f_\eta\in L^1(\mathbb R)$ by \cite[Lemma~3.5]{BeckerRouzeSalzmannToAppearcmp}. The integral form \eqref{def:Lalpha} allows implementation via oracle access to block encodings of the Hamiltonian evolution and the bare jumps $A^\alpha$, after time discretization. The resulting Lindblad generator is
\begin{align}
\mathcal{L}_{\widehat{f},H}(\rho)
&= -i[B,\rho]
 + \sum_{\alpha\in\mathcal{A}} \Big(
 L^\alpha \rho (L^\alpha)^\dagger
 - \tfrac12 \{(L^\alpha)^\dagger L^\alpha,\rho\}
 \Big),
\end{align}
where the Hermitian operator $B$ is chosen so that $\mathcal{L}_{\widehat{f},H}(\sigma_\beta)=0$.
Explicitly, the operator $B$ is given by 
\begin{equation}
\begin{split}\label{eq:Bintro} \langle E'&|B|E\rangle=\frac{i}{2}\operatorname{tanh}\Big(\tfrac{\beta(E'-E)}{4}\Big)\sum_{\alpha\in\mathcal{A}}\sum_{\substack{\nu_1,\nu_2\in B(H)\\\nu_2-\nu_1=E'-E}}\overline{\widehat{f}(\nu_1)}\widehat{f}(\nu_2)\,\langle E'|(A^\alpha)^\dagger P_{E+\nu_2}A^\alpha |E\rangle. 
\end{split}
\end{equation}
To regularise the drift for implementation with an unbounded Hamiltonian, we work with the following regularization of $\mathcal{L}_{\widehat{f},H}$: for $\sigma_E\in[0,\infty]$, define
\begin{align}
\mathcal{L}_{\sigma_E,\widehat{f},H}(\rho)
:=\sum_{\alpha}\sum_{\nu_1,\nu_2}
e^{-\frac{(\nu_1-\nu_2)^2}{8\sigma_E^2}}
\,\overline{\widehat{f}(\nu_1)}\widehat{f}(\nu_2)
\Big(
-i[B^\alpha_{\nu_1,\nu_2},\rho]
+A^\alpha_{\nu_2}\rho (A^\alpha_{\nu_1})^\dagger
-\tfrac12\{(A^\alpha_{\nu_1})^\dagger A^\alpha_{\nu_2},\rho\}
\Big),
\nonumber
\end{align}
where
\begin{align*}
B^\alpha_{\nu_1\nu_2}
=\frac{i}{2}\tanh\!\big(\beta(\nu_2-\nu_1)/4\big)\,
(A^\alpha_{\nu_1})^\dagger A^\alpha_{\nu_2}.
\end{align*}
At $\sigma_E=0$ and $\sigma_E=\infty$, the Gaussian factor is interpreted as $\mathbf1_{\{\nu_1=\nu_2\}}$ and $1$, respectively. Under the assumptions below, this generator defines a quantum Markov semigroup with invariant state $\sigma_\beta$, with KMS symmetry expressed through the inner product
\begin{align}\label{eq:KMSsymmetry}
\langle X,Y\rangle_{\sigma_\beta}
:=\Tr\!\left(\sigma_\beta^{1/2}X^\dagger\sigma_\beta^{1/2}Y\right).
\end{align}
Uniqueness of the invariant state requires additional assumptions and is proved for the models considered here. Moreover, $\mathcal{L}_{\sigma_E,\widehat{f},H}$ interpolates between
$\mathcal{L}_{\widehat{f},H}$ ($\sigma_E=\infty$) and the Davies generator
$\mathcal{L}_{\mathrm D}$ with rate $\Upsilon=|\widehat f|^2$ ($\sigma_E=0$).
The Gaussian envelope regularizes the drift: for $0<\sigma_E<\infty$ and $f\in L^1(\mathbb R)$, we have the integral representation of \cite[Section~4.2]{BeckerRouzeSalzmannToAppearcmp},
\begin{align}\label{eq:LsigmaEintro}
\mathcal{L}_{\sigma_E,\widehat{f},H}(\rho)
=G_{\sigma_E}\rho+\rho G_{\sigma_E}^\dagger
+\Phi_{\sigma_E,\widehat{f},H}(\rho),
\end{align}
where
\begin{align*}
G_{\sigma_E}
:=-\sum_{\alpha\in\mathcal A}\int_{\mathbb R}
g(t)\,e^{itH}((L^\alpha)^\dagger L^\alpha)e^{-itH}\,dt,
\end{align*}
with
\[
g(t)=\frac{1}{2\pi}\int_{\mathbb R}
\frac{e^{-\nu^2/8\sigma_E^2}}{1+e^{\beta\nu/2}}
e^{-i\nu t}\,d\nu,\qquad
X_s^\alpha:=e^{isH}L^\alpha e^{-isH},
\]
and the CP map
\begin{align*}
\Phi_{\sigma_E,\widehat{f},H}(\rho)
:=\sigma_E\sqrt{\frac{2}{\pi}}
\sum_{\alpha}\int_{\mathbb R}
e^{-2\sigma_E^2 s^2}\,
X_s^\alpha\rho (X_s^\alpha)^\dagger\,ds .
\end{align*}
Hence, for $f\in L^1(\mathbb R)$, this yields a spectral-agnostic generator that exactly fixes the Gibbs state of $H$ at inverse temperature $\beta$.

 In order for the map $\mathcal{L}_{\sigma_E,\widehat{f},H}$ to define the generator of a quantum Markov semigroup, the following conditions were derived in \cite[Sections~2 and~4.1]{BeckerRouzeSalzmannToAppearcmp}. In what follows, we make use of the framework of quantum Sobolev spaces \cite{gondolf2024energy}: given $\delta_1,\delta_2\ge 0$ and denoting $\widetilde{H}:=H+(h_0+1)I$, 
\begin{align*}
D(\mathcal{W}_H^{\delta_1,\delta_2}):=\Big\{\widetilde{H}^{-\delta_1}a \widetilde{H}^{-\delta_2}\,\Big|\,a\in\mathscr{T}_1(\mathcal{H})\Big\}\text{ with norms }
\|x\|_{\mathcal{W}_H^{\delta_1,\delta_2}}:=\|\mathcal{W}_H^{\delta_1,\delta_2}(x)\|_1,
\end{align*}
where $\mathcal{W}^{\delta_1,\delta_2}_H(x):=\widetilde{H}^{\delta_1} x \widetilde{H}^{\delta_2}$ and $\mathscr{T}_1(\mathcal{H})$ denotes the Banach space of trace-class operators on $\mathcal{H}$. 

 \begin{proposition}[\cite{BeckerRouzeSalzmannToAppearcmp}]\label{condpropdefsampler}
Under the Gibbs and bare-jump assumptions above, fix $\beta>0$ and $\sigma_E\in[0,\infty]$. The following condition leads to the well-posedness of the map $\mathcal{L}_{\sigma_E,\widehat{f},H}$ as the generator of a strongly continuous semigroup of quantum channels with domain including $\mathscr{F}:=\operatorname{span}\{|E_i\rangle\langle E_j|\}_{E_i,E_j\in\operatorname{Sp}(H)}$ for some $0\le\gamma\le   \mu$: there exists a constant $C>0$ such that
\begin{align}
\label{eq:BoundBareJumpsWithHam}
\|A^\alpha \widetilde{H}^{-\gamma}\|,\quad \|\widetilde{H}^\gamma\,A^\alpha \widetilde{H}^{-\mu}\|\le C.
\end{align}
Moreover, the function $\widehat{f}:\mathbb{R}\to\mathbb{C}$ satisfies 
\begin{align}
\label{eq:KMSFilterFunction}
\overline{\widehat{f}(\nu)}=\widehat{f}(-\nu)\,e^{-\beta\nu/2}\qquad \forall\nu\in\mathbb{R}.
\end{align}
We also assume there is a constant $C'>0$ such that
\begin{align}
\sup_\nu\,|\widehat{f}(\nu)| \le C'.
\end{align}
The weighted bound follows from \eqref{eq:KMSFilterFunction}, since $e^{\beta\nu/2}|\widehat f(\nu)|=|\widehat f(-\nu)|\le C'$.
 Moreover, whenever $f\in L^1(\mathbb R)$ and $0<\sigma_E<\infty$, the integral formula \eqref{eq:LsigmaEintro} holds for any $\rho\in D(\mathcal{W}_H^{\gamma,\gamma})\cap D(\mathcal{W}_H^{\mu,0})\cap D(\mathcal{W}_H^{0,\mu})$. 
\end{proposition}

\noindent As shown in \cite{BeckerRouzeSalzmannToAppearcmp},
the generators $\mathcal{L}_{\widehat{f},H}$ and $\mathcal{L}_{\sigma_E,\widehat{f},H}$ induce self-adjoint generators $L_{\widehat{f},H}$ and $L_{\sigma_E,\widehat{f},H}$ on the space $\mathscr{T}_2(\mathcal{H})$ of Hilbert-Schmidt operators whose Dirichlet forms have the form core $\mathscr F+\C\sqrt{\sigma_\beta}$ such that, given the embedding
\begin{align*}
\iota_2:\mathscr{T}_2(\mathcal{H})\to\mathscr{T}_1(\mathcal{H})\,,\qquad \iota_2(x)=\sigma_\beta^{\frac{1}{4}}x\sigma_\beta^{\frac{1}{4}},
\end{align*}
the semigroups $e^{t L_{\sigma_E,\widehat{f},H}}$ and $e^{t \mathcal{L}_{\sigma_E,\widehat{f},H}}$ satisfy
\begin{align}
e^{t\mathcal{L}_{\sigma_E,\widehat{f},H}}\circ \iota_2(x)=\iota_2\circ e^{tL_{\sigma_E,\widehat{f},H}}(x).
\end{align}
 For $\sigma_E=\infty$, the latter takes the following form: denoting $\Gamma_\tau(X):=e^{\tau H}Xe^{-\tau H}$,
\[
L_{\widehat f,H}(X)
=
-i(B_+X-XB_-)
+\sum_{\alpha\in\mathcal A}
\Big(L^\alpha_+\,X\,(L^\alpha_{+})^\dagger-\tfrac12 K^\alpha_+X-\tfrac12 XK^\alpha_-\Big),
\qquad X\in\mathscr F.
\]
with 
\[
B:=\frac{i}{2}\sum_{\alpha\in\mathcal A}\sum_{E,E',G\in\operatorname{Sp}(H)}
\tanh\!\Big(\frac{\beta(E'-E)}{4}\Big)\,
\overline{\widehat f(G-E')}\,\widehat f(G-E)\,
P_{E'}(A^\alpha)^\dagger P_GA^\alpha P_E
\]
and
\[
L^\alpha_\pm:=\Gamma_{\pm\beta/4}(L^\alpha),
\qquad
K^\alpha_\pm:=\Gamma_{\pm\beta/4}\big((L^\alpha)^\dagger L^\alpha\big),\qquad B_\pm:=\Gamma_{\pm\beta/4}(B).
\] 
For a more detailed mathematical treatment of the generator $\mathcal{L}_{\sigma_E,\widehat{f},H}$ for unbounded Hamiltonians over separable Hilbert spaces, we refer the interested reader to \cite{BeckerRouzeSalzmannToAppearcmp}.

\subsection{Unperturbed Gaussian dynamics}

\noindent In this section, we list a few simple examples of bosonic systems for which the gap can be readily controlled.

\subsubsection{Single-mode setting}

\noindent We start with the quantum Ornstein–Uhlenbeck semigroup converging to the Gibbs state of the number operator $N:=a^\dagger a$ on $L^2(\mathbb{R})$, with associated creation and annihilation operators $a$ and $a^\dagger$, whose generator is given by
\begin{align}
\label{eq:qOUgen}
\mathcal{L}_{\widehat{f}, N}(\rho)
=&|\widehat{f}(1)|^2\,\Big(a^\dagger \rho a-\frac{1}{2}\{aa^\dagger,\rho\}\Big)
+|\widehat{f}(-1)|^2\,\Big(a\rho a^\dagger-\frac{1}{2}\{a^\dagger a,\rho\}\Big).
\end{align}

\noindent Denoting the birth rate $\nu_+:=|\widehat{f}(1)|^2$ and death rate $\nu_-:=|\widehat{f}(-1)|^2$, the associated self-adjoint generator on $\mathscr{T}_2(\cH)$ takes the form \cite{OU}:
\begin{align}
\label{eq:HSgenerator}
L_{\widehat{f},N}(x)
=&
-\left( \frac{\nu_- + \nu_+}{2}(Nx + xN) + \nu_+ x \right)
+ \sqrt{\nu_+ \nu_-}( a x a^{\dagger} +  a^{\dagger} x a).
\end{align}
While the spectrum of $L_{\widehat{f},N}$ is explicitly known, see e.g.~\cite{OU}, its eigenmodes are rather intricate and not directly amenable to spectral perturbation analysis. To overcome this difficulty, we shall frequently use that
\[
-L_{\widehat{f},N} \;\ge\; L_{\operatorname{LB}},
\]
in the sense of quadratic forms. Indeed, under vectorization the estimate $a_1a_2+a_1^\dagger a_2^\dagger\le N_1+N_2+\id$, obtained from $(a_1-a_2^\dagger)^\dagger(a_1-a_2^\dagger)\ge0$, gives this comparison directly from \eqref{eq:HSgenerator}. Here $L_{\operatorname{LB}}$ is the so-called \textit{ladder-block} operator defined by
\begin{equation}
\label{eq:ladderblock}
L_{\operatorname{LB}}(x)
:= \frac{1}{2}(\sqrt{\nu_-}-\sqrt{\nu_+})^2 (N x + x N) \;-\; \sqrt{\nu_+}(\sqrt{\nu_-}-\sqrt{\nu_+})\, x.
\end{equation}
The operator $L_{\operatorname{LB}}$ admits the spectral decomposition
\begin{equation}
\label{eq:LBdecomp}
L_{\operatorname{LB}} = \sum_{k \ge 0} \kappa_k\, R_k,
\end{equation}
where
\begin{equation}
\label{eq:kappak}
\kappa_k := \frac{1}{2}(\sqrt{\nu_-}-\sqrt{\nu_+})^2\, k \;-\; \sqrt{\nu_+}(\sqrt{\nu_-}-\sqrt{\nu_+}),
\end{equation}
and $R_k$ is the orthogonal projection onto the subspace
\begin{equation}
\label{eq:Rkspace}
\operatorname{span}\{\, |n\rangle\langle m| \;:\; n+m = k \,\},
\end{equation}
where $\ket{n}$ denotes the Fock state with $N\ket{n}=n\ket{n}$. This decomposition provides a simpler spectral structure than that of $L_{\widehat{f},N}$, making it a useful comparison operator for perturbative spectral analysis.

\subsubsection{Aubry-Andr\'e model}

\noindent Next, we consider a finite 2D lattice $\Lambda:=\mathbb Z^2/(L\mathbb Z)^2$ in a constant magnetic field. The Hamiltonian is given by
\[ \begin{split}H  &=\sum_{j \in \Lambda} a_j^{\dagger} a_j + t\sum_{j \in \Lambda} \left(e^{ij_2 \gamma} a_{j+e_1}^{\dagger} a_j+e^{-ij_2 \gamma} a_{j+e_1} a_j^{\dagger} \right) + t\sum_{j \in \Lambda} \left(a_{j+e_2} a_j^{\dagger}+ a_{j+e_2}^{\dagger} a_j\right), 
\end{split}\]
here $\gamma$ is the magnetic flux through one square of the lattice, and we impose the assumption that $\gamma=\frac{2\pi p}{L}$ for some $p \in \mathbb N_0$. The magnetic field is incorporated using the so-called Peierls' substitution as an exponential factor. To diagonalize this Hamiltonian, we use the Fourier transform in $e_1$ direction
\[\begin{split}
a_j &= \frac{1}{\sqrt{L}} \sum_{k_1 \in (2\pi/L) \mathbb Z / (2\pi \mathbb Z)} e^{-i k_1 j_1} c_{j_2,k_1 }, \text{ and conversely }
c_{j_2,k_1 }= \frac{1}{\sqrt{L}} \sum_{j_1 \in \mathbb Z / (L \mathbb Z)} e^{i k_1 j_1}  a_{j}.
\end{split}\] 
We then have the central term of the Hamiltonian
\[\begin{split}
\frac{1}{L} \sum_{j \in \Lambda} \sum_{k_1,k_1'} \left(e^{ij_2 \gamma} e^{ik_1(j_1+1)-ik_1'j_1}c_{j_2,k_1 }^{\dagger}c_{j_2,k_1' } +\text{h.c.} \right)
&=  \sum_{k_1,j_2} \left(e^{ij_2 \gamma} e^{ik_1}c_{j_2,k_1 }^{\dagger} c_{j_2,k_1 } +\text{h.c.} \right) \\
&=  \sum_{k_1,j_2} 2\cos(j_2 \gamma + k_1)c_{j_2,k_1 }^{\dagger}c_{j_2,k_1 }.
\end{split}\]
Performing similar computations for the other terms, we find 
\[\begin{split} H &=\sum_{k_1,j_2} c_{j_2,k_1 }^{\dagger}c_{j_2,k_1 } + \sum_{k_1,j_2} 2t\cos(j_2 \gamma + k_1)c_{j_2,k_1 }^{\dagger} c_{j_2,k_1 } + t\sum_{k_1,j_2} \left(c_{j_2+1,k_1 }c_{j_2,k_1 }^{\dagger}+ c_{j_2+1,k_1 }^{\dagger}c_{j_2,k_1 }\right) \\
&=\sum_{k_1} \left(\sum_{i} h_{i,k_1} c_{i_1,k_1 }^{\dagger}c_{i_2,k_1 } \right) \text{ with }h_{i,k_1} = (1+2t\cos(i_1 \gamma+k_1)) \delta_{i_1,i_2} +t(\delta_{i_1-1,i_2}+ \delta_{i_1+1,i_2}).
\end{split}\]
Diagonalizing $h_{k_1} = (h_{i,k_1})_{i}$, we find $h_{k_1} = \sum_{i_1} \varepsilon_{i_1,k_1} \vert u_{i_1,k_1} \rangle \langle u_{i_1,k_1} \vert,$ we then define the unitary matrix $U_{k_1}=(U_{i,k_1})_i$ with $U_{i,k_1}=u_{i_1,k_1}(i_2).$
We then define $b_{i_1,k_1} :=\sum_{i_2=0}^{L-1} U_{i,k_1} c_{i_2,k_1}$ and $c_{i_2,k_1} = \sum_{i_1=0}^{L-1} \overline{U_{i,k_1}} b_{i_1,k_1}$ and can write the Hamiltonian as 
\[ H = \sum_{i_1,k_1} \varepsilon_{i_1,k_1} b_{i_1,k_1}^{\dagger} b_{i_1,k_1}.\]
Assume $\varepsilon_{i_1,k_1}>0$ for every normal mode, as required by the Gibbs hypothesis (for example, $4|t|<1$ suffices). Combining with the results of \cite{OU}, we find
\begin{theorem}
\label{thm:DFT}
   The associated generator of the Gibbs sampler in $\mathscr{T}_2(L^2(\mathbb{R}^m))$, with $m=L^2$, is a family of qOU generators with coefficients $\nu_{-,i_1,k_1}:=\vert  \widehat{f}(-\varepsilon_{i_1,k_1})\vert^2 $, $\nu_{+,i_1,k_1}:=\vert  \widehat{f}(\varepsilon_{i_1,k_1})\vert^2$, and $\nu_{i_1,k_1}:=\sqrt{\nu_{+,i_1,k_1}\nu_{-,i_1,k_1}},$ given by
    \begin{equation}
    \begin{split}
L(x) =&-\frac{1}{2} \sum_{i_1,k_1}   \bigg(\nu_{-,i_1,k_1} (b_{i_1,k_1}^{\dagger} b_{i_1,k_1} x + x b_{i_1,k_1}^{\dagger} b_{i_1,k_1}) +\nu_{+,i_1,k_1}(x b_{i_1,k_1} b_{i_1,k_1}^{\dagger} + b_{i_1,k_1} b_{i_1,k_1}^{\dagger}x) \bigg) \\
&+\sum_{i_1,k_1} \nu_{i_1,k_1} (b_{i_1,k_1} x b_{i_1,k_1}^{\dagger} + b_{i_1,k_1}^{\dagger} x b_{i_1,k_1}). 
\end{split}
\end{equation}
with spectrum 
\[ \operatorname{Sp}(L) =- \sum_{i_1,k_1}\left\{ n\left(\frac{\nu_{-,i_1,k_1}-\nu_{+,i_1,k_1}}{2}\right) ; n \in \mathbb N_0\right\}.\]
\end{theorem}

\section{Perturbation theory of spectral gaps}
\label{sec:GapPertDirichlet}

In this section, we prove stability of the spectral gap of the Hilbert-Schmidt space generator $L_{\widehat f,H}$ through a comparison of Dirichlet forms. We use this result for the mean-field Bose-Hubbard model in Section~\ref{sec:meanfieldBH}.

\noindent In Lemma~\ref{lem:GapPertDirichlet2}, we prove stability under small perturbations of the spectral gap of the self-adjoint generators $L_{\widehat{f},H}$. 
We consider the derivation operator for $t\in\R$ and $x\in\mathscr{F}\subset \mathscr{T}_2(\mathcal{H})$
\begin{align*}
    \partial^\alpha_t(x) &:= \sum_{\nu\in B(H)} \widehat f(\nu) e^{i\nu t}  e^{\beta \nu/4} \delta^\alpha_\nu(x) = \sum_{\nu\in B(H)} \widehat f(\nu) e^{i\nu t} \left(A^\alpha_\nu x - e^{\beta \nu/2}xA^\alpha_\nu\right) \\
    &= \sum_{E', E\in\spec(H)} \widehat f(E'-E) e^{i(E'-E) t} \left(P_{E'}A^\alpha P_E x - e^{\beta (E'-E)/2}xP_{E'}A^\alpha P_{E}\right).
\end{align*}
With that, we can write the Dirichlet form $\mathcal{E}_{\widehat{f},H}(x)=-\langle x, L_{\widehat{f},H} x\rangle$ as \cite{BeckerRouzeSalzmannToAppearcmp}
\begin{align}
\cE_{\widehat{f},H}(x) = \sum_{\alpha\in\mathcal{A}}\int_{-\infty}^\infty \frac{\langle \partial^\alpha_t(x),\partial^\alpha_t(x)\rangle}{\beta \cosh(2\pi t/\beta)} \,dt.
\end{align}
\noindent We consider below Hamiltonians $H$ and $\widetilde H$ and denote for simplicity the associated objects in the space of Hilbert Schmidt operators by  $\partial^{\alpha}_t$ and $\widetilde \partial^{\alpha}_t$, $\mathcal{E}$ and $\widetilde{\mathcal{E}}$, and $L$ and $\widetilde{L}$ for a same fixed function $\widehat{f}$, respectively.
\begin{lemma}
\label{lem:GapPertDirichlet2}
Assume that the two forms admit a common form core on which their derivations are defined. Let $H$ and $\widetilde H$ be such that the difference $\partial^{\alpha}_t - \widetilde \partial^{\alpha}_t$ of their corresponding derivations is bounded with operator norm satisfying the bound
\begin{align}
 \Delta :=\sum_{\alpha\in\mathcal{A}}\int_{-\infty}^{\infty}\frac{\|\partial^{\alpha}_t - \widetilde \partial^{\alpha}_t\|^2}{\beta\cosh(2\pi t/\beta)}\,dt<\infty.
\end{align}
Then, for $\kappa_1,\kappa_2>0$, we have
\begin{align}
\label{eq:tildeDirichiLowerBound}
\widetilde \cE(x) \ge  
 \tfrac{(1-\frac{\kappa_1}{2})\cE(x) - (\kappa_1^{-1}+\kappa_2^{-1})\frac{\Delta}{2}\|x\|_{2}^2}{ 1+\frac{\kappa_2}{2}}.
\end{align}
Assume further that $\ker L=\C\sqrt{\sigma_\beta}$ and that $L$ has a positive spectral gap that satisfies $\operatorname{gap}(L) >\Delta.$ Then also $\widetilde L$ has a positive spectral gap that satisfies
\begin{align}
\label{eq:tildeGapLowerbound2}
\operatorname{gap}(\widetilde L)  \ge \left(\sqrt{\operatorname{gap}(L)} - \sqrt{\Delta}\right)^{2}.
\end{align}
\end{lemma}

\begin{proof} By the Cauchy Schwarz inequality, we see for $\kappa_1,\kappa_2>0$ that
\begin{align*}
    |\cE(x) - \widetilde\cE(x)|&\le \sum_{\alpha\in \mathcal{A}} \int_{-\infty}^\infty  \frac{\left(\|(\partial^{\alpha}_t - \widetilde \partial^{\alpha}_t)(x)\|_2\|\partial^\alpha_t(x)\|_{2} +\|(\partial^{\alpha}_t - \widetilde \partial^{\alpha}_t)(x)\|_2\|\widetilde \partial^\alpha_t(x)\|_{2}\right)}{\beta \cosh(2\pi t/\beta)}dt\\
    &\le \left(\kappa_1^{-1} +\kappa^{-1}_2\right) \frac{\Delta}{2}\|x\|_{2}^2 +\frac{1}{2}\sum_{\alpha\in \mathcal{A}} \int_{-\infty}^\infty dt \frac{\left(\kappa_1\langle\partial^\alpha_t(x),\partial^\alpha_t(x)\rangle +  \kappa_2 \langle\widetilde\partial^\alpha_t(x),\widetilde\partial^\alpha_t(x)\rangle\right)}{\beta \cosh(2\pi t/\beta)} \\
    &=\left(\kappa_1^{-1} +\kappa^{-1}_2\right) \frac{\Delta}{2} \|x\|_{2}^2+\frac{\kappa_1}{2}\cE(x) + \frac{\kappa_2}{2} \widetilde\cE(x) 
\end{align*}
where we used the standard inequality $ab\le \tfrac{\kappa a^2+\kappa^{-1}b^2}{2}$ which holds for all $a,b\in\R$ and $\kappa>0.$  From this, we immediately see \eqref{eq:tildeDirichiLowerBound}. The bounded integrated difference implies $|\mathcal E(x)^{1/2}-\widetilde{\mathcal E}(x)^{1/2}|\le\sqrt\Delta\left\|x\right\|_2$ on the common core. Hence the form norms are equivalent, their closures have the same domain, and \eqref{eq:tildeDirichiLowerBound} extends to that domain.
To prove \eqref{eq:tildeGapLowerbound2} assume $L$ is gapped with $\operatorname{gap}(L) > \Delta$. By definition, we have 
\begin{align*}
    \cE(x) \ge \operatorname{gap}(L) \|x\|^2_{2} 
\end{align*}
for all $x\in\{\sqrt{\sigma_\beta}\}^\perp.$ Hence, we see by \eqref{eq:tildeDirichiLowerBound} and the max-min theorem for eigenvalues of self-adjoint operators that the spectral gap of $\widetilde L$ is lower bounded as 
\begin{align}
\label{eq:TildeGapLowerKappa}
   \operatorname{gap}(\widetilde L) \ge  
   \left(1+\frac{\kappa_2}{2}\right)^{-1}\left((1-\frac{\kappa_1}{2})\operatorname{gap}(L) - (\kappa_1^{-1}+\kappa_2^{-1})\frac{\Delta}{2}\right).
\end{align}
For $\operatorname{gap}(L) >\Delta>0$, the optimal $\kappa_1$ and $\kappa_2$ are given by
\begin{align*}
\kappa_1^* &= \sqrt{\frac{\,\Delta}{\operatorname{gap}(L)}}\text{ and }
\kappa_2^* = \frac{\sqrt{\Delta}}{\sqrt{\operatorname{gap}(L)} - \sqrt{\Delta}}.
\end{align*}
Inserting these into \eqref{eq:TildeGapLowerKappa} yields \eqref{eq:tildeGapLowerbound2}. For $\Delta=0$, take $\kappa_1,\kappa_2\downarrow0$ in \eqref{eq:TildeGapLowerKappa}.

\end{proof}

\section{Mean-field Bose-Hubbard model}
\label{sec:meanfieldBH}

We first consider the generator corresponding to the mean-field Bose-Hubbard model $H_{\operatorname{MF}}$ obtained by decoupling the hopping term in terms of the superfluid order parameter $\psi = \langle  a_i \rangle\in\mathbb{C}$. By a thorough perturbative spectral analysis, we show that for small enough $|\psi|$, $L_{\smash{\sigma_E,\hatfM,H_{\operatorname{MF}}}}$ is gapped. The full Bose-Hubbard model is treated separately in Section~\ref{BHexp:section}, where we prove a quantitative lower bound on the spectral gap of the corresponding Gibbs sampler.

 More precisely, the mean-field description is obtained by decoupling the hopping term in terms of the superfluid order parameter $\psi = \langle  a_i \rangle$, 
\begin{equation}
 a_i^\dagger  a_j \approx \psi^*  a_j +  a_i^\dagger \psi - \vert \psi \vert^2 ,
\end{equation}
which neglects quadratic fluctuations around $\psi$. This reduces the problem to a sum of effective single-site Hamiltonians
\begin{equation*}
 H_{\text{MF}} = \sum_i \bigg[
 -J z_i \left( \psi  a_i^\dagger + h.c.-\vert\psi \vert^2  \right)
 + \tfrac{U N_i^2}{2} 
 - (\mu+\tfrac{U}{2})  N_i \bigg],
\end{equation*}
where $z_i$ denotes the number of neighbours of site $i$, and  the mean-field phase diagram follows from the self-consistency condition $\psi = \langle  a_i \rangle_{\text{MF}}$. Absorbing $Jz_i$ into the source parameter $\psi$ and changing an additive scalar, which does not affect the Gibbs sampler, we study the single-mode Hamiltonian defined by
\begin{equation}
\label{eq:mfBHHam}
    H:=-\mu N +U \tfrac{N(N-1)}{2} -\overline{\psi} a - \psi a^\dagger  + \vert \psi \vert^2
\end{equation}
for some parameters $\mu\in\mathbb{R}$, $U>0$, and $\psi\in\mathbb{C}$. Furthermore, we consider $\mathcal{A}=\{+,-\}$, bare jumps $A^+ = a^\dagger$, and $A^-=a$, as well as the filter function 
\begin{align}
\label{eq:Function}
    \hatfM(\nu) = \exp\left(-\frac{\sqrt{1+(\beta\nu)^2}+ \beta\nu}{4}\right),\qquad \qquad \beta>0.
\end{align}
 \begin{figure}[h!]
     \centering
     \includegraphics[width=0.3\linewidth]{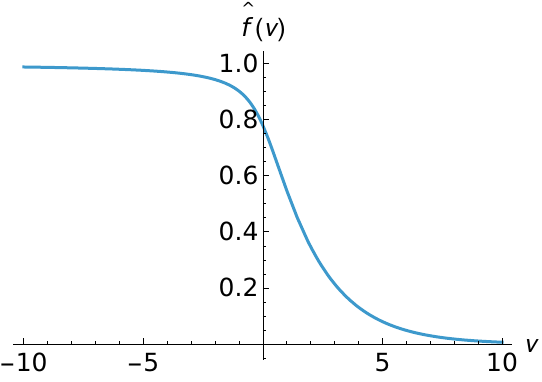}
     \caption{Metropolis-type filter $\hatfM$ in \eqref{eq:Function} for $\beta=1$ smoothly approximating a step function.}
     \label{fig:placeholder}
 \end{figure}

\noindent Let $\mathcal{L}_{\operatorname{MF}}$ be the corresponding generator in the Schrödinger picture as defined in Section \ref{sec:CVGS}.%\ref{propDirichlettoSchrointro} 
In the usual way, we can define the corresponding self-adjoint generator, $L_{\hatfM,H}$, on the space of Hilbert-Schmidt operators through the relation 
\begin{align}\label{eqcLtoL}
e^{t\mathcal L_{\operatorname{MF}}}\circ \iota_2(x)=\iota_2\circ e^{tL_{\hatfM,H}}(x).
\end{align}
 Note that the Gibbs state $\sigma_\beta$ is a well-defined state since $H$ has a discrete spectrum with sufficiently fast growing eigenvalues, as we establish in Lemma~\ref{lem:mfBHEigPert} below. The main result of this section is that $L_{\hatfM,H}$ has a positive spectral gap, as shown in the following theorem.
\begin{theorem}
\label{thm:mfBHGap}
Let $\beta, U>0$ and $\mu\in\R$. Then, for $\psi\in\C$ with $|\psi|$ small enough, we have
$\operatorname{gap}(L_{\hatfM,H})>0.$
\end{theorem}
The same conclusion holds for every $\sigma_E>0$ by \mbox{\cite[Proposition~4.2]{BeckerRouzeSalzmannToAppearcmp}}.

\noindent We prove this perturbatively. In particular, we consider
\[ H_0:=-\mu N +U \tfrac{N(N-1)}{2}+\vert \psi \vert^2 \text{ and } V:= -\overline{\psi} a - \psi a^\dagger\]
such that $H= H_0 + V.$
By definition, we see that $H_0$ is diagonal in the Fock basis, and we denote the eigenvalue corresponding to eigenvector $\ket{n}$ by $E^{(0)}_n$. 
Explicitly, the eigenvalues of $H_0$ satisfy
\begin{equation}\label{eq:H0Eigenvalues}
 E^{(0)}_n =-\mu n+\frac{U}{2}n(n-1)+|\psi|^2,\text{ and }E^{(0)}_{n+1}-E^{(0)}_{n} = -\mu + Un , 
\end{equation}
which in particular shows that  the sequence $E^{(0)}_{n}$ is non-decreasing for  $n\ge \mu/U.$
Throughout this section, particle indices belong to $\N_0$. We denote $n_0:=\max\{0,\lceil 3+4\mu/U\rceil\}$, so that $n\ge 3+4\mu/U$ is equivalent to $n\ge n_0$. Terms containing $\sqrt{0}\ket{-1}$ or $\sqrt{0}\ket{E_{-1}}$ are understood to be zero. 

In the following lemma, we show some bounds on the eigenenergies of $H_0$, in particular that the large energies are far separated from the rest of the spectrum. In Lemma~\ref{lem:mfBHEigPert}, we use this to show that $H$ also has a discrete spectrum  $\spec(H)=\left(E_n\right)_{n\in\N_0}$ satisfying similar bounds as the unperturbed eigenvalues. In Lemma~\ref{lem:mfBHEigVec}, we then use this to derive several perturbative bounds on the eigenvectors of the mean field Bose Hubbard Hamiltonian $H.$ We then combine all of these results with the perturbative result on the spectral gap contained in Lemma~\ref{lem:GapPertDirichlet2} to prove Theorem~\ref{thm:mfBHGap} at the end of this section.
\begin{lemma}[Eigenvalue bounds for $H_0$]
\label{lem:E_n0}
For $n\ge 3+\tfrac{4\mu}{U}$ we have
\begin{equation} 
\label{eq:E_n0SeperationLargen}
|E^{(0)}_m-E^{(0)}_n| \ge \frac{U}{4} (m+n+1) \text{ for all }m\neq n.
\end{equation}

On the other hand, the low eigenenergies of $H_0$ are given in clusters of $E^{(0)}_n$ which are individually close to each other and, furthermore, separated from the rest of the spectrum. More precisely, for $\delta\in [0,\frac{U}{16})$ and $n<3+\tfrac{4\mu}{U}$ we consider $\mathcal{C}_n\subset \N_0$ to be the smallest set such that $n\in \mathcal{C}_n$ and 
\begin{align} 
\label{eq:C_nDefOutside}
|E^{(0)}_{m_1}-E^{(0)}_{m_2}| > 4\delta (m_1+m_2+1)
\end{align}
for all $m_1\in\mathcal{C}_n$ and $m_2\in\N_0\setminus \mathcal{C}_n.$ These sets $\mathcal{C}_n$ define equivalence classes and in particular satisfy $\mathcal{C}_m=\mathcal{C}_n$ if $m\in\mathcal{C}_n.$
Furthermore, we have $\mathcal{C}_n \subseteq\{0,\cdots ,3+\floor{\tfrac{4\mu}{U}}\}$ and
\begin{align}
\label{eq:C_nDefInside}
\max_{m_1,m_2\in\mathcal{C}_n}|E^{(0)}_{m_1}-E^{(0)}_{m_2}|\le 128\,\delta\left(1+\tfrac{\mu}{U}\right)^2.
\end{align} 
\end{lemma}
\begin{proof}
Note that for $n,m\in\N_0$, we have \begin{align}E^{(0)}_m-E^{(0)}_n= (m-n)(\tfrac{U}{2}(m+n-1)-\mu).
\end{align}
From this, we see that for $n\ge 3+\tfrac{4\mu}{U}$ and $m\in\N_0$, 
\begin{align*}
 |E^{(0)}_m - E^{(0)}_n| &=|m-n|\!\left|\frac{U}{2}(m+n-1) - \mu \right|\ge \frac{U}{4} |m-n|(m+n+1), 
\end{align*}
which, in particular, shows \eqref{eq:E_n0SeperationLargen}.

For $\delta\in[0,\frac{U}{16})$ and $n<3+\tfrac{4\mu}{U}$, we construct
$\mathcal{C}_n\subseteq \N_0$ satisfying \eqref{eq:C_nDefOutside} and  \eqref{eq:C_nDefInside} in the following iterative way:
We start with the set $\{n\}$ and then attach all $m_2\in\N_0$ such that
\begin{align} 
\label{eq:C_nDef}
\left|E^{(0)}_{m_1}-E^{(0)}_{m_2}\right| \le  4\delta (m_1+m_2+1)
\end{align} 
with $m_1\equiv n,$ resulting in a larger set. Then we iterate this procedure by attaching all $m_2\in\N_0$ such that there exists $m_1$ in the current set satisfying \eqref{eq:C_nDef}. We call the set that we obtain when this iterative procedure terminates $\mathcal{C}_n.$ From this iterative procedure, it is also clear that the $\mathcal{C}_n$ define equivalence classes.

By construction, the set $\mathcal{C}_n$ satisfies \eqref{eq:C_nDefOutside}.
 Further, note that for $n<3+\tfrac{4\mu}{U}$ we have $\mathcal{C}_n \subseteq \{0,\cdots,3+\floor{\tfrac{4\mu}{U}}\}$; as for $m_2\ge 3+\tfrac{4\mu}{U}$ with $m_2\neq n \equiv m_1$, the inequality \eqref{eq:C_nDef} cannot hold true by using $\delta<\frac{U}{16}$ and \eqref{eq:E_n0SeperationLargen}. Using this, we can pick, by the construction of $\mathcal{C}_n$, for $m_1,m_2\in \mathcal{C}_n$, a chain of mutually different elements  $k_1,\cdots,k_K\in\mathcal{C}_n$ such that $k_1=m_1$, $k_K=m_2$, and  
\begin{align*} 
\left|E^{(0)}_{k_i}-E^{(0)}_{k_{i+1}}\right| &\le 4\delta (k_i+k_{i+1}+1) \le 4\delta \left(7+\frac{8\mu}{U}\right)\le 32\delta\left(1+\frac{\mu}{U}\right).
\end{align*} 
Since $\mathcal{C}_n\subseteq \{0,\cdots,3+\floor{\tfrac{4\mu}{U}}\}$, there can only exist $4+\floor{\tfrac{4\mu}{U}}$ such elements, i.e., $K\le 4+\floor{\tfrac{4\mu}{U}}$ such that \eqref{eq:C_nDefInside} holds.

\end{proof}
\begin{lemma}[Eigenvalue perturbation theory]
\label{lem:mfBHEigPert}
Let $H$ be the mean-field Bose-Hubbard Hamiltonian defined in \eqref{eq:mfBHHam} with $U>0.$ Then $H$ has a discrete spectrum consisting of eigenvalues $\left(E_n\right)_{n\in\N_0},$ which, for $|\psi|<\frac{U}{16}$, satisfy the following:

For $n\in\N_0$ large enough, i.e., precisely
\begin{equation}
\label{eq:Large_n_mfBH}
     n\ge 3 + \frac{4\mu}{U}\quad,
\end{equation}
 we have that $H$ has a unique eigenvalue, $E_n$, satisfying, for any $r\in(2|\psi|,\frac{U\sqrt{n+1}}{8}],$
\begin{align}\label{eq.boundperturb}
    \left|E_n-E^{(0)}_n\right|&< r\sqrt{n+1}.
\end{align}
Hence, for any $\beta>0$, the Gibbs state corresponding to $H$ satisfies the Gibbs hypothesis for $|\psi|<\frac{U}{16}$.
Furthermore, this $E_n$ has multiplicity $1$ and satisfies
\begin{align}
\label{eq:ENEm0lowerbound}
     \left|E_n-E^{(0)}_m\right|&\ge \frac{U}{8}\left(n+m+1\right)
\end{align}
for all $m\neq n.$

   For $n<3+\frac{4\mu}{U}$, $\delta\in[|\psi|,\frac{U}{16})$, and $\mathcal{C}_n\subseteq\left\{0,\cdots,3+\floor{\frac{4\mu}{U}}\right\}$ defined in Lemma~\ref{lem:E_n0}, we have $n\in\mathcal{C}_n$, 
\begin{align}  
\label{eq:EnCnCloseness}
\max_{m_1,m_2\in\mathcal{C}_n}\left|E_{m_1}-E_{m_2}\right| \le 144\,\delta\max\left\{\left(1+\frac{\mu}{U}\right)^2,1\right\}
\end{align}
and \begin{align}
\label{eq:Enlowlowerbound}
\min_{m_1\in\mathcal{C}_n}\left|E_{m_1} -E^{(0)}_{m_2}\right|>2\delta\left(m_2+1\right)
\end{align}
for all $m_2\in\N_0\setminus\mathcal{C}_n.$

Lastly, we have that
\begin{align}
\label{eq:UniformDeltaEBound}
    \Delta_E = \sup_{|\psi|\in I}\max_{0\le j,m\le 2n_0}|E_j-E_m|<\infty
\end{align}
for all closed intervals $I\subseteq [0,\frac{U}{16}).$

\end{lemma}
\begin{proof}
The perturbation $V$ is infinitesimally $N^2$-bounded, since for $\eps>0$ and $\xi\in D(N^2)$ we have
\begin{align*}
 \left\|V\xi\right\|\le 2|\psi|\left\|(N+1)^{1/2}\xi\right\|
 \le\eps\left\|N^2\xi\right\|+C_\eps\left\|\xi\right\|,
\end{align*}
where $C_\eps\ge0$ denotes a constant depending on $\eps$ and $|\psi|$. Hence $H$ is self-adjoint on $D(N^2)$ with equivalent graph norms and compact resolvent, since $D(N^2)$ embeds compactly into $\mathcal H$. The same estimate as a quadratic-form bound gives $H\ge (U/4)N^2-C$ for a constant $C\ge0$, so the min-max principle gives the Gibbs hypothesis.
We next derive a resolvent representation. For this, we use that for $\lambda\in \C\setminus\spec(H_0)$
we have for $m\in\N$
\[(H_0-\lambda)^{-1}a\ket{m}=\frac{\sqrt{m}}{E^{(0)}_{m-1}-\lambda}\,\ket{m-1}\] and \[(H_0-\lambda)^{-1}a^\dagger\ket{m}=\frac{\sqrt{m+1}}{E^{(0)}_{m+1}-\lambda}\,\ket{m+1},\] with the second identity also valid for $m=0$ and $a\ket{0}=0$, and therefore
\begin{align*}
  \left\|(H_0-\lambda)^{-1}V \right\| &\le|\psi|\left( \sup_{m\ge 1} \frac{\sqrt{m}}{|E^{(0)}_{m-1}-\lambda|}+\sup_{m\ge 0} \frac{\sqrt{m+1}}{|E^{(0)}_{m+1}-\lambda|}\right)\le 2|\psi|\sup_{m\ge 0} \frac{\sqrt{m+1}}{|E^{(0)}_{m}-\lambda|} =: |\psi|M(\lambda).
\end{align*}
Hence, for $\lambda\in \C\setminus\spec(H_0)$ such that \begin{align}
\label{eq:SmallM(lambda)}
|\psi|M(\lambda)<1,
\end{align}
we have, by 
a standard von Neumann series argument, that $\lambda\in\C\setminus \spec(H)$ with resolvent given by the convergent series
\begin{align}
\label{eq:VonNeumannResolvent}
    (H-\lambda)^{-1} =\sum_{k=0}^{\infty} (-(H_0-\lambda)^{-1}V)^k (H_0-\lambda)^{-1}.
\end{align}
Below, we show that $\lambda$ satisfying \eqref{eq:SmallM(lambda)} always exists, which in particular implies by \eqref{eq:VonNeumannResolvent}
 that $(H-\lambda)^{-1}$ is compact since $(H_0-\lambda)^{-1}$ is compact. This already shows that the spectrum of $H$ is discrete, i.e. consists of eigenvalues of finite multiplicity.

 We proceed to locate the eigenvalues of $H:$ For $n\in\N_0$ large enough, i.e. satisfying \eqref{eq:Large_n_mfBH}, let $\Gamma_n\subset \C$ be a circle in the complex plane with center of origin $E^{(0)}_n$ and radius $r\sqrt{n+1}$ for some $r \in (2|\psi|,\frac{U\sqrt{n+1}}{8}].$ By \eqref{eq:E_n0SeperationLargen} of Lemma~\ref{lem:E_n0} we have for $m\neq n$
 \begin{align}
 \label{eq:E0nmlowerbound}
     \inf_{\lambda\in\Gamma_n}|E^{(0)}_m - \lambda| &\ge |E^{(0)}_m-E^{(0)}_n| -r\sqrt{n+1}\ge  \frac{U}{4}(m+n+1) - r\sqrt{n+1} \ge \frac{U}{8}(m+n+1).
 \end{align}
 Therefore, we have 
 \begin{align}
\label{eq:VonNeumannConvergence}
|\psi|\sup_{\lambda\in\Gamma_n}M(\lambda)
&\le 2|\psi|\sup_{\lambda\in\Gamma_n}\max\!\left\{
\frac{1}{r},\;
\sup_{m \ne n} \frac{\sqrt{m+1}}{\left|E^{(0)}_m -\lambda\right|}
\right\}
\\
&\le 2|\psi|\max\!\left\{
\frac{1}{r},\;
\frac{8}{U\sqrt{n+1}}
\right\}= \frac{2|\psi|}{r} <1
\end{align}
which by \eqref{eq:VonNeumannResolvent} gives that $(H-\lambda)^{-1}$ exists for all $\lambda\in\Gamma_n.$

 For $s\in[0,1]$, we define $H_s:=H_0+sV$ with $H_0$ fixed. By holomorphic functional calculus, the spectral projection of $H_s$  corresponding to spectral values in the interval $(E^{(0)}_n-r\sqrt{n+1},E^{(0)}_n+r\sqrt{n+1})$  is given by \begin{align*}
     P_n(s)= \frac{1}{2\pi i }\oint_{\Gamma_n} (\lambda-H_s)^{-1}d\lambda.
 \end{align*}  
 By \eqref{eq:VonNeumannResolvent} and \eqref{eq:VonNeumannConvergence}, we see that $s\mapsto P_n(s)$ is continuous (in fact, analytic) on the interval $[0,1]$, and therefore 
 \begin{align}
 \label{eq:rankEquality}
     \operatorname{rank}(P_n(s)) =\operatorname{rank}(P_n(0)) =1  
 \end{align}
 for all $s\in[0,1].$ This in particular implies that $H=H_1$ has a unique eigenvalue, $E_n,$ in the interval $(E^{(0)}_n-r\sqrt{n+1},E^{(0)}_n+r\sqrt{n+1})$ which furthermore has multiplicity 1. Further, using the same argument as for \eqref{eq:E0nmlowerbound} we see \eqref{eq:ENEm0lowerbound}. 

 Similarly, for $n<3+4\mu/U$ and $\delta\in[|\psi|,U/16)$, we consider the cluster $\mathcal C_n$ from Lemma~\ref{lem:E_n0}. The case $\delta=0$ is immediate since then $H=H_0$. For $\delta>0$, we define the closed discs $D_m:=\{z\in\C:|z-E_m^{(0)}|\le2\delta(m+1)\}$. By \eqref{eq:C_nDefOutside}, discs belonging to different clusters are disjoint, since
 \begin{align*}
 |E_m^{(0)}-E_j^{(0)}|>4\delta(m+j+1)\ge2\delta(m+1)+2\delta(j+1)
 \end{align*}
 for $m\in\mathcal C_n$ and $j\notin\mathcal C_n$. We choose $\Gamma'_n$ as the positively oriented boundary of a sufficiently small neighbourhood of $\bigcup_{m\in\mathcal C_n}D_m$, with one boundary component for each connected component if needed. It encloses exactly the eigenvalues indexed by $\mathcal C_n$ and satisfies
 \begin{align}
 \label{eq:ClusterEncoiondition}
 \inf_{\lambda\in\Gamma'_n}|E_m^{(0)}-\lambda|>2\delta(m+1)
 \end{align}
 for every $m\in\N_0$. Thus \eqref{eq:VonNeumannResolvent} applies uniformly on $\Gamma'_n$ to every $H_s$, and the rank argument in \eqref{eq:rankEquality} gives exactly $|\mathcal C_n|$ enclosed eigenvalues. Moreover, \eqref{eq:SmallM(lambda)} holds outside $\bigcup_{m\ge0}D_m$, so these eigenvalues lie in $\bigcup_{m\in\mathcal C_n}D_m$. The separation from every external disc proves \eqref{eq:Enlowlowerbound}. Using \eqref{eq:C_nDefInside}, we also obtain
 \begin{align*}
 \max_{m_1,m_2\in\mathcal C_n}|E_{m_1}-E_{m_2}|
 &\le \max_{m_1,m_2\in\mathcal C_n}|E^{(0)}_{m_1}-E^{(0)}_{m_2}|+4\delta\max_{m\in\mathcal C_n}(m+1)\\
 &\le128\delta(1+\mu/U)^2+16\delta(1+\mu/U)\\
 &\le144\delta\max\{(1+\mu/U)^2,1\},
 \end{align*}
 which proves \eqref{eq:EnCnCloseness}. The same resolvent exclusion and rank argument, together with the isolated high-energy circles, accounts for the whole spectrum.

Lastly, choose $\delta_I\in(\sup I,U/16)$. The preceding contours can be chosen uniformly for $|\psi|\in I$ after subtracting the common scalar $|\psi|^2$. Since only finitely many indices are involved, this gives \eqref{eq:UniformDeltaEBound}.
\end{proof}

For completeness, the generation hypotheses recalled in Section~\ref{sec:CVGS} hold here. For a sufficiently large $c\ge0$, we denote $\widetilde H:=H+c\ge1$. The quadratic-form and graph-norm equivalences give $D(\widetilde H^{1/2})=D(N)$ and $D(\widetilde H)=D(N^2)$, and hence
\begin{align*}
 \left\|A^\alpha\widetilde H^{-1/2}\right\|,
 \left\|\widetilde H^{1/2}A^\alpha\widetilde H^{-1}\right\|\le C,
\end{align*}
where $C\ge0$ denotes a constant and $\alpha\in\{+,-\}$. These are the jump bounds in \eqref{eq:BoundBareJumpsWithHam} with moment exponents $1/2$ and $1$. The filter $\hatfM$ is bounded and satisfies the KMS condition.

By the previous lemma, we know that $H$ has a discrete spectrum $(E_n)_{n\in\N_0}.$ Further, for $n$ large enough, i.e., \eqref{eq:Large_n_mfBH}, we have seen that the $E_n$ are non-degenerate and separated from each other. For small $n$, the $E_n$ can, in principle, be degenerate. To take this into account, we consider $\spec(H)= (E_n)_{n\in\N_0}$ to be the spectrum of $H$, including multiplicities; i.e., $E_n$ appearing multiple times if it is degenerate. We then fix some eigenbasis of $H$, which we denote by $(\ket{E_n})_{n\in\N_0}$ by slight abuse of notation (e.g.~$|E_n\rangle\ne |E_{n+1}\rangle$, even if the energies may coincide $E_n=E_{n+1}$). We use the convention that $c_n := \langle n|E_n\rangle \ge 0$ by possibly changing the phase factor of $\ket{E_n}.$ In the following, we study the state $\ket{E_n}$ as a perturbation of the corresponding Fock basis vector $\ket{n}$. 

\begin{lemma}[Eigenvector perturbation theory]
\label{lem:mfBHEigVec}
Let $\mu\in\R$, $\psi\in\C$, and $U>0$ be such that $\frac{16 |\psi|}{U} <1.$ Then for $n\ge 3+\frac{4\mu}{U}$ and using the convention $c_n = \bra{n} E_n\rangle\in [0,1]$, we have that 
\begin{align}
\label{eq:EnvsFocknLarge}
    \left\|\ket{E_n}-\ket{n}\right\|  \le \frac{32|\psi|}{U\sqrt{n+1}}.
\end{align}
Furthermore, taking $\frac{16|\psi|}{U}\le 1-\kappa$ for some $\kappa>0$, we have for sequences $\left(u_m\right)_{m\in\N_0}\subseteq \C$ such that $|u_m|\le 1$ for all $m$ that
\begin{align}
\label{eq:EvenWeirdOpAreBounded}   
&\left\|\sum_{m\ge n_0+1} \!\!u_m\sqrt{m}\left(\ket{E_{m-1}}\!\bra{E_m} \!- \!\ket{m-1}\!\bra{m}\right)\right\| \le C_\kappa \frac{|\psi|}{U},\\ \label{eq:EvenWeirdOpAreBoundedCreation}
&\left\|\sum_{m\ge 3+\frac{4\mu}{U}}\!\!\!\!\!\!u_m\sqrt{m+1}\left(\ket{E_{m+1}}\!\bra{E_m}\! -\! \ket{m+1}\!\bra{m}\right)\right\| \le C_\kappa \frac{|\psi|}{U},
\end{align}
for some $C_\kappa\ge 0$ depending only on $\kappa.$

For $n<3+\frac{4\mu}{U}$, $\delta\in(|\psi|,\frac{U}{16})$ and $\mathcal{C}_n\subseteq\left\{0,\cdots,3+\floor{\frac{4\mu}{U}}\right\}$ defined in Lemma~\ref{lem:E_n0}, the projections $P_{\mathcal{C}_n}=\sum_{m\in\mathcal{C}_n}\kb{E_m}$ and $P^{(0)}_{\mathcal{C}_n}=\sum_{m\in\mathcal{C}_n}\kb{m}$ satisfy
\begin{align}
\label{eq:P_CnCloseness}
   \left\|P_{\mathcal{C}_n}-P^{(0)}_{\mathcal{C}_n}\right\| \le \frac{24|\psi|}{\delta}\left(1+\frac{\mu}{U}\right), \left\|a\left(P_{\mathcal{C}_n}-P^{(0)}_{\mathcal{C}_n}\right)\right\|  &\le \nn\frac{48|\psi|}{\delta}\left(1+\frac{\mu}{U}\right)^{3/2},\\\left\|a^\dagger\left(P_{\mathcal{C}_n}-P^{(0)}_{\mathcal{C}_n}\right)\right\|  &\le \frac{48|\psi|}{\delta}\left(1+\frac{\mu}{U}\right)^{3/2}.
\end{align}

Furthermore, for $n,m\ge 3+\frac{4\mu}{U}$ and $\frac{16|\psi|}{U} \le 1-\kappa$ for some $\kappa\in(0,1)$, we have
\begin{align}
\label{eq:mfBHMatrixDecayAnn}
&|\bra{ E_m}a\ket{E_n} - \sqrt{n}\delta_{m+1,n}|\le  C_\kappa \frac{|\psi|}{U}\left(1-\frac{\kappa}{2}\right)^{|m-n|},\\\label{eq:mfBHMatrixDecayCrea}&|\bra{ E_m}a^\dagger\ket{ E_n} - \sqrt{n+1}\delta_{m,n+1}| \le  C_\kappa \frac{|\psi|}{U} \left(1-\frac{\kappa}{2}\right)^{|m-n|}
\end{align}
for some $C_\kappa\ge 0$ depending only on $\kappa.$

Lastly, for $n<3+\frac{4\mu}{U}$, $m\ge 3+\frac{4\mu}{U}$, $\psi, \delta$, and $U$ such that $\frac{|\psi|}{\delta}\le 1-\kappa$ for some $\kappa\in(0,1)$, we have 
\begin{equation}
\begin{split}
\label{eq:nLowvsmLargeAnn}
  &\left\|P_{\mathcal{C}_n}a\ket{E_m} - \sqrt{m}P^{(0)}_{\mathcal{C}_n}\ket{m-1}\right\| \le C_\kappa\left(\left(1+\frac{\mu}{U}\right)^{3/2}+1\right) \,\frac{|\psi|}{\delta}\left(1-\frac{\kappa}{2}\right)^{\operatorname{dist}(\mathcal{C}_n,m)}\\
&\left\|P_{\mathcal{C}_n}a^{\dagger}\ket{E_m}\right\|\le C_\kappa\left(1+\frac{\mu}{U}\right) \,\frac{|\psi|}{\delta}\left(1-\frac{\kappa}{2}\right)^{\operatorname{dist}(\mathcal{C}_n,m)}
\end{split}
\end{equation}
for $\operatorname{dist}(\mathcal{C}_n,m) := \min_{k\in\mathcal{C}_n}|k-m|$ and some $C_\kappa\ge 0$ only depending on $\kappa.$
\end{lemma}
\begin{proof}
As established in Lemma~\ref{lem:mfBHEigPert}, the mean-field Bose-Hubbard Hamiltonian, $H=H_0+V$, has a discrete spectrum with eigenvalues denoted by $(E_n)_{n\in\N_0}.$ For $n\ge 3 + \frac{4\mu}{U}$ and $Q_n=\1 -\kb{n}$, we define
\begin{align*}
    R_n= Q_n(E_n-H_0)^{-1}Q_n V,
\end{align*}
where we interpret $Q_n(E_n-H_0)^{-1}Q_n$ as the resolvent of $H_0$ in the space $Q_n\mathcal{H}$, which is well-defined by \eqref{eq:ENEm0lowerbound} of Lemma~\ref{lem:mfBHEigPert}.
Since, given $V=-\overline{\psi}a-\psi a^\dagger$, for $m\in\N$
\begin{equation}
\label{eq:RnOnFock}
   \nn R_n \vert m\rangle = -\psi(E_n-E^{(0)}_{m+1})^{-1}\delta_{m+1 \neq n}\sqrt{m+1}\vert m+1 \rangle -\overline{\psi} (E_n-E^{(0)}_{m-1})^{-1}\delta_{m-1 \neq n}\sqrt{m}\vert m-1 \rangle 
\end{equation}
 we see by \eqref{eq:ENEm0lowerbound} that, by the Schur test,
\begin{align}
\label{eq:Rnbound}
  \|R_n\| \le \frac{16|\psi|}{U\sqrt{n+1}},\ \|aR_n\| \le \frac{16|\psi|}{U}, \text{ and }\|a^\dagger R_n\| \le \frac{16|\psi|}{U}.
\end{align}
    From the eigenvalue equation $H \ket{E_n} =E_n \ket{E_n}$, we find that $(E_n-H_0)\ket{E_n} = V\ket{E_n}$ and 
\begin{align*}
    Q_n \ket{E_n} = R_n\ket{E_n}. 
\end{align*}
Therefore
\begin{equation}
\label{eq:EnNLarge}
\ket{E_n} =\ket{n}\!\bra{ n}\ket{E_n} +Q_n\ket{E_n}= c_n\ket{n} +R_n \ket{E_n} 
\end{equation}
with $c_n = \langle n|E_n\rangle\in[0,1].$ From this and \eqref{eq:Rnbound}, we already see that for large $n$, the $\ket{E_n}$ essentially becomes $\ket{n}$ as 
\begin{equation}
\label{eq:cnLowerBound}
    c_n \ge 1 - \frac{16|\psi|}{U\sqrt{n+1}}.
\end{equation}
In particular, using the triangle inequality with \eqref{eq:EnNLarge}, together with \eqref{eq:cnLowerBound} and \eqref{eq:Rnbound}, we get \eqref{eq:EnvsFocknLarge}:
\begin{align*}
    \left\|\ket{E_n}-\ket{n}\right\| \le (1-c_n) + \|R_n\| \le \frac{32|\psi|}{U\sqrt{n+1}}.
\end{align*}

Using that, by \eqref{eq:Rnbound}, we see $\|R_n\|<1,$  
\begin{equation}
\label{eq:EnVonNeumannFock}
\ket{E_n} = c_n \left(\1 -R_n\right)^{-1}\ket{n}  =c_n \sum_{j=0}^{\infty} (R_n)^j\ket{n}.
\end{equation}
We use this to prove \eqref{eq:EvenWeirdOpAreBounded}. Note that
\begin{equation}
\begin{split}
\label{eq:TruncAnnDiff}
&\sum_{m\ge n_0+1} u_m\sqrt{m}\Big(\ket{E_{m-1}}\!\bra{E_m} - \ket{m-1}\!\bra{m}\Big) \\
&= \sum_{m\ge n_0+1}u_m\sqrt{m}(c_mc_{m-1}-1)\ket{m-1}\!\bra{m} \\& \quad+ \sum_{m\ge n_0+1} u_m c_{m}c_{m-1}\sqrt{m}\sum_{(j,l)\in\N_0^2\setminus \{(0,0)\}} (R_{m-1})^j\ket{m-1}\!\bra{m}(R^\dagger_m)^l.
\end{split}
\end{equation}
Since we have $|u_m|\,\sqrt{m}|1-c_mc_{m-1}|\le\frac{32|\psi|}{U}$, we see that the first term in \eqref{eq:TruncAnnDiff} is a bounded operator satisfying
\begin{align}
\label{eq:Weirdterm1}
    \left\|\sum_{m\ge n_0+1}u_m\sqrt{m}(c_mc_{m-1}-1)\ket{m-1}\!\bra{m}\right\| \le \frac{32|\psi|}{U}.
\end{align}
For the second term in \eqref{eq:TruncAnnDiff}, we denote $\widetilde u_m = u_mc_mc_{m-1}$ which satisfies $|\widetilde u_m|\le 1$ and write
\begin{equation}
\begin{split}
\label{eq:TwoWeirdTerms}
   &\sum_{m\ge n_0+1}\widetilde u_m\sum_{(j,l)\in\N_0^2\setminus\{(0,0)\}} \sqrt{m}(R_{m-1})^j\ket{m-1}\!\bra{m}(R_m^\dagger)^l \\
&=  \sum_{m\ge n_0+1}\widetilde u_m\sum_{0<j+l< 4} \sqrt{m}(R_{m-1})^j\ket{m-1}\!\bra{m}(R_m^\dagger)^l \\&\quad+\sum_{m\ge n_0+1}\widetilde u_m\sum_{j+l\ge 4} \sqrt{m}(R_{m-1})^j\ket{m-1}\!\bra{m}(R_m^\dagger)^l. 
   \end{split}
\end{equation}
We bound both terms individually. For the first term, we use the fact that from \eqref{eq:RnOnFock}, we have for $j\in\N_0$
\begin{align*}
    (R_m)^j\ket{m} = \sum_{-j\le k\le j} d^{(j)}_{k,m}\ket{m+k}
\end{align*}
for some coefficients $d^{(j)}_{k,m}$ satisfying $|d^{(j)}_{k,m}| \le \|R_m\|^j\le \left(\frac{16|\psi|}{U\sqrt{m+1}}\right)^j$ by \eqref{eq:Rnbound}, and therefore 
\begin{align*}
   & \sum_{m\ge n_0+1} \widetilde u_m\sqrt{m}(R_{m-1})^j\ket{m-1}\!\bra{m}(R_m^\dagger)^l \\
&= \sum_{\substack{-j\le k_1\le j\\-l\le k_2\le l}} \sum_{m\ge n_0+1}  \widetilde u_m\sqrt{m}d^{(j)}_{k_1,m-1}\overline{d^{(l)}_{k_2,m}}\ket{m-1+k_1}\!\bra{m+k_2}.
\end{align*}
Using that for $j+l>0$ and $\frac{16|\psi|}{U}<1$, we have $|\widetilde u_m|\,\sqrt{m}\,|d^{(j)}_{k_1,m-1}d^{(l)}_{k_2,m}|\le \frac{16|\psi|}{U}$; we see that this operator is, in fact, bounded, which, in particular, gives for the first term in \eqref{eq:TwoWeirdTerms} that
\begin{equation}
\label{eq:WeirdTerm2}
    \begin{aligned}
&\left\|\sum_{m\ge n_0+1}\sum_{0<j+l< 4} \widetilde u_m\sqrt{m}(R_{m-1})^j\ket{m-1}\!\bra{m}(R_m^\dagger)^l\right\| \\
&\le \sum_{0<j+l<4} \sum_{\substack{-j\le k_1\le j\\-l\le k_2\le l}}\frac{16|\psi|}{U}= \frac{1104|\psi|}{U}.
\end{aligned}
\end{equation}
Denoting $\gamma:= \frac{16|\psi|}{U}<1$, the second term in \eqref{eq:TwoWeirdTerms} can be estimated using \eqref{eq:Rnbound} as
\begin{align*}
    &\left\|\sum_{m\ge n_0+1}\widetilde u_m\sum_{j+l\ge 4} \sqrt{m}(R_{m-1})^j\ket{m-1}\!\bra{m}(R_m^\dagger)^l \right\|\\
&\le \left(\sum_{m\ge n_0+1}\frac{1}{m^{3/2}} \right)\left(\sum_{k=4}^\infty  (k+1)\gamma^{k}\right) \lesssim  \frac{\gamma^{4}}{(1-\gamma)^{2}}.
\end{align*}
Combining this with \eqref{eq:Weirdterm1} and \eqref{eq:WeirdTerm2} proves \eqref{eq:EvenWeirdOpAreBounded}. The statement \eqref{eq:EvenWeirdOpAreBoundedCreation} follows analogously.

We continue to prove \eqref{eq:mfBHMatrixDecayAnn}: Using \eqref{eq:EnVonNeumannFock}, we have for $n,m\ge 3 + \frac{4\mu}{U}$ 
\begin{align}
\label{eq:OffdiagonalLargemnProof}
    \nn&\bra{ E_m}a\ket{ E_n} -\sqrt{n}\delta_{m+1,n} = (c_nc_m-1)\sqrt{n}\delta_{m+1,n}+c_nc_m\sum_{(j,l)\in \N^2_0\setminus\{(0,0)\}} \bra{m}(R^\dagger_m)^l a(R_n)^j\ket{n}.
\end{align}
Using \eqref{eq:cnLowerBound}, the first term can be bounded as $|c_nc_m-1|\sqrt{n} \delta_{m+1,n}\le \frac{32|\psi|}{U} \delta_{m+1,n}.$
For the second term in \eqref{eq:OffdiagonalLargemnProof}, we first note that, by the orthogonality of the Fock states and by the definition of $R_n$, we have that
\begin{align*}
    \bra{m}(R^\dagger_m)^l a(R_n)^j\ket{n} =0
\end{align*}
for all $k := l+j <|m-n|-1\le|m-(n-1)|$. Hence, denoting once again $\gamma:=\frac{16|\psi|}{U} <1$ and $K:= \max\{|m-n|-1,1\}$, and noting further that $|c_n|\le 1$, we see that
\begin{align*}
  \left|c_nc_m\sum_{(j,l)\in \N^2_0\setminus\{(0,0)\}} \bra{m}(R^\dagger_m)^l a(R_n)^j\ket{n}\right| &\le \sum_{k =K}^{\infty} (k+ 1) \gamma^{k}  =\frac{\gamma^{K}\,(K+1-K\gamma)}{(1-\gamma)^{2}}\le \frac{16|\psi|}{U}\, \frac{\gamma^{K-1}(K+1)}{(1-\gamma)^2}.
\end{align*}
Using  $\gamma \le 1-\kappa$ for some $\kappa\in(0,1)$ and $ 
\sup_{K\ge0}\left(\frac{1-\kappa}{1-\kappa/2}\right)^{K-1}\frac{K+1}{\kappa^2}<\infty$  shows \eqref{eq:mfBHMatrixDecayAnn}. The corresponding statement for the creation operator, i.e. \eqref{eq:mfBHMatrixDecayCrea}, follows analogously.

For $n< 3+\frac{4\mu}{U}$ and $\delta\in(|\psi|,\frac{U}{16})$, we consider the set $\mathcal{C}_n$ from Lemma~\ref{lem:E_n0} and~\ref{lem:mfBHEigPert} and define $P^{(0)}_{\mathcal{C}_n}=\sum_{m\in\mathcal{C}_n}\kb{m}$, $Q^{(0)}_{\mathcal{C}_n} = \1 - P^{(0)}_{\mathcal{C}_n}$
and 
\begin{align}
   \widetilde R_n= Q^{(0)}_{\mathcal{C}_n}(E_n-H_0)^{-1}Q^{(0)}_{\mathcal{C}_n} V,
\end{align}
where we interpret $ Q^{(0)}_{\mathcal{C}_n}(E_n-H_0)^{-1}Q^{(0)}_{\mathcal{C}_n}$ as a resolvent on the space $Q^{(0)}_{\mathcal{C}_n}\mathcal{H}$, which is well-defined by \eqref{eq:Enlowlowerbound} of Lemma~\ref{lem:mfBHEigPert}. Using \eqref{eq:Enlowlowerbound} and arguing  similarly as for \eqref{eq:Rnbound} we find
\begin{align}
\label{eq:tildeRnbound}
\|\widetilde R_n\|,\,\|a\widetilde R_n\|,\,\|a^\dagger \widetilde R_n\| \le \frac{|\psi|}{\delta}<1. 
\end{align}
From the eigenvalue equation $H \ket{E_n} =E_n \ket{E_n}$, we find that $(E_n-H_0)\ket{E_n} = V\ket{E_n}$, 
\begin{align*}
    Q^{(0)}_{\mathcal{C}_n} \ket{E_n} = \widetilde R_n\ket{E_n}
\end{align*}
and therefore
\begin{align} 
\label{eq:EnvectorLow}
\ket{E_n} = P^{(0)}_{\mathcal{C}_n}\ket{E_n} + \widetilde R_n \ket{E_n}, 
\end{align}
which gives
\begin{equation}
\label{eq:EnVonNeumannLow}
\ket{E_n} = \left(\1-\widetilde R_n\right)^{-1}P^{(0)}_{\mathcal{C}_n}\ket{E_n} =\sum_{j=0}^{\infty}(\widetilde R_n)^jP^{(0)}_{\mathcal{C}_n}\ket{E_n}.
\end{equation}
Denoting $P_{\mathcal{C}_n}= \sum_{m\in\mathcal{C}_n}\kb{E_m},$, we 
see from \eqref{eq:EnvectorLow} and \eqref{eq:tildeRnbound}, together with the fact that $\mathcal{C}_n$ forms equivalence classes (cf.~Lemma \ref{lem:E_n0}), that 
\begin{equation}
\begin{split}
 &\left\|P_{\mathcal{C}_n}-P^{(0)}_{\mathcal{C}_n} P_{\mathcal{C}_n}P^{(0)}_{\mathcal{C}_n}\right\| \\
&\le  \left\|\sum_{m\in\mathcal{C}_n}P^{(0)}_{\mathcal{C}_n}\kb{E_m}\widetilde R_m^\dagger\right\|\\
&\quad+\left\|\sum_{m\in\mathcal{C}_n}\widetilde R_m\kb{E_m}P^{(0)}_{\mathcal{C}_n}\right\|\\
&\quad+\left\|\sum_{m\in\mathcal{C}_n}\widetilde R_m\kb{E_m}\widetilde R_m^\dagger\right\|\\
 &\le \frac{12|\psi|}{\delta}\left(1+\frac{\mu}{U}\right),
 \end{split}
\end{equation}
where in the last inequality we have used that $|\mathcal{C}_n|\le \left(4 +\frac{4\mu}{U}\right)$.
Since $P^{(0)}_{\mathcal{C}_n} P_{\mathcal{C}_n}P^{(0)}_{\mathcal{C}_n}$ has support in $P^{(0)}_{\mathcal{C}_n}\mathcal{H}$ with eigenvalues lying above in $[1-\frac{12|\psi|}{\delta}\left(1+\frac{\mu}{U}\right),1]$, we also see that $\|P^{(0)}_{\mathcal{C}_n}-P^{(0)}_{\mathcal{C}_n} P_{\mathcal{C}_n}P^{(0)}_{\mathcal{C}_n}\| \le \frac{12|\psi|}{\delta}\left(1+\frac{\mu}{U}\right)$ and therefore \begin{align*}
    \left\|P_{\mathcal{C}_n}-P^{(0)}_{\mathcal{C}_n}\right\| \le \frac{24|\psi|}{\delta}\left(1+\frac{\mu}{U}\right).
\end{align*}
Using the same argument together with \eqref{eq:tildeRnbound} and the fact that $\|a P^{(0)}_{\mathcal{C}_n}\|, \|a^\dagger P^{(0)}_{\mathcal{C}_n}\| \le 2\sqrt{1+\frac{\mu}{U}}$, we see
\begin{align*}
    \left\|a\left(P_{\mathcal{C}_n}-P^{(0)}_{\mathcal{C}_n}\right)\right\|  \le \frac{48|\psi|}{\delta}\left(1+\frac{\mu}{U}\right)^{3/2}\text{ and }\left\|a^\dagger\left(P_{\mathcal{C}_n}-P^{(0)}_{\mathcal{C}_n}\right)\right\|  \le \frac{48|\psi|}{\delta}\left(1+\frac{\mu}{U}\right)^{3/2}.
\end{align*}
This shows \eqref{eq:P_CnCloseness}.

Lastly, to show \eqref{eq:nLowvsmLargeAnn}, we first consider the sum
\begin{align*}
   \sum_{j=0}^{\infty} P^{(0)}_{\mathcal{C}_n}(\widetilde R^\dagger_n)^ja\ket{E_m} = \!c_m\!\!\sum_{j,l=0}^\infty \!P^{(0)}_{\mathcal{C}_n}(\widetilde R^\dagger_n)^ja(R_m)^l\!\ket{m},
 \end{align*}
 where we used \eqref{eq:EnVonNeumannFock} once again. Denoting $\operatorname{dist}(\mathcal{C}_n,m) = \min_{m_1\in \mathcal{C}_n}|m_1-m|$, we see
 \begin{align*}
   P^{(0)}_{\mathcal{C}_n}(\widetilde R^\dagger_n)^ja(R_m)^l\ket{m} = 0
 \end{align*}
 for all $k:=l+j<\operatorname{dist}(\mathcal{C}_n,m) -1$. Therefore, denoting $K=\max\{\operatorname{dist}(\mathcal{C}_n,m) -1,1\}$ and $\eta = \frac{|\psi|}{\delta}$, and using \eqref{eq:Rnbound} and \eqref{eq:tildeRnbound}, we get that 
 \begin{align*}
&\left\|\sum_{j=0}^{\infty} P^{(0)}_{\mathcal{C}_n}(\widetilde R^\dagger_n)^ja\ket{E_m} - \sqrt{m}P^{(0)}_{\mathcal{C}_n}\ket{m-1}\right\| \\
&\le (1-c_m)\sqrt m\left\|P^{(0)}_{\mathcal C_n}\ket{m-1}\right\|+\sum_{(j,l)\in\N^2_0\setminus\{(0,0)\}}\left\|\ P^{(0)}_{\mathcal{C}_n}(\widetilde R^\dagger_n)^ja(R_m)^l\ket{m}\right\|\\& \\
&\le \eta\1_{\{m-1\in\mathcal C_n\}}+\sum_{k=K}^\infty  (k+1)\,\eta^k\\
&=\eta\1_{\{m-1\in\mathcal C_n\}}+\frac{\eta^{K}\bigl(K+1 - K\eta\bigr)}{(1 - \eta)^2} \le \eta\1_{\{m-1\in\mathcal C_n\}}+\eta\,\frac{\eta^{K-1}\bigl(K+1\bigr)}{(1 - \eta)^2}.
\end{align*}
Here we used $(1-c_m)\sqrt m\le16|\psi|/U\le\eta$ from \eqref{eq:cnLowerBound}. The additional term vanishes unless $m-1\in\mathcal C_n$, in which case $\operatorname{dist}(\mathcal C_n,m)=1$. Using $\eta\le 1-\kappa$ for some $\kappa\in(0,1)$ and $
C'_\kappa:=2\left(1-\frac{\kappa}{2}\right)^{-2}\sup_{K\ge0}\left(\frac{1-\kappa}{1-\kappa/2}\right)^{K-1}\frac{\bigl(K+1\bigr)}{\kappa^2}<\infty$ shows
\begin{equation}
\label{eq:LastExpDecayEver}
\left\|\sum_{j=0}^{\infty} P^{(0)}_{\mathcal{C}_n}(\widetilde R^\dagger_n)^ja\ket{E_m} - \sqrt{m}P^{(0)}_{\mathcal{C}_n}\ket{m-1}\right\|\le  C'_\kappa\,\eta \left(1-\frac{\kappa}{2}\right)^{\operatorname{dist}(\mathcal{C}_n,m)}. 
\end{equation}
Using \eqref{eq:EnVonNeumannLow} together with the equivalence-class structure of $\mathcal{C}_n$ (cf.~Lemma \ref{lem:E_n0}), we obtain
\begin{align*}
    P_{\mathcal{C}_n} = \sum_{k\in\mathcal{C}_n}\kb{E_k} =\sum_{\substack{k\in\mathcal{C}_n\\j\in\N_0}} \kb{E_k}P^{(0)}_{\mathcal{C}_n}(\widetilde R^{\dagger}_k)^j.
\end{align*}
Thus, we see that using \eqref{eq:LastExpDecayEver} and  $|\mathcal{C}_n|\le \left(4 +\frac{4\mu}{U}\right)$ 
\begin{equation}
\begin{split}
\label{eq:MoreInequalitiesMoreFun}\left\|P_{\mathcal{C}_n}a\ket{E_m}-\sqrt{m}P_{\mathcal{C}_n}P^{(0)}_{\mathcal{C}_n}\ket{m-1}\right\|  &\le |\mathcal{C}_n|C'_\kappa\,\eta \left(1-\frac{\kappa}{2}\right)^{\operatorname{dist}(\mathcal{C}_n,m)} \le C'_\kappa\left(4+\frac{4\mu}{U}\right)\eta \left(1-\frac{\kappa}{2}\right)^{\operatorname{dist}(\mathcal{C}_n,m)}.
     \end{split}
\end{equation}
Lastly, by \eqref{eq:P_CnCloseness} we have 
\begin{equation}
\sqrt{m}\left\|P_{\mathcal{C}_n}P^{(0)}_{\mathcal{C}_n}\ket{m-1} - P^{(0)}_{\mathcal{C}_n}\ket{m-1}\right\|  \le \begin{cases} \frac{48|\psi|}{\delta}\left(1+\frac{\mu}{U}\right)^{3/2},\quad &\text{if } m-1\in\mathcal{C}_n,\\ 0,&\text{otherwise}\end{cases}
\end{equation}
shows \eqref{eq:nLowvsmLargeAnn} for the annihilation operator. The corresponding statement involving the creation operator can be proven in an analogous fashion using that $P^{(0)}_{\mathcal{C}_n}a^\dagger\ket{m}=\sqrt{m+1}P^{(0)}_{\mathcal{C}_n}\ket{m+1}=0$ for all $n<3+\frac{4\mu}{U}$ and $m\ge3+\frac{4\mu}{U}.$ In fact, in this case, the argument in \eqref{eq:MoreInequalitiesMoreFun} already suffices, which leads to the improved scaling of the right hand side of \eqref{eq:nLowvsmLargeAnn} with respect to $\mu$ and $U.$
 \end{proof}

\noindent We are now ready to prove Theorem~\ref{thm:mfBHGap}.

\begin{proof}[Proof of Theorem~\ref{thm:mfBHGap}]
Consider for $\psi \in\C$ with $|\psi|\in I$ for some closed interval $I\subseteq [0,\frac{U}{16})$ the self-adjoint generator $L_{\hatfM,h(N)}$ on the space of Hilbert-Schmidt operators corresponding to the Hamiltonian $h(N)= \sum_{n=0}^\infty E_n\kb{n}$, with $h(n):=E_n$ being the energies of the mean field Bose-Hubbard Hamiltonian, $H,$ defined in \eqref{eq:mfBHHam}. 
We want to employ the spectral gap result for generators of number preserving Hamiltonians in our companion paper, i.e. \hskip0pt\relax%
\cite[Corollary~3.3]{BeckerRouzeSalzmannToAppearcmp}\hskip0pt%
, to show that $\operatorname{gap}(L_{\hatfM,h(N)})>0.$  
In particular combining \eqref{eq:H0Eigenvalues} with 
Lemma~\ref{lem:mfBHEigPert}, we see that the $E_n$ satisfy the growth assumption of \hskip0pt\relax%
\cite[Corollary~3.3]{BeckerRouzeSalzmannToAppearcmp}\hskip0pt%
:%Theorem~\ref{thm:SpectralGapGeneralh(N)} 
 In the language of \hskip0pt\relax%
\cite[Corollary~3.3]{BeckerRouzeSalzmannToAppearcmp}\hskip0pt%
 , we pick $n_0=\max\{0,\lceil3+4\mu/U\rceil\}$ and note that by \eqref{eq:H0Eigenvalues} and \eqref{eq.boundperturb} with $r=U/8$ we have for $n\ge n_0$
\begin{align*}
    E_{n+1} -E_n \ge  Un -\mu - \frac{U\sqrt{n+2}}{4} \ge \frac{U}{4}.
\end{align*}
Therefore, we can pick $\delta= \frac{U}{4}>0$ and $s=1$ in \hskip0pt\relax%
\cite[Corollary~3.3]{BeckerRouzeSalzmannToAppearcmp}\hskip0pt%
. From this and \eqref{eq:UniformDeltaEBound}, we can employ 
\begin{align}
    \label{eq:UniformGapBound}
    \operatorname{gap}(L_{\hatfM,h(N)}) \ge \kappa(\beta,n_0,\delta,s,\Delta_E) >0,
\end{align} which shows that 
\begin{align*}
    \operatorname{gap}(L_{\hatfM,h(N)}) \ge \kappa
\end{align*}
for some $\kappa>0$, which is independent of $\psi$ with $|\psi|\in I.$

Furthermore, the corresponding Dirichlet form, defined through $\mathcal{E}(x) = -\langle x,L_{\hatfM,h(N)}x\rangle_{2}$, can explicitly be written as (cf. \cite[Section 2.1]{BeckerRouzeSalzmannToAppearcmp})
\begin{align*}
\cE(x) &= \int_{-\infty}^\infty \frac{\langle\partial^+_t(x),\partial^+_t(x)\rangle_{2} +\langle\partial^-_t(x),\partial^-_t(x)\rangle_{2}}{\beta \cosh(2\pi t/\beta)} dt
\end{align*}
with
\begin{align*}
    \partial^{+}_t(x) &= \sum_{n=0}^{\infty} \hatfM(E_{n+1}-E_n) e^{i(E_{n+1}-E_n) t} \sqrt{n+1}\left(\ket{n+1}\!\!\bra{n} x - e^{\beta (E_{n+1}-E_n)/2}x\ket{n+1}\!\!\bra{n}\right)
\end{align*}
and \begin{align*}
    \partial^{-}_t(x) &= \sum_{n=1}^\infty \hatfM(E_{n-1}-E_n) e^{i(E_{n-1}-E_n) t} \sqrt{n}\left(\ket{n-1}\!\!\bra{n} x - e^{\beta (E_{n-1}-E_n)/2}x\ket{n-1}\!\!\bra{n}\right).
\end{align*}

The reference kernel is one-dimensional. Indeed, if $\mathcal E(x)=0$, every matrix element of $\partial_t^\pm(x)$ vanishes for almost every $t$. Each contains at most two exponentials and is therefore continuous, so it also vanishes at $t=0$. Since $\hatfM$ never vanishes, the boundary entries first give $\langle0|x|j\rangle=\langle j|x|0\rangle=0$ for $j\ge1$. Iterating the annihilation relation then gives
\begin{align*}
 \langle m|x|n\rangle=\delta_{m,n}e^{-\beta(E_n-E_0)/2}\langle0|x|0\rangle,
\end{align*}
for $m,n\in\N_0$, hence $\ker L_{\hatfM,h(N)}=\C e^{-\beta h(N)/2}$.

In the following, we use the perturbation theory result for the spectral gap, Lemma~\ref{lem:GapPertDirichlet2}, to establish that $L_{\hatfM,H}$ also has a positive spectral gap due to the positive gap of $L_{\hatfM,h(N)}$ defined above. 

For that, we use the notation and results corresponding to the spectral decomposition of $H$, established in Lemma~\ref{lem:mfBHEigPert} and Lemma~\ref{lem:mfBHEigVec}, to 
write the corresponding Dirichlet form $\widetilde\cE(x) = - \langle x,L_{\hatfM,H}\,x\rangle$ as
\begin{align*}
\widetilde\cE(x) &= \int_{-\infty}^\infty \frac{\langle\widetilde\partial^+_t(x),\widetilde\partial^+_t(x)\rangle_{2} +\langle\widetilde\partial^-_t(x),\widetilde\partial^-_t(x)\rangle_{2}}{\beta \cosh(2\pi t/\beta)}dt.
\end{align*} 
with
\begin{align*}
    \widetilde \partial^{\pm}_t(x) &= \sum_{n,m=0}^\infty \hatfM(E_n-E_m) e^{i(E_n-E_m) t} \Big(\kb{E_n}a^{\pm} \kb{E_m} x- e^{\beta (E_n-E_m)/2}x\kb{E_n}a^{\pm} \kb{E_m}\Big).
\end{align*}
In the following, we show that $(\widetilde \partial^{+}_t -\partial^{+}_t)(x) = B^+_Lx - xB^+_R$ and $(\widetilde \partial^{-}_t -\partial^{-}_t)(x) = B^-_Lx - xB^-_R$, where 
\begin{align}
\label{eq:PartialOHMYGOD}
\nn B^+_L&=\sum_{n,m=0}^\infty \hatfM(E_n-E_m) e^{i(E_n-E_m) t} \big(\kb{E_n}a^\dagger \kb{E_m}- \delta_{n,m+1}\sqrt{m+1}\ket{m+1}\!\bra{m} \big),\\\nn
B^+_R&=\sum_{n,m=0}^\infty \hatfM(E_n-E_m) e^{i(E_n-E_m) t}e^{\beta (E_n-E_m)/2}\big(\kb{E_n}a^\dagger \kb{E_m}-\delta_{n,m+1}\sqrt{m+1}\ket{m+1}\!\bra{m}\big),\\ \nn
B^-_L&=\sum_{n,m=0}^\infty \hatfM(E_n-E_m) e^{i(E_n-E_m) t} \big(\kb{E_n}a \kb{E_m}- \delta_{n,m-1}\sqrt{m}\ket{m-1}\!\bra{m} \big),\\
B^-_R&=\sum_{n,m=0}^\infty \hatfM(E_n-E_m) e^{i(E_n-E_m) t}e^{\beta (E_n-E_m)/2}\big(\kb{E_n}a \kb{E_m}-\delta_{n,m-1}\sqrt{m}\ket{m-1}\!\bra{m}\big),
\end{align}
define bounded operators with operator norm going to zero as $|\psi|\to 0.$ We focus in the following on the $B^-_L$ term, as the other terms can be treated similarly. To prove its boundedness, we split the sums occurring in \eqref{eq:PartialOHMYGOD} into four different regions of indices $n,m$, namely: 
\[
\begin{array}{c|cc}
 & m < 3+\frac{4\mu}{U} & m \ge 3+\frac{4\mu}{U} \\ \hline
n < 3+\frac{4\mu}{U}
  & \text{(1)} 
  & \text{(3)} \\[1mm]
n \ge 3+\frac{4\mu}{U}
  & \text{(4)} 
  & \text{(2)}
\end{array}
\]

and treat each of them individually. Throughout, we assume $\frac{16|\psi|}{U}\le 1-\kappa$ for some $\kappa\in(0,1)$ such that the results from Lemma~\ref{lem:mfBHEigVec} can be employed. 

We start with the region of small $n,m:$ i.e. region (1): 

\bigskip

\noindent \underline{Case $n,m <3 +\frac{4\mu}{U}:$}

\smallskip

\noindent For this, we need the following continuity bounds: Since $|\frac{d}{d\nu}\hatfM(\nu)| \le \frac{\beta}{2},$, we see by the mean value theorem that for $\nu,\nu'\in \R$, we have

     \begin{equation}
    \label{eq:ContBound1}
    |\hatfM(\nu)e^{it\nu} - \hatfM(\nu')e^{it\nu'}| \le \left(\frac{\beta}{2}+|t|\right)|\nu-\nu'|
    \end{equation}
    and similarly, using $|\frac{d}{d\nu}(\hatfM(\nu)e^{\beta\nu/2})| \le \frac{\beta}{2}$ that
    \begin{align}
    \label{eq:ContBound2}
    |\hatfM(\nu)e^{it\nu}e^{\beta\nu/2} - \hatfM(\nu')e^{it\nu'}e^{\beta\nu'/2}| \le \left(\frac{\beta}{2}+|t|\right)|\nu-\nu'|.
    \end{align}

For $n<3+\frac{4\mu}{U}$ and $\delta\in(|\psi|,\frac{U}{16})$ we consider the set $\mathcal{C}_n\subset \{0,\cdots,3+\frac{4\mu}{U}\}$ defined and studied in Lemma~\ref{lem:E_n0},~\ref{lem:mfBHEigPert}
and~\ref{lem:mfBHEigVec}.  
Using this, in particular the bound \eqref{eq:EnCnCloseness}, and \eqref{eq:ContBound1} we see 
\begin{align}
\label{eq:SmallEnergyFromCont}
&\nn\left\|\sum_{n,m<3+\frac{4\mu}{U}} \hatfM(E_n-E_m) e^{i(E_n-E_m) t} \kb{E_n}a\kb{E_m}-  \sum_{\substack{[n],[m]\\n,m<3+\frac{4\mu}{U}}}\hatfM(E_n-E_m) e^{i(E_n-E_m) t} P_{\mathcal{C}_n}aP_{\mathcal{C}_m}\right\| \\
&\nn\le \left\|\sum_{\substack{[n],[m]\\n,m<3+\frac{4\mu}{U}}}\sum_{k\in[n],l\in[m]}\Big(\hatfM(E_k-E_l) e^{i(E_k-E_l) t}-\hatfM(E_n-E_m) e^{i(E_n-E_m) t}\Big)\kb{E_k}a\kb{E_l}\right\|
\\&\lesssim\delta \left(\frac{\beta}{2} +|t|\right) \max\left\{\left(1+\frac{\mu}{U}\right)^{7/2},1\right\}
\end{align}
In the second sum on the left-hand side above, we used the notation $[n] \equiv \mathcal{C}_n$ to denote that we are not summing over the individual elements in $\mathcal{C}_n$ but rather the different sets occurring, and the energies $E_n$, resp.~$E_m$, are chosen with respect to any representative $n\in[n]$, resp~$m\in[m]$. Furthermore, in the last inequality, we used Schur's test\footnote{For some orthonormal basis $(\ket{i})_i$ and operator $B= \sum_{i,j}b_{ij}\ket{i}\!\bra{j}$ satisfying $M:=\max\{\sup_i\sum_j|b_{ij}|,\sup_j\sum_i|b_{ij}|\}<\infty,$ Schur's test gives that $B$ is a bounded operator with $\|B\| \le M.$}\setcounter{BHSchurFootnote}{\value{footnote}} and the fact that $\|a\ket{E_l}\|\le \sqrt{3+\frac{4\mu}{U}} + 1$ for $l<3 +\frac{4\mu}{U}$, which follows from \eqref{eq:EnvectorLow}, \eqref{eq:tildeRnbound}, with 
$\|a\widetilde R_n \ket{E_n}\|\le \frac{|\psi|}{\delta }\le 1$ and $\|a P^{(0)}_{\mathcal{C}_n}\| \le \sqrt{3+\frac{4\mu}{U}}$. 

Similarly, we can replace the corresponding sums over terms involving $\kb{n}a\kb{m}= \delta_{n,m-1}\sqrt{m}\ket{m-1}\!\bra{m}$ with $n,m<3+\frac{4\mu}{U}$ in $B^-_L$ of \eqref{eq:PartialOHMYGOD} by corresponding terms involving  $P^{(0)}_{\mathcal{C}_n}aP^{(0)}_{\mathcal{C}_m}.$ Lastly, we use \eqref{eq:P_CnCloseness} which gives 

\begin{align*}
\left\|\sum_{\substack{[n],[m]\\n,m<3+\frac{4\mu}{U}}}\hatfM(E_n-E_m) e^{i(E_n-E_m) t} \Big(P_{\mathcal{C}_n}aP_{\mathcal{C}_m} -P^{(0)}_{\mathcal{C}_n}aP^{(0)}_{\mathcal{C}_m}\Big)\right\| \lesssim \frac{|\psi|}{\delta}\max\left\{\left(1+\frac{\mu}{U}\right)^{7/2}, 1\right\}.
\end{align*}

In total, this gives

    \begin{align}
\label{eq:BoundRegion1}
     \nn&\left\|\sum_{n,m<3+\frac{4\mu}{U}} \hatfM(E_n-E_m) e^{i(E_n-E_m) t}\left( \kb{E_n}a\kb{E_m} -\kb{n}a\kb{m}\right)\right\|\\
&\lesssim   \left(\delta \left(\tfrac{\beta}{2} +|t|\right)+ \tfrac{|\psi|}{\delta}\right)\max\left\{\left(1+\tfrac{\mu}{U}\right)^{7/2}, 1\right\}.
    \end{align}

We continue with the region of both $n,m$ large, i.e., region (2):

\bigskip

\noindent \underline{Case $ n,m\ge 3+\frac{4\mu}{U}:$}

\noindent 
We choose $\psi$ small enough such that $\frac{16|\psi|}{U}\le 1-\kappa$ for some $\kappa\in(0,1)$
and employ \eqref{eq:mfBHMatrixDecayAnn}
from Lemma~\ref{lem:mfBHEigVec} to see 
\begin{align*}
    \left|\bra{E_n}a\ket{E_m}-\delta_{n,m-1}\sqrt{m}\right| \le C_\kappa \,\frac{|\psi|}{U} \left(1-\frac{\kappa}{2}\right)^{|m-n|}
\end{align*}
for some $C_\kappa\ge 0.$ 
Furthermore, note that
\begin{align*}
    \sup_{n\in\N_0}\sum_{m=0}^{\infty} \left(1-\frac{\kappa}{2}\right)^{|m-n|} =:\widetilde C_\kappa <\infty
\end{align*}
for some $ \widetilde C_\kappa\ge 0$ depending only on $\kappa$. Combining this $|\hatfM|\le 1$, we can employ Schur's test\footnotemark[\value{BHSchurFootnote}] which gives

    \begin{align*}
       &\left\|\sum_{n,m\ge 3+\frac{4\mu}{U}}^\infty \hatfM(E_n-E_m) e^{i(E_n-E_m) t} \big(\kb{E_n}a \kb{E_m}- \delta_{n,m-1}\sqrt{m}\ket{E_{m-1}}\!\bra{E_m} \big) \right\| \\&=\left\|\sum_{n,m\ge 3+\frac{4\mu}{U}}^\infty \hatfM(E_n-E_m) e^{i(E_n-E_m) t} \big(\bra{E_n}a\ket{ E_m}- \delta_{n,m-1}\sqrt{m}\big) \ket{E_n}\!\bra{E_m}\right\| \le C_\kappa \widetilde C_\kappa \frac{|\psi|}{U}.    \end{align*}

Combining this with \eqref{eq:EvenWeirdOpAreBounded} from Lemma~\ref{lem:mfBHEigVec} and the choice $u_m=\hatfM(E_{m-1}-E_m)e^{it(E_{m-1}-E_m)}$  which satisfies $|u_m|\le 1$, we see 

    \begin{align}
 \label{eq:BoundRegion2}
       &\left\|\sum_{n,m\ge 3+\frac{4\mu}{U}}^\infty \hatfM(E_n-E_m) e^{i(E_n-E_m) t} \big(\kb{E_n}a \kb{E_m}- \delta_{n,m-1}\sqrt{m}\ket{m-1}\!\bra{m} \big) \right\|  \le C'_\kappa\frac{|\psi|}{U}   \end{align}

for some $C'_\kappa\ge 0$ depending only on $\kappa.$

We continue with the region (3).

\bigskip

\noindent \underline{Case $n<3+\frac{4\mu}{U}\text{ and } m\ge 3+\frac{4\mu}{U}:$} 
For $n<3+\frac{4\mu}{U}$ and $\delta\in(|\psi|,\frac{U}{16})$, we consider the set $\mathcal{C}_n\subset \{0,\cdots,3+\frac{4\mu}{U}\}$ defined and studied in Lemma~\ref{lem:E_n0},~\ref{lem:mfBHEigPert}, and~\ref{lem:mfBHEigVec}.   
Arguing similarly to \eqref{eq:SmallEnergyFromCont}, in particular, using the bounds \eqref{eq:EnCnCloseness} and \eqref{eq:ContBound1}, we see 

\noindent 

    \begin{align*}
        &\left\|\sum_{\substack{n<3+\frac{4\mu}{U}\\m\ge 3+\frac{4\mu}{U}}} \hatfM(E_n-E_m) e^{i(E_n-E_m) t} \kb{E_n}a\kb{E_m}-  \sum_{\substack{[n] \text{ s.t. } n<3+\frac{4\mu}{U},\\m\ge 3+\frac{4\mu}{U}}}\hatfM(E_n-E_m) e^{i(E_n-E_m) t} P_{\mathcal{C}_n}a\kb{E_m}\right\| \\&\le
         \left\|\sum_{\substack{[n] \text{ s.t. } n<3+\frac{4\mu}{U},\\m\ge 3+\frac{4\mu}{U}}}\sum_{k\in[n]}\left( \hatfM(E_k-E_m) e^{i(E_k-E_m) t}- \hatfM(E_n-E_m) e^{i(E_n-E_m) t}\right) \kb{E_k}a\kb{E_m}\right\| \\&\le\sum_{\substack{[n] \text{ s.t. } n<3+\frac{4\mu}{U}}}\sum_{k\in[n]}\|\kb{E_k}a\| \left\|\sum_{m\ge 3+\frac{4\mu}{U}}\left( \hatfM(E_k-E_m) e^{i(E_k-E_m) t}- \hatfM(E_n-E_m) e^{i(E_n-E_m) t}\right) \kb{E_m}\right\|\\&\lesssim \delta\left(\frac{\beta}{2} +|t|\right) \max\left\{ \left(1+\frac{\mu}{U}\right)^{7/2},1\right\}
    \end{align*}

Here again, we used the notation $[n] \equiv \mathcal{C}_n$ to denote that we are not summing over the individual elements in $\mathcal{C}_n$ but rather the different sets occurring. Furthermore, in the last inequality, we used that $\|\kb{E_k}a\|=\|a^\dagger\ket{E_k}\|\le \sqrt{4+\frac{4\mu}{U}} + 1$ for $k<3 +\frac{4\mu}{U}$, which follows from \eqref{eq:EnvectorLow} and \eqref{eq:tildeRnbound}, with 
$\|a^\dagger\widetilde R_n \ket{E_n}\|\le \frac{|\psi|}{\delta }\le 1$ and $\|a^\dagger P^{(0)}_{\mathcal{C}_n}\| \le \sqrt{4+\frac{4\mu}{U}}$.

Similarly, we can replace the corresponding sums over terms involving $\kb{n}a\kb{m}= \delta_{n,m-1}\sqrt{m}\ket{m-1}\!\bra{m}$ with $n<3+\frac{4\mu}{U}$ in $B^-_L$ of \eqref{eq:PartialOHMYGOD} by corresponding terms involving  $P^{(0)}_{\mathcal{C}_n}a\kb{m}.$

 Lastly, we choose $\psi,\delta$ such that $\frac{|\psi|}{\delta}\le 1-\kappa$ for some $\kappa\in(0,1)$.

For the operator difference, we also account for the bra vectors by writing
\begin{align*}
 P_{\mathcal C_n}a\kb{E_m}-P^{(0)}_{\mathcal C_n}a\kb m
 &=\left(P_{\mathcal C_n}a\ket{E_m}-\sqrt m P^{(0)}_{\mathcal C_n}\ket{m-1}\right)\bra{E_m}\\
 &\quad+\sqrt m P^{(0)}_{\mathcal C_n}\ket{m-1}\left(\bra{E_m}-\bra m\right),
\end{align*}
where the second term vanishes unless $m-1\in\mathcal C_n$, and its norm is at most $32|\psi|/U$ by \eqref{eq:EnvsFocknLarge}. It therefore satisfies the same bound after increasing $C_\kappa$. Together with \eqref{eq:nLowvsmLargeAnn}, this gives
\begin{align*}
 &\left\|\sum_{\substack{[n] \text{ s.t. } n<3+\frac{4\mu}{U},\\m\ge 3+\frac{4\mu}{U}}} \hatfM(E_n-E_m) e^{i(E_n-E_m) t} \left(P_{\mathcal{C}_n}a\kb{E_m}-  P^{(0)}_{\mathcal{C}_n}a\kb{m}\right)\right\|   \\&\le C_\kappa\max\left\{ \left(1+\frac{\mu}{U}\right)^{3/2},1\right\} \,\frac{|\psi|}{\delta}\sum_{\substack{[n] \text{ s.t. } n<3+\frac{4\mu}{U},\\m\ge 3+\frac{4\mu}{U}}} \left(1-\frac{\kappa}{2}\right)^{\operatorname{dist}(\mathcal{C}_n,m)} \le C_\kappa\widetilde C_\kappa \max\left\{ \left(1+\frac{\mu}{U}\right)^{5/2},1\right\} \,\frac{|\psi|}{\delta},
\end{align*}

for some $C_\kappa,\widetilde C_\kappa\ge 0$, depending on $\kappa$, but on no other parameters. 

In total, we have proven

\begin{align}
\label{eq:BoundRegion3}
  \nn& \left\|\sum_{\substack{n<3+\frac{4\mu}{U}\\m\ge 3+\frac{4\mu}{U}}} \hatfM(E_n-E_m) e^{i(E_n-E_m) t} \left(\kb{E_n}a\kb{E_m}-\kb{n}a\kb{m}\right)\right\| \\&\lesssim \delta\left(\frac{\beta}{2} +|t|\right) \max\left\{ \left(1+\frac{\mu}{U}\right)^{7/2},1\right\} + \max\left\{ \left(1+\frac{\mu}{U}\right)^{5/2},1\right\}\,\frac{|\psi|}{\delta},
\end{align}

where we have hidden constants that depend only on $\kappa$ with the $\lesssim$ notation.

We continue with the region (4):

\bigskip

\noindent\underline{ Case $m<3+\frac{4\mu}{U}\text{ and } n\ge 3+\frac{4\mu}{U}:$} 

\smallskip

\noindent This term can similarly be estimated using $\kb{n}a\kb{m}=\sqrt m\,\delta_{n,m-1}\ket n\!\bra m=0$ in the considered regime and the fact that by \eqref{eq:nLowvsmLargeAnn}   \begin{align*}
    &\|\kb{E_n}a\kb{E_m}\| = |\bra{E_m}a^\dagger\ket{E_n}|\le C_\kappa\left(1+\frac{\mu}{U}\right) \,\frac{|\psi|}{\delta}\left(1-\frac{\kappa}{2}\right)^{\operatorname{dist}(\mathcal{C}_{m},n)}
\end{align*} which yields 
\begin{align}
\label{eq:BoundRegion4}
   \nn \left\|\sum_{\substack{m<3+\frac{4\mu}{U}\\n\ge 3+\frac{4\mu}{U}}} \hatfM(E_n-E_m) e^{i(E_n-E_m) t} \kb{E_n}a\kb{E_m}\right\| &\le C_\kappa\left(1+\frac{\mu}{U}\right) \,\frac{|\psi|}{\delta}\sum_{\substack{m<3+\frac{4\mu}{U}\\n\ge 3+\frac{4\mu}{U}}}\left(1-\frac{\kappa}{2}\right)^{\operatorname{dist}(\mathcal{C}_{m},n})\\&\le C_\kappa\widetilde C_\kappa \left(1+\frac{\mu}{U}\right)^2\, \frac{|\psi|}{\delta}.
\end{align}
for some $C_\kappa$ and $\widetilde C_\kappa$ depending only on $\kappa.$

\medskip

Collecting the estimates from all four regions, i.e. \eqref{eq:BoundRegion1}, \eqref{eq:BoundRegion2}, \eqref{eq:BoundRegion3}, and \eqref{eq:BoundRegion4}, and using that $\frac{|\psi|}{U}\le \frac{|\psi|}{\delta}$ since $\delta<\frac{U}{16}$, we obtain the bound below. In the case $1+\frac{\mu}{U}<0$, we interpret $\max\left\{ \left(1+\frac{\mu}{U}\right)^{5/2},1\right\}$ and $\max\left\{ \left(1+\frac{\mu}{U}\right)^{7/2},1\right\}$ as 1. The reason for this is that the regions (1), (3), and  (4), where either $n$ or $m$ or both are small, are all empty. Hence, in this case, we only need to take the bound from region (2) into account and obtain the simpler bound $\|B^{-}_L\| \lesssim \frac{|\psi|}{U}.$
\begin{align*}
    \|B^-_L\| &\lesssim\left( \delta \left(\frac{\beta}{2} +|t|\right)+\frac{|\psi|}{\delta}\right)\max\left\{ \left(1+\frac{\mu}{U}\right)^{7/2},1\right\}  .
\end{align*}
The case $\psi=0$ is the reference model itself. Assuming $0<|\psi|$ is small enough, in particular $|\psi| <\left(\frac{U}{16}\right)^2\beta$, and furthermore $|\psi|\le \beta^{-1}(1-\kappa')$ for some $\kappa'\in(0,1)$, we can choose \footnote{Note that this $\delta$ satisfies all requirements we have needed so far, namely $\delta< \frac{U}{16}$ and $\frac{|\psi|}{\delta}\le 1-\kappa$  for some $\kappa\in(0,1)$.} $\delta=\sqrt{\frac{|\psi|}{\beta}}$, which yields  

\begin{align*}
   \|B^-_L\| \lesssim \sqrt{|\psi|}\left(\sqrt{\beta}+\frac{|t|}{\sqrt{\beta}}\right)\,\max\left\{\left(1+\frac{\mu}{U}\right)^{7/2},1\right\}. 
\end{align*}

The terms $B^+_L, B^+_R$ and $B^-_R$ in \eqref{eq:PartialOHMYGOD} can be estimated in the same fashion, yielding the same upper bound on their respective operator norms. Therefore, the $\widetilde \partial^{\pm}_t -\partial^\pm_t$ defines a bounded operator on the space of Hilbert-Schmidt operators, with the corresponding operator norm being bounded as
\begin{align*}
   &\left\|\widetilde \partial^{\pm}_t -\partial^\pm_t\right\|\lesssim \sqrt{|\psi|}\left(\sqrt{\beta}+
   \frac{|t|}{\sqrt{\beta}}\right) \,\max\left\{\left(1+\frac{\mu}{U}\right)^{7/2},1\right\}
   . 
\end{align*}
To apply Lemma~\ref{lem:GapPertDirichlet2} to the closed forms, we verify a common form core. For $K\in\{h(N),H\}$, choose $c\ge0$ with $K+c\ge1$. The quadratic energy growth and the jump bounds above give $A^\alpha(K+c)^{-1}\in\mathscr T_2$. Multiplication of its energy-basis entries by $\hatfM(E'-E)e^{it(E'-E)}$ does not increase its Hilbert-Schmidt norm, and the same holds for $e^{\beta(E'-E)/2}\hatfM(E'-E)$ by the KMS condition. Consequently,
\begin{align*}
 \left\|\partial^\alpha_{t,K}(x)\right\|_2
 &\le\left\|A^\alpha(K+c)^{-1}\right\|_2\left\|(K+c)x\right\|_2\\
 &\quad+\left\|(A^\alpha)^\dagger(K+c)^{-1}\right\|_2\left\|x(K+c)\right\|_2,
\end{align*}
where $\partial^\alpha_{t,h(N)}=\partial^\alpha_t$ and $\partial^\alpha_{t,H}=\widetilde\partial^\alpha_t$. Since the graph norms of both $K+c$ are equivalent to that of $(N+1)^2$, this bounds either form by $C(\left\|(N+1)^2x\right\|_2^2+\left\|x(N+1)^2\right\|_2^2)$ for a constant $C\ge0$. Finite Fock truncations converge in this weighted norm whenever it is finite. Moreover, it is finite on both form cores $\mathscr F_K+\C\sqrt{\sigma_{\beta,K}}$ from \mbox{\cite[Lemma~2.3]{BeckerRouzeSalzmannToAppearcmp}}, where $\mathscr F_K$ denotes the span of the energy matrix units and $\sigma_{\beta,K}$ the Gibbs state of $K$, because the energy eigenvectors belong to $D(N^2)$ and the Gibbs states have finite energy moments. Approximating first by energy truncations and then by Fock truncations therefore shows that $\operatorname{span}\{|m\rangle\langle n|:m,n\in\N_0\}$ is a common form core. The integrated bounded difference estimated above makes the two form norms equivalent on this core, so the comparison extends to their closed forms.

From that, we can see \begin{align*}
    \Delta &:= \int_{-\infty}^{\infty } \frac{\big\|\widetilde \partial^{+}_t -\partial^+_t\big\|^2+\big\|\widetilde \partial^{-}_t -\partial^-_t\big\|^2}{\beta\cosh(2\pi t/\beta)} dt\lesssim |\psi|\beta\, \max\left\{\left(1+\frac{\mu}{U}\right)^7,1\right\},
\end{align*}
where in the inequality we have used that $\int_{-\infty}^{\infty } \frac{|t|^2\beta^{-1} +\beta}{\beta\cosh(2\pi t/\beta)} dt = C \beta$ for some $C\ge 0.$ Hence, we can use Lemma~\ref{lem:GapPertDirichlet2}, which gives for $|\psi|$ small enough such that $\operatorname{gap}(L_{\hatfM,h(N)}) > \Delta $ that $L_{\hatfM,H}$ has a positive spectral gap satisfying
\begin{align}
    \operatorname{gap}(L_{\hatfM,H}) \ge \left(\sqrt{\operatorname{gap}(L_{\hatfM,h(N)})} - \sqrt{\Delta}\right)^2>0.
\end{align}
\end{proof}

\section{An exponential gap for the Bose-Hubbard Gibbs sampler}
\label{BHexp:section}

In this section, we prove that the Bose-Hubbard Gibbs sampler $L_{\widehat f_\eta,H_{\operatorname{BH}}}$ has a positive spectral gap for $\eta\in[0,1/2]$. More precisely, we establish a lower bound on the gap that decays exponentially in the number of modes (Theorem~\ref{BHexp:theorem}).

Following the convention of Section~\ref{sec.BHmain}, we consider the Bose-Hubbard Hamiltonian on the lattice $\Lambda=\llbracket 1,L\rrbracket^D$ with open boundaries and $n=L^D$ modes,
\begin{align}
 H_{\operatorname{BH}}=\frac U2\sum_iN_i(N_i-1)-\mu\Ntot-J\sum_{\langle i,j\rangle}(a_i^\dagger a_j+a_j^\dagger a_i),
 \qquad \Ntot=\sum_iN_i,
 \label{BHexp:hamiltonian}
\end{align}
where 
 $U,\beta>0$, $\mu,J\in\R$, and $D\in\N$ is the lattice dimension. For $q\in\N_0$  the \emph{$q$-particle sector} is defined as
\begin{align*}
 \cH_q=\ker(\Ntot-q\id)
 =\operatorname{span}\left\{\lvert m_1,\ldots,m_n\rangle:
 (m_1,\ldots,m_n)\in\N_0^n,\ \sum_{i=1}^n m_i=q\right\},
\end{align*}
where $\lvert m_1,\ldots,m_n\rangle$ denotes a normalised occupation-number vector. Let $\Pi_q$ be the orthogonal projection from $\cH=L^2(\R^n)$ onto $\cH_q$, and define $\Pi_{\le R}=\sum_{q=0}^R\Pi_q$ for $R\in\N_0$. Each term in $H_{\operatorname{BH}}$ preserves total particle number, so $[H_{\operatorname{BH}},\Ntot]=0$ on finite-particle vectors and every $\cH_q$ is invariant. Consequently, $H_{\operatorname{BH}}$ decomposes into its finite-dimensional particle-sector restrictions as
\begin{align*}
 H_{\operatorname{BH}}=\bigoplus_{q=0}^{\infty}H_{\operatorname{BH},q},
 \qquad H_{\operatorname{BH},q}=H_{\operatorname{BH}}|_{\cH_q}.
\end{align*}
Since $\pm\sum_{\langle i,j\rangle}(a_i^\dagger a_j+a_j^\dagger a_i)\le2D\Ntot$, each restriction satisfies\footnote{To see the hopping bound, note that for each neighbouring pair, positivity of $(a_i\pm a_j)^\dagger(a_i\pm a_j)$ gives $\pm(a_i^\dagger a_j+a_j^\dagger a_i)\le N_i+N_j$. Summing over edges and using that every site has at most $2D$ neighbours gives the claimed bound. Moreover, to obtain \eqref{BHexp:energy}, note that on $\cH_q$ the occupation numbers satisfy $\sum_i m_i^2\ge q^2/n$ by Cauchy-Schwarz, so $\sum_iN_i(N_i-1)\ge(q^2/n-q)\id$. Combining this with the hopping bound and $-\mu q\ge-|\mu|q$ gives \eqref{BHexp:energy}.}
\begin{align}
 H_{\operatorname{BH},q}\ge\frac{Uq^2}{2n}-\left(\frac U2+|\mu|+2D|J|\right)q.
 \label{BHexp:energy}
\end{align}
Together with $\dim\cH_q=\binom{q+n-1}{n-1}$, this proves the Gibbs hypothesis. We define $\sigma_\beta=e^{-\beta H_{\operatorname{BH}}}/\mathcal Z_\beta$, where $\mathcal Z_\beta=\Tr(e^{-\beta H_{\operatorname{BH}}})$. Its positive square root satisfies $\left\|\sqrt{\sigma_\beta}\right\|_2=1$. All constants below may depend on $\beta,U,\mu,J,D$, but not on $n,q,R$ or $\eta\in[0,1/2]$.

For $0\le\eta\le1/2$, we consider the polynomially decaying filters
\begin{align}
 \hatfM(\nu)=\exp\left(-\frac{\sqrt{1+\beta^2\nu^2}+\beta\nu}{4}\right),
 \qquad \widehat f_\eta(\nu)=\frac{\hatfM(\nu)}{(1+\beta^2\nu^2)^{\eta/2}}.
 \label{BHexp:filter}
\end{align}
In particular, $\widehat f_0=\hatfM$. Each filter satisfies the KMS symmetry $\widehat f_\eta(-\nu)=e^{\beta\nu/2}\widehat f_\eta(\nu)$ and the bound $\left\|\widehat f_\eta\right\|_\infty\le1$. Take $\cA=[n]\times\{-,+\}$, with $A^{i,-}=a_i$ and $A^{i,+}=a_i^\dagger$. To treat both these filters and the auxiliary filters introduced below, we define the filtered operators for a general bounded real KMS-symmetric function $\widehat f$, satisfying $\widehat f(-\nu)=e^{\beta\nu/2}\widehat f(\nu)$. Using the operators $A^\alpha_\nu$ from Section~\ref{sec:CVGS} with $H=H_{\operatorname{BH}}$, we define
\begin{align*}
 T_{\widehat f}^\alpha(t)&=\sum_\nu e^{it\nu}\widehat f(\nu)A^\alpha_\nu,
 &S_{\widehat f}^\alpha(t)&=\sum_\nu e^{it\nu}\widehat f(-\nu)A^\alpha_\nu,\\
 \partial_{t,\widehat f}^\alpha(X)&=T_{\widehat f}^\alpha(t)X-XS_{\widehat f}^\alpha(t),
 &\partial_{t,\widehat f}(X)&=\big(\partial_{t,\widehat f}^\alpha(X)\big)_\alpha.
\end{align*}
Here $\partial_{t,\widehat f}^\alpha$ is the derivation from Section~\ref{sec:GapPertDirichlet}, with the filter dependence made explicit. For a finite index set $I$ and a family $F=(F^\alpha(t))_{\alpha\in I,\,t\in\R}$ of Hilbert-Schmidt operators, measurable in $t$, we define
\begin{align*}
 \left\|F\right\|_w^2
 &=\sum_{\alpha\in I}\int_{\R}w_\beta(t)\left\|F^\alpha(t)\right\|_2^2\,dt,
\end{align*}
where $w_\beta(t)=1/(\beta\cosh(2\pi t/\beta))$. Below we use both $I=\mathcal A=[n]\times\{-,+\}$, indexing modes and jump signs, and $I=[n]$, indexing only modes.
With this notation, the Dirichlet form from Section~\ref{sec:GapPertDirichlet} reads $\cE_{\widehat f,H_{\operatorname{BH}}}(X)=\left\|\partial_{t,\widehat f}(X)\right\|_w^2$.

\begin{theorem}[Exponential lower bound on spectral gap]
\label{BHexp:theorem}
Let $H_{\operatorname{BH}}$ be the Bose-Hubbard Hamiltonian on the $D$-dimensional lattice $\Lambda=\llbracket 1,L\rrbracket^D$, where $D\in\N$, with open boundaries and $n=L^D$ modes, as defined in \eqref{BHexp:hamiltonian}, with $U>0$ and $\mu,J\in\R$. For $\beta>0$, there are $C,c>0$, depending only on $\beta,U,\mu,J,D$, such that for $L\in\N$, $\eta\in[0,1/2]$ and $X\in D(\cE_{\widehat f_\eta,H_{\operatorname{BH}}})$ we have
\begin{align}
 \cE_{\widehat f_\eta,H_{\operatorname{BH}}}(X)\ge C e^{-cn}\left\|X-\langle \sqrt{\sigma_\beta},X\rangle_2\sqrt{\sigma_\beta}\right\|_2^2.
 \label{BHexp:poincare}
\end{align}
Thus $\ker L_{\widehat f_\eta,H_{\operatorname{BH}}}=\C \sqrt{\sigma_\beta}$ and $\operatorname{gap}(L_{\widehat f_\eta,H_{\operatorname{BH}}})\ge C e^{-cn}$. Moreover, by gap antimonotonicity \cite[Proposition~4.2]{BeckerRouzeSalzmannToAppearcmp}, the same gap lower bound holds for $L_{\sigma_E,\widehat f_\eta,H_{\operatorname{BH}}}$ for $\sigma_E>0$.
\end{theorem}

\begin{remark}
In fact, for a total-particle-number cutoff $R\ge1$, we also have
\begin{align*}
 \operatorname{gap}\bigl(L^{(R)}_{\widehat f_\eta,H_{\operatorname{BH}}}\bigr)\ge C e^{-cn},
\end{align*}
with the same constants $C,c$, independent of $R$, where $L^{(R)}_{\widehat f,H_{\operatorname{BH}}}$ denotes the Gibbs sampler on $\Pi_{\le R}\cH$ with Hamiltonian $H_{\operatorname{BH}}^{(R)}=H_{\operatorname{BH}}|_{\Pi_{\le R}\cH}$ and compressed bare jumps $A_R^\alpha=\Pi_{\le R}A^\alpha\Pi_{\le R}$.

Furthermore, we note that Theorem~\ref{BHexp:theorem} also holds on finite simple graphs of maximum degree at most $\Delta$, with constants depending only on $\beta,U,\mu,J,\Delta$. In the proof, the neighbour bound $2D$ is replaced by $\Delta$.
\end{remark}

For the proof of Theorem~\ref{BHexp:theorem}, we introduce the auxiliary filter function for $\eta\in[0,1/2]$ by
\begin{align}
 \widetilde f_\eta(\nu)=\frac{\widehat f_\eta(\nu)}{(1+\nu^2/\Omega^2)^{(1/2-\eta)/2}},
 \qquad \Omega=\frac{2\pi}{\beta}.
 \label{BHexp:auxiliary-filter}
\end{align}
For $\eta\in[0,1/2]$, the filter $\widetilde f_\eta(\nu)$ decays as $|\nu|^{-1/2}$ as $\nu\to-\infty$, and we  moreover have $\widetilde f_{1/2}=\widehat f_{1/2}$. We introduce this family to prove the gap bound with a common decay exponent, while the choice $\Omega=2\pi/\beta$ ensures finiteness of a certain exponential moment needed for the Dirichlet-form comparison with $\widehat f_\eta$, i.e.
$\cE_{\widehat f_\eta,H_{\operatorname{BH}}}(X)\ge c\,\cE_{\widetilde f_\eta,H_{\operatorname{BH}}}(X)$,
where $c>0$ denotes a constant independent of $\eta,n$ and the inequality holds for all $X\in D(\cE_{\widehat f_\eta,H_{\operatorname{BH}}})$.

We define the self-adjoint operator $\cN$ on the space of Hilbert-Schmidt operators $\mathscr T_2(\cH)$ by
\begin{align*}
 \cN(X)=\Ntot X+X\Ntot+(n+1)X,
\end{align*}
and in particular $\Pi_p\cN(X)\Pi_q=(p+q+n+1)X_{p,q}$ where we denoted $X_{p,q}=\Pi_pX\Pi_q$ and with domain given by
\begin{align*}
 D(\cN)&=\left\{X\in\mathscr T_2(\cH):
 \sum_{p,q\ge0}(p+q+n+1)^2\left\|X_{p,q}\right\|_2^2<\infty\right\}.
\end{align*}

We start by stating the following lemma which is employed in the proof of Theorem~\ref{BHexp:theorem}. 

\begin{lemma}
\label{BHexp:sources}
Let $q\in\N$ and  
 $S_{i,q}(t):=\Pi_{q-1}S_{\widetilde f_\eta}^{i,-}(t)\Pi_q$ and $T_{i,q}(t):=\Pi_{q-1}T_{\widetilde f_\eta}^{i,-}(t)\Pi_q.$
Then we have
\begin{align}
 \int w_\beta\sum_i\big(S_{i,q}S_{i,q}^\dagger+\dot S_{i,q}\dot S_{i,q}^\dagger\big)\,dt
 &\le C_0(q+n)\id,\label{BHexp:source-upper}\\
 \int w_\beta\sum_i\big(T_{i,q}^\dagger T_{i,q}+\dot T_{i,q}^\dagger\dot T_{i,q}\big)\,dt
 &\ge c_0q(1+q/n)\id,\label{BHexp:source-lower}
\end{align}
and furthermore, for $X\in\mathscr T_2(\cH),$
\begin{align}
 \left\|\cN^{-2}\big(\partial_{t,\widetilde f_\eta}(X)\big)\right\|_w+\left\|\cN^{-2}\big(\dot\partial_{t,\widetilde f_\eta}(X)\big)\right\|_w+\left\|\cN^{-2}\big(\ddot\partial_{t,\widetilde f_\eta}(X)\big)\right\|_w&\le C_1\left\|X\right\|_2,
 \label{BHexp:derivative-bound}
\end{align}
where $c_0,C_0,C_1>0$ denote constants depending only on $\beta,U,\mu,J,D$ and independent of $\eta,n,q$.

\end{lemma}

With these estimates, we can now prove Theorem~\ref{BHexp:theorem}. After that we give the proof of Lemma~\ref{BHexp:sources}. 

\begin{proof}[Proof of Theorem~\ref{BHexp:theorem}]
We first fix a total-particle-number cutoff $R\ge1$ and work on $\Pi_{\le R}\cH$ with the Hamiltonian $H_{\operatorname{BH}}^{(R)}=H_{\operatorname{BH}}|_{\Pi_{\le R}\cH}$ and compressed bare jumps $A_R^\alpha=\Pi_{\le R}A^\alpha\Pi_{\le R}$. The corresponding normalised Gibbs state is
\begin{align*}
 \sigma_{\beta,R}
 =\frac{e^{-\beta H_{\operatorname{BH}}^{(R)}}}{\Tr(e^{-\beta H_{\operatorname{BH}}^{(R)}})}
 =\frac{\Pi_{\le R}\sigma_\beta\Pi_{\le R}}{\Tr(\Pi_{\le R}\sigma_\beta)}.
\end{align*}
We denote the corresponding derivations and Dirichlet form of $L^{(R)}_{\widehat f,H_{\operatorname{BH}}}$ by $\partial_{t,\widehat f,R}^\alpha$ and $\cE_{\widehat f,R}$. Since the cutoff preserves particle sectors, for $X=\Pi_{\le R}X\Pi_{\le R}$ we have
\begin{align*}
 A_{\nu,R}^\alpha=\Pi_{\le R}A_\nu^\alpha\Pi_{\le R},
 \qquad \partial_{t,\widehat f,R}^\alpha(X)=\Pi_{\le R}\partial_{t,\widehat f}^\alpha(X)\Pi_{\le R},
 \qquad \cE_{\widehat f,R}(X)=\left\|\partial_{t,\widehat f,R}(X)\right\|_w^2.
\end{align*}
Fix $\eta\in[0,1/2]$. We now show, for $X\in\mathscr T_2(\Pi_{\le R}\cH)$ and uniformly in $\eta,R$, the lower bounds
\begin{align}
 \cE_{\widehat f_\eta,R}(X)
 \ge\frac{\sqrt3}{8}\,\cE_{\widetilde f_\eta,R}(X)
 \ge C e^{-cn}\left\|X-\langle \sqrt{\sigma_{\beta,R}},X\rangle_2\sqrt{\sigma_{\beta,R}}\right\|_2^2.
 \label{BHexp:metropolis-comparison}
\end{align}
The first inequality is immediate for $\eta=1/2$, since the filters and hence their Dirichlet forms coincide. To prove it for $0\le\eta<1/2$, we define, for $\theta>0$, the function $k_\theta:\R\to[0,\infty)$ by 
\begin{align*}
 k_\theta(s)
 =\int_0^\infty
 \underbrace{\frac{u^{\theta/2-1}e^{-u}}{\Gamma(\theta/2)}}_{\text{probability density in }u}
 \quad
 \underbrace{\frac{\Omega}{\sqrt{4\pi u}}e^{-\Omega^2s^2/(4u)}}_{\text{Gaussian probability density in }s}
 \,du,
\end{align*}
for $s\neq 0$ and $k_\theta(0)=0$  and where $\Gamma$ denotes the Gamma function. Note that $k_\theta$ is an even probability density, since it is a mixture of even Gaussian probability densities. Denoting $m_\theta(\nu)=(1+\nu^2/\Omega^2)^{-\theta/2}$, we obtain directly from the Gaussian Fourier transform and the Gamma integral
\begin{align*}
 \int_{\R}e^{i\nu s}\,k_\theta(s)\,ds
 =\frac1{\Gamma(\theta/2)}\int_0^\infty
 u^{\theta/2-1}e^{-(1+\nu^2/\Omega^2)u}\,du
 =m_\theta(\nu).
\end{align*}
Now choosing $\theta=1/2-\eta>0$ and using $\widetilde f_\eta=\widehat f_\eta m_\theta$ and evenness of $m_\theta$, we obtain for $\alpha\in\mathcal A$
\begin{align*}
 \partial_{t,\widetilde f_\eta,R}^\alpha(X)
 &=\sum_\nu e^{it\nu}m_\theta(\nu)
 \big(\widehat f_\eta(\nu)A_{\nu,R}^\alpha X-\widehat f_\eta(-\nu)XA_{\nu,R}^\alpha\big)\\
 &=\int_{\R}\sum_\nu e^{i(t+s)\nu}
 \big(\widehat f_\eta(\nu)A_{\nu,R}^\alpha X-\widehat f_\eta(-\nu)XA_{\nu,R}^\alpha\big)\,k_\theta(s)\,ds
 =\int_{\R}\partial_{t+s,\widehat f_\eta,R}^\alpha(X)\,k_\theta(s)\,ds.
\end{align*}
Note that the sum over the Bohr frequencies $\nu$ is finite as this is for the finite dimensional Hamiltonian $H_{\operatorname{BH}}^{(R)}$ and, hence, can be interchanged with the integral. For each $\alpha$, Minkowski's integral inequality in $L^2(\R,w_\beta(t)\,dt;\mathscr T_2)$ 
\begin{align}
 \left(\int_{\R}w_\beta(t)\left\|\partial_{t,\widetilde f_\eta,R}^\alpha(X)\right\|_2^2\,dt\right)^{1/2}
 &\le\int_{\R}\left(\int_{\R}w_\beta(t)\left\|\partial_{t+s,\widehat f_\eta,R}^\alpha(X)\right\|_2^2\,dt\right)^{1/2}k_\theta(s)\,ds\nn\\
 &=\int_{\R}\left(\int_{\R}w_\beta(r-s)\left\|\partial_{r,\widehat f_\eta,R}^\alpha(X)\right\|_2^2\,dr\right)^{1/2}k_\theta(s)\,ds\nn\\
 &\le\left(\int_{\R}e^{\pi|s|/\beta}\,k_\theta(s)\,ds\right)
 \left(\int_{\R}w_\beta(r)\left\|\partial_{r,\widehat f_\eta,R}^\alpha(X)\right\|_2^2\,dr\right)^{1/2},
 \label{BHexp:weighted-filter-comparison}
\end{align}
where we used the change of variables $r=t+s$ and $w_\beta(r-s)\le e^{2\pi|s|/\beta}w_\beta(r)$ in the second and third line respectively. 
Since $\Omega=2\pi/\beta$, we have $\int_{\R}e^{\pi|s|/\beta}k_\theta(s)\,ds\le2(3/4)^{-\theta/2}$, which, after substitution into \eqref{BHexp:weighted-filter-comparison}, squaring and summing over $\alpha$, yields
\begin{align*}
 \cE_{\widetilde f_\eta,R}(X)
 \le4(3/4)^{-\theta}\cE_{\widehat f_\eta,R}(X)
 \le\frac8{\sqrt3}\cE_{\widehat f_\eta,R}(X),
\end{align*}
where the last inequality uses $0<\theta\le1/2$. This proves the first inequality in \eqref{BHexp:metropolis-comparison} for all $\eta\in[0,1/2]$.

We continue with the proof of the second inequality in \eqref{BHexp:metropolis-comparison}. We define
\begin{align*}
 \mathcal Q_R(X)=\left\|\cN^{-1}\big(\partial_{t,\widetilde f_\eta,R}(X)\big)\right\|_w^2+\left\|\cN^{-1}\big(\dot\partial_{t,\widetilde f_\eta,R}(X)\big)\right\|_w^2.
\end{align*}
Write $X_{p,q}=\Pi_pX\Pi_q$. Under the assumption $X_{0,0}=0$, we first prove
\begin{align*}
 \mathcal Q_R(X)\ge\kappa_n\left\|X\right\|_2^2,
\end{align*}
where $\kappa_n>0$ denotes an $n$-dependent constant that we determine explicitly below. For that, we derive a recurrence between neighbouring particle blocks, initially using only the annihilation labels $\alpha=(i,-)$, so the families below are indexed by $i=1,\ldots,n$. For $0\le q\le p\le R$, $p\ge1$, denote $z_{p,q}=\left\|X_{p,q}\right\|_2$. In the notation of Lemma~\ref{BHexp:sources}, the $(p-1,q)$ blocks of the annihilation derivations are
\begin{align}
 \Pi_{p-1}\partial_{t,\widetilde f_\eta,R}^{i,-}(X)\Pi_q
 =T_{i,p}(t)X_{p,q}-X_{p-1,q-1}S_{i,q}(t),
 \label{BHexp:annihilation-block-identity}
\end{align}
where we used that the annihilation blocks $S_{i,q},T_{i,q}$ for $1\le q\le R$ are unchanged by the cutoff at total particle number $R$. The same identity holds after time differentiation, with the term containing $S_{i,q}$ absent when $q=0$.

To apply Lemma~\ref{BHexp:sources}, we first use cyclicity of the trace and \eqref{BHexp:source-lower} with particle number $p$ to obtain
\begin{align}
 \left\|\big(T_{i,p}X_{p,q}\big)_i\right\|_w^2+\left\|\big(\dot T_{i,p}X_{p,q}\big)_i\right\|_w^2=\Tr\left[X_{p,q}X_{p,q}^\dagger\int_{\R}w_\beta(t)\sum_i\big(T_{i,p}^\dagger T_{i,p}+\dot T_{i,p}^\dagger\dot T_{i,p}\big)\,dt\right]
 \ge c_0p(1+p/n)z_{p,q}^2.
 \label{BHexp:block-T-lower}
\end{align}
Similarly, for $q\ge1$, \eqref{BHexp:source-upper} gives
\begin{align}
 \left\|\big(X_{p-1,q-1}S_{i,q}\big)_i\right\|_w^2+\left\|\big(X_{p-1,q-1}\dot S_{i,q}\big)_i\right\|_w^2
 &=\Tr\left[X_{p-1,q-1}^\dagger X_{p-1,q-1}\int_{\R}w_\beta(t)\sum_i\big(S_{i,q}S_{i,q}^\dagger+\dot S_{i,q}\dot S_{i,q}^\dagger\big)\,dt\right]
 \nn\\&\le C_0(q+n)z_{p-1,q-1}^2.
 \label{BHexp:block-S-upper}
\end{align}
On the other hand, using \eqref{BHexp:annihilation-block-identity} and applying the triangle inequality gives
\begin{align*}
 \left\|\big(T_{i,p}X_{p,q}\big)_i\right\|_w
 &\le\left\|\big(\Pi_{p-1}\partial_{t,\widetilde f_\eta,R}^{i,-}(X)\Pi_q\big)_i\right\|_w
 +\left\|\big(X_{p-1,q-1}S_{i,q}\big)_i\right\|_w,\\
 \left\|\big(\dot T_{i,p}X_{p,q}\big)_i\right\|_w
 &\le\left\|\big(\Pi_{p-1}\dot\partial_{t,\widetilde f_\eta,R}^{i,-}(X)\Pi_q\big)_i\right\|_w
 +\left\|\big(X_{p-1,q-1}\dot S_{i,q}\big)_i\right\|_w.
\end{align*}
Combining these inequalities with the triangle inequality for the Euclidean norm on $\R^2$ gives
\begin{align}
 &\left(\left\|\big(T_{i,p}X_{p,q}\big)_i\right\|_w^2+\left\|\big(\dot T_{i,p}X_{p,q}\big)_i\right\|_w^2\right)^{1/2}\nn\\
 &\quad\le(p+q+n)e_{p,q}
 +\left(\left\|\big(X_{p-1,q-1}S_{i,q}\big)_i\right\|_w^2+\left\|\big(X_{p-1,q-1}\dot S_{i,q}\big)_i\right\|_w^2\right)^{1/2},
 \label{BHexp:block-triangle}
\end{align}
where we defined $e_{p,q}\ge0$ by
\begin{align*}
 e_{p,q}^{\,2}
 &:=\left\|\big(\cN^{-1}\big(\Pi_{p-1}\partial_{t,\widetilde f_\eta,R}^{i,-}(X)\Pi_q\big)\big)_i\right\|_w^2
 +\left\|\big(\cN^{-1}\big(\Pi_{p-1}\dot\partial_{t,\widetilde f_\eta,R}^{i,-}(X)\Pi_q\big)\big)_i\right\|_w^2\\
 &=\frac{\left\|\big(\Pi_{p-1}\partial_{t,\widetilde f_\eta,R}^{i,-}(X)\Pi_q\big)_i\right\|_w^2
 +\left\|\big(\Pi_{p-1}\dot\partial_{t,\widetilde f_\eta,R}^{i,-}(X)\Pi_q\big)_i\right\|_w^2}{(p+q+n)^2}.
\end{align*}
Combining the joint triangle inequality \eqref{BHexp:block-triangle} with the lower bound \eqref{BHexp:block-T-lower} and the upper bound \eqref{BHexp:block-S-upper} gives
\begin{align*}
 \sqrt{c_0p(p+n)/n}\,z_{p,q}
 \le (p+q+n)e_{p,q}+\sqrt{C_0(q+n)}\,z_{p-1,q-1},
\end{align*}
where $1\le q\le p\le R$, and for $q=0$ the same bound holds without the term $\sqrt{C_0(q+n)}\,z_{p-1,q-1}$, using \eqref{BHexp:block-T-lower} and the fact that the $S$-term in \eqref{BHexp:annihilation-block-identity} is absent. Dividing by the coefficient on the left and using $(p+q+n)/\sqrt{c_0p(p+n)/n}\le C(n+1)$ and $C_0(q+n)/(c_0p(p+n)/n)\le C_0n/(c_0p)$, which follow from $q\le p$, gives
\begin{align}
 z_{p,q}\le C(n+1)e_{p,q}+K\sqrt{n/p}\,z_{p-1,q-1},
 \label{BHexp:recurrence}
\end{align}
where we denoted $K=\sqrt{C_0/c_0}$.
Defining $e_{0,0}=0$ and using $z_{0,0}=0$, which follows from $X_{0,0}=0$, we now iterate \eqref{BHexp:recurrence} to obtain, for $0\le q\le p\le R$,
\begin{align*}
 z_{p,q}\le C(n+1)\sum_{j=0}^{q}(K\sqrt n)^j\sqrt{\frac{(p-j)!}{p!}}\,e_{p-j,q-j},
\end{align*}
where the iteration ends at $(p-q,0)$ and, for $p>q$, we used $z_{p-q,0}\le C(n+1)e_{p-q,0}$ by the $q=0$ case of \eqref{BHexp:recurrence}. Using the triangle inequality of the Euclidean norm and $p!/(p-j)!\ge j!$ for $0\le j\le p$ gives
\begin{align*}
 \left(\sum_{0\le q\le p\le R}z_{p,q}^2\right)^{1/2}
 &\le C(n+1)\sum_{j=0}^{R}(K\sqrt n)^j
 \left(\sum_{j\le q\le p\le R}\frac{(p-j)!}{p!}\,e_{p-j,q-j}^2\right)^{1/2}\\
 &\le C(n+1)\sum_{j=0}^{R}\frac{(K\sqrt n)^j}{\sqrt{j!}}
 \left(\sum_{0\le q\le p\le R}e_{p,q}^2\right)^{1/2}\\
 &\le\sqrt2 C(n+1)e^{K^2n}\left(\sum_{0\le q\le p\le R}e_{p,q}^2\right)^{1/2},
\end{align*}
where in the second inequality we also changed indices from $(p,q)$ to $(p-j,q-j)$, and in the last inequality we used Cauchy-Schwarz to obtain
\begin{align*}
 \sum_{j=0}^{R}\frac{(K\sqrt n)^j}{\sqrt{j!}}
 \le\left(\sum_{j=0}^{\infty}\frac{(2K^2n)^j}{j!}\right)^{1/2}
 \left(\sum_{j=0}^{\infty}2^{-j}\right)^{1/2}
 =\sqrt2\,e^{K^2n}.
\end{align*}
We apply the same argument to $X^\dagger$ on $p>q$, denoting by $\widetilde e_{p,q}$ the quantities obtained from the definition of $e_{p,q}$ with $X$ replaced by $X^\dagger$. Using $(\partial_{t,\widetilde f_\eta,R}^{i,-}(X^\dagger))^\dagger=-\partial_{t,\widetilde f_\eta,R}^{i,+}(X)$, we obtain
\begin{align*}
 \left\|X\right\|_2^2
 &=\sum_{0\le q\le p\le R}\left\|X_{p,q}\right\|_2^2
   +\sum_{0\le q<p\le R}\left\|(X^\dagger)_{p,q}\right\|_2^2\\
 &\le 2C^2(n+1)^2e^{2K^2n}
 \left(\sum_{0\le q\le p\le R}e_{p,q}^2+\sum_{0\le q<p\le R}\widetilde e_{p,q}^2\right)\\
 &\le 2C^2(n+1)^2e^{2K^2n}
 \left(\left\|\cN^{-1}\big(\partial_{t,\widetilde f_\eta,R}(X)\big)\right\|_w^2
 +\left\|\cN^{-1}\big(\dot\partial_{t,\widetilde f_\eta,R}(X)\big)\right\|_w^2\right)\\
 &=2C^2(n+1)^2e^{2K^2n}\mathcal Q_R(X),
\end{align*}
where in the last inequality the sum of $e_{p,q}^2$ uses distinct $(p-1,q)$ blocks of the annihilation derivations, while the sum of $\widetilde e_{p,q}^2$ uses distinct $(q,p-1)$ blocks of the creation derivations. Thus $\mathcal Q_R(X)\ge\kappa_n\left\|X\right\|_2^2$ when $X_{0,0}=0$, with $\kappa_n=(2C^2)^{-1}(n+1)^{-2}e^{-2K^2n}$.

For general $X$, we apply this bound to $Y=X-\frac{\langle0|X|0\rangle}{\langle0|\sqrt{\sigma_{\beta,R}}|0\rangle}\sqrt{\sigma_{\beta,R}}$, which satisfies $Y_{0,0}=0$. Using $\partial_{t,\widetilde f_\eta,R}(\sqrt{\sigma_{\beta,R}})=0$ for every $t$  and then orthogonally projecting $Y$ onto $\{\sqrt{\sigma_{\beta,R}}\}^{\perp}$, gives
\begin{align}
 \mathcal Q_R(X)=\mathcal Q_R(Y)\ge\kappa_n\left\|Y\right\|_2^2
 \ge\kappa_n\left\|X-\langle\sqrt{\sigma_{\beta,R}},X\rangle_2\sqrt{\sigma_{\beta,R}}\right\|_2^2.
 \label{BHexp:aux-gap}
\end{align}
Using this we now continue to prove the desired lower bound on the Dirichlet form $\cE_{\widetilde f_\eta,R}.$ For that we denote $\ell(t)=w_\beta'(t)/w_\beta(t)=-(2\pi/\beta)\tanh(2\pi t/\beta)$, which satisfies $|\ell(t)|\le2\pi/\beta$, and obtain by self-adjointness of $\cN^{-2}$ and integration by parts that 
\begin{align*}
 \left\|\cN^{-1}\big(\dot\partial_{t,\widetilde f_\eta,R}(X)\big)\right\|_w^2
 &=\sum_\alpha\int_{\R}w_\beta(t)
 \langle \cN^{-2}\big(\dot\partial_{t,\widetilde f_\eta,R}^\alpha(X)\big),\dot\partial_{t,\widetilde f_\eta,R}^\alpha(X)\rangle_2\,dt\\
 &=-\sum_\alpha\int_{\R}w_\beta(t)
 \left\langle \cN^{-2}\big(\ddot\partial_{t,\widetilde f_\eta,R}^\alpha(X)+\ell(t)\dot\partial_{t,\widetilde f_\eta,R}^\alpha(X)\big),\partial_{t,\widetilde f_\eta,R}^\alpha(X)\right\rangle_2\,dt.
\end{align*}
The boundary term vanishes because $w_\beta(t)$ tends to zero exponentially, while $\partial_{t,\widetilde f_\eta,R}^\alpha(X)$ and its time derivatives are bounded finite trigonometric sums on $\Pi_{\le R}\cH$, where $R<\infty$. Substituting into $\mathcal Q_R$ and applying Cauchy-Schwarz yields
\begin{align*}
 \mathcal Q_R(X)
 &=\left\langle \cN^{-2}\big(\partial_{t,\widetilde f_\eta,R}(X)-\ddot\partial_{t,\widetilde f_\eta,R}(X)-\ell\dot\partial_{t,\widetilde f_\eta,R}(X)\big),\partial_{t,\widetilde f_\eta,R}(X)\right\rangle_w\\
 &\le\left(\left\|\cN^{-2}\big(\partial_{t,\widetilde f_\eta,R}(X)\big)\right\|_w
 +\left\|\cN^{-2}\big(\ddot\partial_{t,\widetilde f_\eta,R}(X)\big)\right\|_w
 +\frac{2\pi}{\beta}\left\|\cN^{-2}\big(\dot\partial_{t,\widetilde f_\eta,R}(X)\big)\right\|_w\right)\left\|\partial_{t,\widetilde f_\eta,R}(X)\right\|_w\\
 &\le C\left\|X\right\|_2\cE_{\widetilde f_\eta,R}(X)^{1/2},
\end{align*}
where $C\ge0$ denotes some constant independent of $\eta,n,R$, and we used $\left\|\partial_{t,\widetilde f_\eta,R}(X)\right\|_w^2=\cE_{\widetilde f_\eta,R}(X)$ together with \eqref{BHexp:derivative-bound} in Lemma~\ref{BHexp:sources}. The latter bound holds with the same constant for the compressed derivations and their time derivatives, since $\cN$ commutes with compression and compression does not increase the Hilbert-Schmidt norm.

Now define $X_\perp=X-\langle\sqrt{\sigma_{\beta,R}},X\rangle_2\sqrt{\sigma_{\beta,R}}$. Both forms are unchanged by this subtraction, so \eqref{BHexp:aux-gap} and the preceding estimate give
\begin{align*}
 \kappa_n\left\|X_\perp\right\|_2^2
 \le\mathcal Q_R(X_\perp)
 \le C\left\|X_\perp\right\|_2\cE_{\widetilde f_\eta,R}(X)^{1/2}.
\end{align*}
For $X_\perp\ne0$, dividing by $C\left\|X_\perp\right\|_2$ and squaring proves the following bound, which is also immediate when $X_\perp=0$:
\begin{align}
 \cE_{\widetilde f_\eta,R}(X)\ge c(n+1)^{-4}e^{-4K^2n}\left\|X-\langle\sqrt{\sigma_{\beta,R}},X\rangle_2\sqrt{\sigma_{\beta,R}}\right\|_2^2.
 \label{BHexp:finite-gap}
\end{align}
Absorbing the polynomial factor into the exponential proves the second inequality in \eqref{BHexp:metropolis-comparison}, and hence \eqref{BHexp:poincare} at finite cutoff.

Lastly, we let $R\to\infty$ at fixed $n$. For $X=\Pi_{\le M}X\Pi_{\le M}$ and $R\ge M+1$, we have $\cE_{\widehat f_\eta,R}(X)=\cE_{\widehat f_\eta,H_{\operatorname{BH}}}(X)$, since the jumps change particle number by one and all neighbouring sectors are retained. Using $\sqrt{\sigma_{\beta,R}}=\Pi_{\le R}\sqrt{\sigma_\beta}/\left\|\Pi_{\le R}\sqrt{\sigma_\beta}\right\|_2\to\sqrt{\sigma_\beta}$ in Hilbert-Schmidt norm and the independence of the constants $C$ and $c$ from $R$, the finite-cutoff bound therefore gives
\begin{align*}
 \cE_{\widehat f_\eta,H_{\operatorname{BH}}}(X)
 =\lim_{R\to\infty}\cE_{\widehat f_\eta,R}(X)
 \ge C e^{-cn}\left\|X-\langle\sqrt{\sigma_\beta},X\rangle_2\sqrt{\sigma_\beta}\right\|_2^2.
\end{align*}
This proves \eqref{BHexp:poincare} on $\mathscr F=\bigcup_M\Pi_{\le M}\mathscr T_2(\cH)\Pi_{\le M}$. Since the derivations annihilate $\sqrt{\sigma_\beta}$, the bound also holds on $\mathscr F+\C\sqrt{\sigma_\beta}$ and extends to the full form domain by density of this subspace in the form norm.
\end{proof}

We now finish the section by giving the proof of Lemma~\ref{BHexp:sources}.

\begin{proof}[Proof of Lemma~\ref{BHexp:sources}]
We write $\operatorname{ad}_{H_{\operatorname{BH}}}$ for the self-adjoint closure on Hilbert--Schmidt space of the commutator $B\mapsto[H_{\operatorname{BH}},B]$, initially defined on operators supported in finitely many particle sectors. Since $H_{\operatorname{BH}}$ preserves particle number, its restriction to $B:\cH_q\to\cH_p$ is $\operatorname{ad}_{H_{\operatorname{BH}}}(B)=H_{\operatorname{BH},p}B-BH_{\operatorname{BH},q}$.
Denoting $a_{i,q}:=a_i\big|_{\cH_q}= \Pi_{q-1}a_i\Pi_q$ and using that $e^{it\operatorname{ad}_{H_{\operatorname{BH}}}}(B)=e^{itH_{\operatorname{BH}}}Be^{-itH_{\operatorname{BH}}}$ we see, for integers $r\ge0$,
\begin{align}
 \frac{d^r}{dt^r}T_{i,q}(t)&=\bigl(i\operatorname{ad}_{H_{\operatorname{BH}}}\bigr)^r e^{it\operatorname{ad}_{H_{\operatorname{BH}}}}\widetilde f_\eta(\operatorname{ad}_{H_{\operatorname{BH}}})(a_{i,q}),\qquad
 \frac{d^r}{dt^r}S_{i,q}(t)=\bigl(i\operatorname{ad}_{H_{\operatorname{BH}}}\bigr)^r e^{it\operatorname{ad}_{H_{\operatorname{BH}}}}\widetilde f_\eta(-\operatorname{ad}_{H_{\operatorname{BH}}})(a_{i,q}).
 \label{BHexp:functional-calculus}
\end{align}
For fixed particle numbers $p,q$, let $M=(M_i)_{i=1}^n$ be a family of linear operators $M_i:\cH_q\to\cH_p$, where $i$ labels the $n$ modes, one at each lattice site. We define the column and row norms by
\begin{align*}
 \left\|M\right\|_{\mathrm{col}}^2=\left\|\sum_{i=1}^nM_i^\dagger M_i\right\|,
 \qquad \left\|M\right\|_{\mathrm{row}}^2=\left\|\sum_{i=1}^nM_iM_i^\dagger\right\|,
 \qquad \left\|M\right\|_{\mathrm{rc}}=\max\{\left\|M\right\|_{\mathrm{col}},\left\|M\right\|_{\mathrm{row}}\}.
\end{align*}
Write $R(z)=(z-\operatorname{ad}_{H_{\operatorname{BH}}})^{-1}$ for the resolvent on Hilbert--Schmidt space. Using the Laplace representation of the resolvent for $r>0$,
\begin{align*}
 R(\pm ir)(B)=\mp i\int_0^\infty e^{-rt}e^{\mp itH_{\operatorname{BH}}}B e^{\pm itH_{\operatorname{BH}}}\,dt.
\end{align*}
this gives for an operator family $M=(M_i)_{i=1}^n$
\begin{align*}
 \left\|\bigl(R(\pm ir)(M_i)\bigr)_{i=1}^n\right\|_{\mathrm{rc}}
 &\le\int_0^\infty e^{-rt}
 \left\|\bigl(e^{\mp itH_{\operatorname{BH}}}M_i e^{\pm itH_{\operatorname{BH}}}\bigr)_{i=1}^n\right\|_{\mathrm{rc}}\,dt
 =\int_0^\infty e^{-rt}\left\|M\right\|_{\mathrm{rc}}\,dt
 =r^{-1}\left\|M\right\|_{\mathrm{rc}}.
\end{align*}
Combining this with the resolvent identity $\operatorname{ad}_{H_{\operatorname{BH}}}R(z)=zR(z)-\id$ yields
\begin{align}
 \left\|R(\pm ir)\right\|_{\mathrm{rc}\to\mathrm{rc}}\le r^{-1},
 \qquad \left\|\operatorname{ad}_{H_{\operatorname{BH}}} R(\pm ir)\right\|_{\mathrm{rc}\to\mathrm{rc}}\le2.
 \label{BHexp:resolvent}
\end{align}

To apply these resolvent bounds to $T_{i,q},S_{i,q}$ and their time derivatives, we first derive a Cauchy integral representation for the two scalar filter functions.
For $g(s)=\widetilde f_\eta(s)$ or $\widetilde f_\eta(-s)$, use the principal branches of the square root and fractional powers in \eqref{BHexp:filter}--\eqref{BHexp:auxiliary-filter}. The same formula then defines an analytic function on
$\C\setminus\bigl(i[1/\beta,\infty)\cup(-i[1/\beta,\infty))\bigr)$,
agreeing with the original filter on the real axis. Indeed, $1+\beta^2z^2$ lies on the nonpositive real axis only when $z=iy$ with $|y|\ge1/\beta$. For $1+z^2/\Omega^2$, the corresponding rays start at $\pm i\Omega$ and are already excluded since $\Omega=2\pi/\beta>1/\beta$.

Since $\operatorname{Re}\sqrt{1+\beta^2z^2}\ge\beta|\operatorname{Re}z|$, Young's inequality gives
\begin{align*}
 |g(z)|
 &\le\bigl(|1+\beta^2z^2|^{-1/4}\bigr)^{2\eta}
       \bigl(|1+z^2/\Omega^2|^{-1/4}\bigr)^{1-2\eta}\\
 &\le 2\eta\,|1+\beta^2z^2|^{-1/4}
       +(1-2\eta)\,|1+z^2/\Omega^2|^{-1/4}
 \le |1+\beta^2z^2|^{-1/4}+|1+z^2/\Omega^2|^{-1/4}.
\end{align*}
For $r>1/\beta$, $r\ne\Omega$, define the boundary-jump densities
\begin{align*}
 \zeta_+(r)=\frac1{2\pi}\lim_{\delta\downarrow0}
 \bigl[g(-\delta+ir)-g(\delta+ir)\bigr],\qquad
 \zeta_-(r)=\frac1{2\pi}\lim_{\delta\downarrow0}
 \bigl[g(-\delta-ir)-g(\delta-ir)\bigr].
\end{align*}
Taking boundary values and applying the triangle inequality yields, uniformly in $\eta$,
\begin{align}
 |\zeta_+(r)|\,,\, |\zeta_-(r)|\le\frac1\pi\bigl(|1-\beta^2r^2|^{-1/4}+|1-r^2/\Omega^2|^{-1/4}\bigr).
 \label{BHexp:zeta-bound}
\end{align}
Cauchy's formula around the cuts\footnote{To see this define first
\[
 \begin{aligned}
 \mathcal C&=i[1/\beta,\infty)\cup(-i[1/\beta,\infty)),\\
 \mathcal D_{L,\delta}
 &=\left\{z\in\C:|z|<L,\ \inf_{w\in\mathcal C}|z-w|>\delta\right\},\\
 \Gamma_{L,\delta}&=\partial\mathcal D_{L,\delta}
 \quad\text{with positive orientation}.
 \end{aligned}
\]
For $s\in\mathbb{R}$ with $L>|s|$ and $0<\delta<1/\beta$, we have $s\in\mathcal D_{L,\delta}$ and $g$ is analytic on a neighbourhood of $\overline{\mathcal D_{L,\delta}}$. Hence
\[
 g(s)=\frac1{2\pi i}\oint_{\Gamma_{L,\delta}}\frac{g(z)}{z-s}\,dz.
\]
As $L\to\infty$ and $\delta\downarrow0$, the outer-circle contribution is $\mathcal O(L^{-1/2})$, and the contour portions within distance $2\delta$ of $i/\beta,-i/\beta,i\Omega,-i\Omega$ contribute at most $\mathcal O(\delta^{3/4})$, uniformly in $\eta$. Only the two sides of each cut remain, traversed upwards on the left and downwards on the right. Parametrising the upper cut by $z=ir$ and the lower cut by $z=-ir$, with $r\ge1/\beta$, gives \eqref{BHexp:slit}.} then gives, for $s\in\R$,
\begin{align}
 g(s)&=\int_{1/\beta}^\infty\frac{\zeta_+(r)}{ir-s}\,dr
 +\int_{1/\beta}^\infty\frac{\zeta_-(r)}{-ir-s}\,dr.
 \label{BHexp:slit}
\end{align}
Note that by \eqref{BHexp:zeta-bound}, the two integrals in \eqref{BHexp:slit} are absolutely convergent.

We now prove \eqref{BHexp:source-upper}. With $g(s)=\widetilde f_\eta(-s)$, \eqref{BHexp:functional-calculus} gives, for $j=0,1$,
\begin{align*}
 \left.\frac{d^j}{dt^j}S_{i,q}(t)\right|_{t=0}
 =i^j\operatorname{ad}_{H_{\operatorname{BH}}}^jg(\operatorname{ad}_{H_{\operatorname{BH}}})(a_{i,q}).
\end{align*}
Writing $b_i=(\mu-UN_i)|_{\cH_{q-1}}$, we add and subtract $b_i^jg(b_i)a_{i,q}$ and apply the triangle inequality:
\begin{align}
 \left\|\left(\left.\frac{d^j}{dt^j}S_{i,q}(t)\right|_{t=0}\right)_i\right\|_{\mathrm{row}}
 &=\left\|\bigl(\operatorname{ad}_{H_{\operatorname{BH}}}^jg(\operatorname{ad}_{H_{\operatorname{BH}}})(a_{i,q})\bigr)_i\right\|_{\mathrm{row}}\nn\\
 &\le\left\|\bigl(b_i^jg(b_i)a_{i,q}\bigr)_i\right\|_{\mathrm{row}}
 +\left\|\bigl(\operatorname{ad}_{H_{\operatorname{BH}}}^jg(\operatorname{ad}_{H_{\operatorname{BH}}})(a_{i,q})-b_i^jg(b_i)a_{i,q}\bigr)_i\right\|_{\mathrm{row}}.
 \label{BHexp:source-upper-triangle}
\end{align}
We bound the two terms on the right-hand side of \eqref{BHexp:source-upper-triangle} separately, for both $j=0$ and $j=1$. For the first term, on $\cH_{q-1}$,
\begin{align}
 \sum_i\bigl(b_i^jg(b_i)a_{i,q}\bigr)\bigl(b_i^jg(b_i)a_{i,q}\bigr)^\dagger
 &=\sum_i(N_i+1)b_i^{2j}|g(b_i)|^2
 \le Cn\id,
 \label{BHexp:atomic-row-bound}
\end{align}
where we used the decay of $g$ at negative frequencies in the form $\sup_{s\in\N_0}(s+1)\bigl(1+(\mu-Us)^2\bigr)|g(\mu-Us)|^2\le C$, with $C\ge0$ denoting some constant depending only on $\beta,U,\mu$ and independent of $\eta\in[0,1/2]$. Below, $C$ may change from one estimate to the next and depend also on $J,D$, but remains independent of $\eta,n,q$.
Taking square roots in \eqref{BHexp:atomic-row-bound} therefore bounds the first term on the right-hand side of \eqref{BHexp:source-upper-triangle} by $C\sqrt n$.

We now bound the second term on the right-hand side of \eqref{BHexp:source-upper-triangle}, using the Cauchy representation \eqref{BHexp:slit} to compare the resolvents. We define $B_i(z)=(z-b_i)^{-1}$. We have
\begin{align}
 \operatorname{ad}_{H_{\operatorname{BH}}}(a_{i,q})
 =b_i a_{i,q}+J\sum_{j\sim i}a_{j,q},
 \label{BHexp:annihilation-commutator}
\end{align}
where $j\sim i$ means that $j$ and $i$ are nearest neighbours.
Combining \eqref{BHexp:annihilation-commutator} with $[H_{\operatorname{BH},q-1},b_i]=UJ\sum_{j\sim i}(a_{i,q}a_{j,q}^\dagger-a_{j,q}a_{i,q}^\dagger)$ and $[H_{\operatorname{BH},q-1},B_i(z)]=B_i(z)[H_{\operatorname{BH},q-1},b_i]B_i(z)$ gives
\begin{align}
 &R(z)(a_{i,q})-B_i(z)a_{i,q}=J R(z)\left(B_i(z)\sum_{j\sim i}a_{j,q}+U B_i(z)\sum_{j\sim i}(a_{i,q}a_{j,q}^\dagger-a_{j,q}a_{i,q}^\dagger)B_i(z)a_{i,q}\right).
 \label{BHexp:ordered}
\end{align}
We estimate the two families inside the parentheses at $z=\pm ir$, $r\ge1/\beta$. For the first, adjunction, neighbour Cauchy-Schwarz, and $N_i a_{j,q}^\dagger=a_{j,q}^\dagger N_i$ for $j\sim i$ give
\begin{align*}
 \left\|\left(B_i(z)\sum_{j\sim i}a_{j,q}\right)_i\right\|_{\mathrm{row}}
 \le C\sqrt{q+n}/r.
\end{align*}
For the second, define $\delta_s=z-\mu+Us$. On an output occupation vector with $N_i=s$ and $N_j=m$, the adjoints of the two products
$B_i(z)a_{i,q}a_{j,q}^\dagger B_i(z)a_{i,q}$ and
$B_i(z)a_{j,q}a_{i,q}^\dagger B_i(z)a_{i,q}$ have respective coefficient magnitudes
\begin{align*}
 \frac{\sqrt{(s+1)(s+2)m}}{|\delta_s\delta_{s+1}|}
 \le\frac{C\sqrt m}{r},
 \qquad
 \frac{s\sqrt{m+1}}{|\delta_s\delta_{s-1}|}
 \le\frac{C\sqrt{m+1}}{r},
\end{align*}
where the first inequality holds for $s\ge0$ and the second for $s\ge1$. The second coefficient is zero at $s=0$. These inequalities follow from $|\delta_s|\ge r$ and $(s+2)/|\delta_s|\le U^{-1}+\beta(2+|\mu|/U)$. Each product sends distinct occupation vectors to distinct occupation vectors. Summing the squared coefficients and using Cauchy-Schwarz over neighbours therefore gives
\begin{align*}
 \left\|\left(B_i(z)\sum_{j\sim i}(a_{i,q}a_{j,q}^\dagger-a_{j,q}a_{i,q}^\dagger)B_i(z)a_{i,q}\right)_i\right\|_{\mathrm{row}}
 \le C\sqrt{q+n}/r.
\end{align*}
Applying the Laplace representation to the outer resolvent in \eqref{BHexp:ordered} contributes another factor $r^{-1}$, so
\begin{align}
 \left\|(R(\pm ir)(a_{i,q})-B_i(\pm ir)a_{i,q})_i\right\|_{\mathrm{row}}
 \le C\sqrt{q+n}/r^2,
 \label{BHexp:resolvent-difference}
\end{align}
where the constants in these estimates depend only on the fixed parameters $\beta,U,\mu,J,D$.

We now insert \eqref{BHexp:resolvent-difference} into \eqref{BHexp:slit}. For the first time derivative, the resolvent identities give
\begin{align}
 \operatorname{ad}_{H_{\operatorname{BH}}}R(z)(a_{i,q})-b_iB_i(z)a_{i,q}
 &=\big[zR(z)(a_{i,q})-a_{i,q}\big]-\big[zB_i(z)a_{i,q}-a_{i,q}\big]\nn\\
 &=z\big[R(z)(a_{i,q})-B_i(z)a_{i,q}\big].
 \label{BHexp:resolvent-cancellation}
\end{align}
Consequently, for $j=0,1$,
\begin{align*}
 \left\|\bigl(\operatorname{ad}_{H_{\operatorname{BH}}}^jg(\operatorname{ad}_{H_{\operatorname{BH}}})(a_{i,q})-b_i^jg(b_i)a_{i,q}\bigr)_i\right\|_{\mathrm{row}}
 &\le C\sqrt{q+n}\int_{1/\beta}^\infty
 \bigl(|\zeta_+(r)|+|\zeta_-(r)|\bigr)r^{j-2}\,dr
 \le C\sqrt{q+n},
\end{align*}
where $\zeta_\pm$ are the boundary-jump densities of $g(s)=\widetilde f_\eta(-s)$, and the last bound follows from \eqref{BHexp:zeta-bound}. This bounds the second term on the right-hand side of \eqref{BHexp:source-upper-triangle}. Substituting this estimate and \eqref{BHexp:atomic-row-bound} into \eqref{BHexp:source-upper-triangle} gives, for $j=0,1$,
\begin{align*}
 \left\|\left(\left.\frac{d^j}{dt^j}S_{i,q}(t)\right|_{t=0}\right)_i\right\|_{\mathrm{row}}
 \le C\sqrt n+C\sqrt{q+n}\le C\sqrt{q+n}.
\end{align*}
By the definition of the row norm and invariance under Hamiltonian conjugation in \eqref{BHexp:functional-calculus}, we obtain
\begin{align*}
 \int_{\R}w_\beta(t)\sum_i\bigl(S_{i,q}(t)S_{i,q}(t)^\dagger+\dot S_{i,q}(t)\dot S_{i,q}(t)^\dagger\bigr)\,dt
 &\le C(q+n)\left(\int_{\R}w_\beta(t)\,dt\right)\id
 \le C_0(q+n)\id,
\end{align*}
where we used $\int_{\R}w_\beta(t)\,dt=1/2$ and absorbed the fixed factors into $C_0$. This proves \eqref{BHexp:source-upper}.

We now prove \eqref{BHexp:source-lower}. Fix $\psi\in\cH_q$. We need a lower bound on
$\int w_\beta(t)\sum_i(\left\|T_{i,q}(t)\psi\right\|^2+\left\|\dot T_{i,q}(t)\psi\right\|^2)\,dt$.
To connect this expression to \eqref{BHexp:functional-calculus}, define $f_0(s)=\widetilde f_\eta(s)$, $f_1(s)=s\widetilde f_\eta(s)$ and $M_{j,i}=f_j(\operatorname{ad}_{H_{\operatorname{BH}}})(a_{i,q})$, so $M_{0,i}=T_{i,q}(0)$ and $M_{1,i}=-i\dot T_{i,q}(0)$.
If $H_{\operatorname{BH},q}\psi=E\psi$, then $M_{j,i}\psi=f_j(H_{\operatorname{BH},q-1}-E)a_{i,q}\psi$ for $j=0,1$.
This allows us to apply a scalar lower bound on $|f_0|^2+|f_1|^2$. For general $\psi$, we will obtain the same comparison by first concentrating its input energy near $E$ and estimating the error.

We first establish the two inequalities needed for this comparison. From the definition of $H_{\operatorname{BH}}$, we obtain
\begin{align}
 -\sum_i a_{i,q}^\dagger\operatorname{ad}_{H_{\operatorname{BH}}}(a_{i,q})
 &=U\sum_iN_i(N_i-1)-\mu q-J\sum_{\langle i,j\rangle}(a_{i,q}^\dagger a_{j,q}+a_{j,q}^\dagger a_{i,q})\nn\\
 &\ge\frac{Uq^2}{n}-C_Dq,
 \label{BHexp:extraction-bound}
\end{align}
where we defined $C_D=U+|\mu|+2D|J|.$ Here we used the same occupation and hopping bounds as in the derivation of \eqref{BHexp:energy}.
The filter satisfies, for $s\in\R$,
\begin{align}
 |f_0(s)|^2+|f_1(s)|^2=(1+s^2)|\widetilde f_\eta(s)|^2
 \ge b(K_0-s),
 \label{BHexp:minorant}
\end{align}
where $K_0=C_D+U$ and $b>0$ denotes some constant depending only on $\beta,U,\mu,J,D$ and independent of $\eta$. Indeed, for $\eta=0$, the ratio $(1+s^2)|\widetilde f_0(s)|^2/(K_0-s)$ is positive on $s<K_0$, tends to $\Omega$ as $s\to-\infty$, and tends to infinity as $s\uparrow K_0$. The same lower bound, up to a fixed factor, holds uniformly in $\eta$ because
\begin{align*}
 (2\pi)^{-1/2}\le\frac{\widetilde f_\eta(s)}{\widetilde f_0(s)}
 =\left(\frac{1+s^2/\Omega^2}{1+\beta^2s^2}\right)^{\eta/2}\le1.
\end{align*}
For $s\ge K_0$, the right side of \eqref{BHexp:minorant} is non-positive.

To apply these inequalities to an arbitrary $\psi$, let $K_E=H_{\operatorname{BH},q-1}-E$ and $\chi_E=(\pi\delta^2)^{-1/4}e^{-(H_{\operatorname{BH},q}-E)^2/(2\delta^2)}$, where $\delta>0$ will be fixed below. Thus $\chi_E\psi$ concentrates the input energy near $E$. Gaussian integration gives
\begin{align}
 \int_{\R}\chi_E^2\,dE=\id,
 \qquad\int_{\R}(H_{\operatorname{BH},q}-E)\chi_E^2\,dE=0,
 \qquad\int_{\R}(H_{\operatorname{BH},q}-E)^2\chi_E^2\,dE=\tfrac12\delta^2\id.
 \label{BHexp:gaussian-windows}
\end{align}
We add and subtract $f_j(K_E)a_{i,q}\chi_E\psi$ and use $\left\|u+v\right\|^2\ge\tfrac12\left\|u\right\|^2-\left\|v\right\|^2$ to obtain
\begin{align}
 &\int_{\R}\sum_i\bigl(\left\|M_{0,i}\chi_E\psi\right\|^2+\left\|M_{1,i}\chi_E\psi\right\|^2\bigr)\,dE\nn\\
 &\quad\ge\frac12\int_{\R}\sum_i\bigl(\left\|f_0(K_E)a_{i,q}\chi_E\psi\right\|^2+\left\|f_1(K_E)a_{i,q}\chi_E\psi\right\|^2\bigr)\,dE\nn\\
 &\qquad-\int_{\R}\sum_i\bigl(\left\|[M_{0,i}-f_0(K_E)a_{i,q}]\chi_E\psi\right\|^2+\left\|[M_{1,i}-f_1(K_E)a_{i,q}]\chi_E\psi\right\|^2\bigr)\,dE.
 \label{BHexp:source-lower-comparison}
\end{align}
We bound the first integral on the right-hand side of \eqref{BHexp:source-lower-comparison} from below and the subtracted integral from above. For the first integral, using $\sum_i a_{i,q}^\dagger a_{i,q}=q\id$, we have $\sum_i a_{i,q}^\dagger K_Ea_{i,q}=q(H_{\operatorname{BH},q}-E)+\sum_i a_{i,q}^\dagger\operatorname{ad}_{H_{\operatorname{BH}}}(a_{i,q})$.
Applying \eqref{BHexp:minorant} to $K_E$ and substituting this identity therefore gives
\begin{align}
 &\int_{\R}\sum_i\bigl(\left\|f_0(K_E)a_{i,q}\chi_E\psi\right\|^2+\left\|f_1(K_E)a_{i,q}\chi_E\psi\right\|^2\bigr)\,dE\nn\\
 &\quad\ge b\int_{\R}\left\langle\chi_E\psi,
 \left[K_0q\id-q(H_{\operatorname{BH},q}-E)
 -\sum_i a_{i,q}^\dagger\operatorname{ad}_{H_{\operatorname{BH}}}(a_{i,q})\right]\chi_E\psi\right\rangle\,dE
 \ge bU(q+q^2/n)\left\|\psi\right\|^2,
 \label{BHexp:localised-source-lower}
\end{align}
where the last inequality follows from \eqref{BHexp:extraction-bound} and \eqref{BHexp:gaussian-windows}. In particular, the term containing $H_{\operatorname{BH},q}-E$ integrates to zero.

We now bound the subtracted integral in \eqref{BHexp:source-lower-comparison}. For this, we first show, for $j=0,1$,
\begin{align}
 \left(\sum_i\left\|[M_{j,i}-f_j(K_E)a_{i,q}]\chi_E\psi\right\|^2\right)^{1/2}
 \le C\sqrt q\,\left\|(H_{\operatorname{BH},q}-E)\chi_E\psi\right\|,
 \label{BHexp:freeze}
\end{align}
where $C\ge0$ denotes some constant depending only on $\beta,U,\mu,J,D$. To prove this, we use the resolvent identity
\begin{align*}
 R(z)(a_{i,q})-(z-K_E)^{-1}a_{i,q}
 =-R(z)\big((z-K_E)^{-1}a_{i,q}(H_{\operatorname{BH},q}-E)\big).
\end{align*}
The Laplace representation used for \eqref{BHexp:resolvent}, together with $\sum_i a_{i,q}^\dagger a_{i,q}=q\id$, gives
\begin{align*}
 \left(\sum_i\left\|\big[R(\pm ir)(a_{i,q})-(\pm ir-K_E)^{-1}a_{i,q}\big]\chi_E\psi\right\|^2\right)^{1/2}
 \le r^{-2}\sqrt q\,\left\|(H_{\operatorname{BH},q}-E)\chi_E\psi\right\|.
\end{align*}
Here the right Hamiltonian unitaries commute with $(H_{\operatorname{BH},q}-E)\chi_E$, so both resolvents contribute a factor $r^{-1}$ without changing the norm on the right. For $f_1(s)=s\widetilde f_\eta(s)$, the same cancellation as in \eqref{BHexp:resolvent-cancellation} gives
\begin{align*}
 \operatorname{ad}_{H_{\operatorname{BH}}}R(z)(a_{i,q})-K_E(z-K_E)^{-1}a_{i,q}
 =z\big[R(z)(a_{i,q})-(z-K_E)^{-1}a_{i,q}\big].
\end{align*}
Inserting these estimates into \eqref{BHexp:slit}, with $g=\widetilde f_\eta$, gives, for $j=0,1$,
\begin{align*}
 \left(\sum_i\left\|[M_{j,i}-f_j(K_E)a_{i,q}]\chi_E\psi\right\|^2\right)^{1/2}
 &\le\sqrt q\,\left\|(H_{\operatorname{BH},q}-E)\chi_E\psi\right\|
 \int_{1/\beta}^\infty\bigl(|\zeta_+(r)|+|\zeta_-(r)|\bigr)r^{j-2}\,dr,
\end{align*}
where the integral is bounded uniformly in $\eta$ by \eqref{BHexp:zeta-bound}, proving \eqref{BHexp:freeze}.

Squaring \eqref{BHexp:freeze} for $j=0$ and $j=1$, adding the two inequalities, and using \eqref{BHexp:gaussian-windows}, we obtain
\begin{align*}
 &\int_{\R}\sum_i\bigl(\left\|[M_{0,i}-f_0(K_E)a_{i,q}]\chi_E\psi\right\|^2+\left\|[M_{1,i}-f_1(K_E)a_{i,q}]\chi_E\psi\right\|^2\bigr)\,dE\\
 &\quad\le Cq\int_{\R}\left\|(H_{\operatorname{BH},q}-E)\chi_E\psi\right\|^2\,dE
 \le Cq\delta^2\left\|\psi\right\|^2,
\end{align*}
where we renamed the constant $C$. Substituting this bound and \eqref{BHexp:localised-source-lower} into \eqref{BHexp:source-lower-comparison} yields
\begin{align}
 \int_{\R}\sum_i\bigl(\left\|M_{0,i}\chi_E\psi\right\|^2+\left\|M_{1,i}\chi_E\psi\right\|^2\bigr)\,dE
 &\ge\left[\frac{bU}{2}(q+q^2/n)-Cq\delta^2\right]\left\|\psi\right\|^2
 \ge\frac{bU}{4}(q+q^2/n)\left\|\psi\right\|^2,
 \label{BHexp:localised-filter-lower}
\end{align}
where we fixed $\delta>0$ small enough that $C\delta^2\le bU/4$. This choice is independent of $\eta,n,q$.

It remains to connect \eqref{BHexp:localised-filter-lower} to the time average in \eqref{BHexp:source-lower}. Gaussian integration and \eqref{BHexp:functional-calculus} give
\begin{align*}
 \int_{\R}\sum_i\bigl(\left\|M_{0,i}\chi_E\psi\right\|^2+\left\|M_{1,i}\chi_E\psi\right\|^2\bigr)\,dE
 &=\int_{\R}\frac\delta{\sqrt\pi}e^{-\delta^2t^2}
 \sum_i\bigl(\left\|T_{i,q}(t)\psi\right\|^2+\left\|\dot T_{i,q}(t)\psi\right\|^2\bigr)\,dt\\
 &\le C_\delta\int_{\R}w_\beta(t)
 \sum_i\bigl(\left\|T_{i,q}(t)\psi\right\|^2+\left\|\dot T_{i,q}(t)\psi\right\|^2\bigr)\,dt,
\end{align*}
where $C_\delta>0$ denotes some constant depending only on $\beta,\delta$, since $\delta e^{-\delta^2t^2}/\sqrt\pi\le C_\delta w_\beta(t)$. Combining this with \eqref{BHexp:localised-filter-lower} proves \eqref{BHexp:source-lower} with $c_0=bU/(4C_\delta)>0$, since $\psi\in\cH_q$ was arbitrary.

We now prove \eqref{BHexp:derivative-bound}. First let $X$ be supported in finitely many particle sectors. For $p\ge1$, $q\ge0$ and $r=0,1,2$, the block identity for the annihilation derivation gives
\begin{align*}
 \Pi_{p-1}\frac{d^r}{dt^r}\partial_{t,\widetilde f_\eta}^{i,-}(X)\Pi_q
 =\left(\frac{d^r}{dt^r}T_{i,p}(t)\right)X_{p,q}
 -X_{p-1,q-1}\left(\frac{d^r}{dt^r}S_{i,q}(t)\right),
\end{align*}
with the second term absent when $q=0$. On this output block, $\cN^{-2}$ multiplies by $(p+q+n)^{-2}$. The inequality $\left\|u-v\right\|_2^2\le2\left\|u\right\|_2^2+2\left\|v\right\|_2^2$, with the column norm for the first term and the row norm for the second, therefore gives
\begin{align}
 &\sum_i\left\|\cN^{-2}\left(\Pi_{p-1}\frac{d^r}{dt^r}\partial_{t,\widetilde f_\eta}^{i,-}(X)\Pi_q\right)\right\|_2^2\nn\\
 &\quad\le\frac{2}{(p+q+n)^4}
 \left(\left\|\left(\frac{d^r}{dt^r}T_{i,p}(t)\right)_i\right\|_{\mathrm{col}}^2\left\|X_{p,q}\right\|_2^2
 +\left\|\left(\frac{d^r}{dt^r}S_{i,q}(t)\right)_i\right\|_{\mathrm{row}}^2\left\|X_{p-1,q-1}\right\|_2^2\right).
 \label{BHexp:weighted-block-split}
\end{align}
We bound the two operator norms on the right-hand side of \eqref{BHexp:weighted-block-split} together. By \eqref{BHexp:functional-calculus}, it is enough to estimate $\operatorname{ad}_{H_{\operatorname{BH}}}^rg(\operatorname{ad}_{H_{\operatorname{BH}}})(a_{i,q})$ for either $g(s)=\widetilde f_\eta(s)$ or $g(s)=\widetilde f_\eta(-s)$.
Before applying the filter, for $r=0,1,2$ we have
\begin{align}
 \left\|\bigl(\operatorname{ad}_{H_{\operatorname{BH}}}^r(a_{i,q})\bigr)_i\right\|_{\mathrm{rc}}
 \le C(q+n+1)^{r+1/2},
 \label{BHexp:bare-commutator-bounds}
\end{align}
where $C\ge0$ denotes some constant depending only on $U,\mu,J,D$. For $r=0$, this follows from $\sum_i a_{i,q}^\dagger a_{i,q}=q\id$ and $\sum_i a_{i,q}a_{i,q}^\dagger=(q+n-1)\id$. For $r=1$, we also use $\operatorname{ad}_{H_{\operatorname{BH}}}(a_{i,q})=-Ua_{i,q}N_i+(U+\mu)a_{i,q}+J\sum_{j\sim i}a_{j,q}$.
For $r=2$, expand $\operatorname{ad}_{H_{\operatorname{BH}}}(a_{i,q}N_i)=\operatorname{ad}_{H_{\operatorname{BH}}}(a_{i,q})N_i+a_{i,q}[H_{\operatorname{BH},q},N_i].$
The column and row squares of $(a_{i,q}N_i^2)_i$ are bounded by $(q+n+1)^5$. Neighbour Cauchy-Schwarz and occupation counting bound the corresponding squares of $(\sum_{j\sim i}a_{j,q}N_i)_i$ by $C(q+n)^3$. Moreover, $\left\|[H_{\operatorname{BH},q},N_i]\right\|\le4D|J|q$ gives $\left\|\bigl(a_{i,q}[H_{\operatorname{BH},q},N_i]\bigr)_i\right\|_{\mathrm{rc}}\le C\sqrt n\,q^{3/2}\le C(q+n+1)^{5/2}$.
Together with the terms already bounded for $r=0,1$, these estimates prove \eqref{BHexp:bare-commutator-bounds} for $r=2$.

We now insert \eqref{BHexp:bare-commutator-bounds} into the Cauchy representation \eqref{BHexp:slit}. Using the two resolvent bounds in \eqref{BHexp:resolvent} and the density bound \eqref{BHexp:zeta-bound}, for $r=1,2$ we obtain
\begin{align*}
 \left\|\bigl(\operatorname{ad}_{H_{\operatorname{BH}}}^rg(\operatorname{ad}_{H_{\operatorname{BH}}})(a_{i,q})\bigr)_i\right\|_{\mathrm{rc}}
 &\le C\int_{1/\beta}^\infty
 \bigl(|1-\beta^2u^2|^{-1/4}+|1-u^2/\Omega^2|^{-1/4}\bigr)\\
 &\quad\times\min\left\{2\left\|\bigl(\operatorname{ad}_{H_{\operatorname{BH}}}^{r-1}(a_{i,q})\bigr)_i\right\|_{\mathrm{rc}},\ u^{-1}\left\|\bigl(\operatorname{ad}_{H_{\operatorname{BH}}}^r(a_{i,q})\bigr)_i\right\|_{\mathrm{rc}}\right\}\,du\\
 &\le C(q+n+1)^r.
\end{align*}
For the last inequality, split the integral at $\max\{q+n+1,2\Omega\}$ and use \eqref{BHexp:bare-commutator-bounds}. For $r=0$, direct integration gives $\left\|\bigl(g(\operatorname{ad}_{H_{\operatorname{BH}}})(a_{i,q})\bigr)_i\right\|_{\mathrm{rc}}\le C\sqrt{q+n+1}$. Since Hamiltonian conjugation preserves both row and column norms, \eqref{BHexp:functional-calculus} therefore implies, for $r=0,1,2$ and $t\in\R$,
\begin{align}
 \left\|\left(\frac{d^r}{dt^r}T_{i,q}(t)\right)_i\right\|_{\mathrm{rc}}
 +\left\|\left(\frac{d^r}{dt^r}S_{i,q}(t)\right)_i\right\|_{\mathrm{rc}}
 \le C(q+n+1)^2,
 \label{BHexp:filtered-sector-bounds}
\end{align}
where $C\ge0$ denotes some constant depending only on $\beta,U,\mu,J,D$, uniformly in $\eta,n,q,t$.

Substituting \eqref{BHexp:filtered-sector-bounds} into the two terms on the right-hand side of \eqref{BHexp:weighted-block-split}, we obtain
\begin{align}
 \sum_i\left\|\cN^{-2}\left(\Pi_{p-1}\frac{d^r}{dt^r}\partial_{t,\widetilde f_\eta}^{i,-}(X)\Pi_q\right)\right\|_2^2
 &\le C\frac{(p+n+1)^4\left\|X_{p,q}\right\|_2^2+(q+n+1)^4\left\|X_{p-1,q-1}\right\|_2^2}{(p+q+n)^4}\nn\\
 &\le C\left(\left\|X_{p,q}\right\|_2^2+\left\|X_{p-1,q-1}\right\|_2^2\right),
 \label{BHexp:weighted-block-bound}
\end{align}
where terms with a negative particle index are omitted. Summing \eqref{BHexp:weighted-block-bound} over $p,q$ counts each input block at most twice, and $\sum_{p,q}\left\|X_{p,q}\right\|_2^2=\left\|X\right\|_2^2$. The creation derivations satisfy $(\partial_{t,\widetilde f_\eta}^{i,-}(X^\dagger))^\dagger=-\partial_{t,\widetilde f_\eta}^{i,+}(X)$, also after time differentiation. Since $\cN$ commutes with adjunction, the same bound holds for the creation terms. Integrating against $w_\beta$ and summing over $r=0,1,2$ therefore gives
\begin{align*}
 \sum_{r=0}^2\left\|\cN^{-2}\left(\frac{d^r}{dt^r}\partial_{t,\widetilde f_\eta}(X)\right)\right\|_w^2
 \le C\left\|X\right\|_2^2.
\end{align*}
Cauchy-Schwarz over the three terms gives \eqref{BHexp:derivative-bound} with $C_1=\sqrt{3C}$. By density, the three weighted maps extend continuously to all $X\in\mathscr T_2(\cH)$, with the same bound.

All retained annihilation blocks are unchanged at a total-particle-number cutoff, completing the proof.
\end{proof}

\section{Finite-dimensional circuit preparations of Bose-Hubbard Gibbs states}
\label{sec:CircuitImplementationBH}
In the following we employ the general framework of    \mbox{
\cite[Sections~4.3 \& 4.4]{BeckerRouzeSalzmannToAppearcmp} }\hskip0pt
 to provide a finite-dimensional circuit implementation for the Gibbs state of    the  Bose-Hubbard    Hamiltonian $H=H_{\operatorname{BH}}$ from Section~\ref{BHexp:section} . In particular we employ the Gibbs sampler generated by    $\cL_{\sigma_E,\widehat f_\eta,H_{\operatorname{BH}}}$  with bare jumps    $\{A^\alpha\}_{\alpha\in\cA} \equiv \{a_i,a^\dagger_i\}_{i=1}^n$ and the filters $\widehat f_\eta$ from \eqref{BHexp:filter}, for fixed $\eta\in(0,1/2]$. By \mbox{
\cite[Lemma~3.5]{BeckerRouzeSalzmannToAppearcmp}}\hskip0pt
, the corresponding time-domain filter satisfies $f_\eta\in L^1(\R)$ and $\int_\R(1+|t|)e^{r|t|^{2\kappa}}|f_\eta(t)|\,dt<\infty$ for every $r\ge0$ and $\kappa\in(0,1/2)$. Moreover, 
\cite[Lemma~4.19]{BeckerRouzeSalzmannToAppearcmp} \hskip0pt%
\hskip0pt
verifies the filter regularity required for circuit implementation. As argued in 
\cite[Theorem~4.12]{BeckerRouzeSalzmannToAppearcmp}\hskip0pt%
\hskip0pt
 , the dynamics    $e^{t\cL_{\sigma_E,\widehat f_\eta,H_{\operatorname{BH}}}}$  can be well-approximated by the finite-dimensional dynamics whose generator can be obtained by replacing the unbounded bare jumps and Hamiltonian by their truncated, finite-dimensional  
  counterparts
\begin{align}
\label{eq:DefTruncOperators}
    a^{\le M}_i := \pi_M^i a_i \pi_M^i,\qquad
    \left(a_i^{\le M}\right)^\dagger = \pi_M^i a_i^\dagger \pi_M^i =: \left(a_i^\dagger\right)^{\le M},\qquad
    H_{\le M} := P_MHP_M,
\end{align} 
for some truncation level    $M\in\N$, where $\pi_M=\sum_{m=0}^{M}\kb{m}$  denotes the projection onto the    first $M+1$ single-mode Fock states, $\pi_M^i$ acts as $\pi_M$ on mode $i$ and as the identity on the other modes, and $P_M:=(\pi_M)^{\otimes n}$. These local cutoffs are used for the circuit implementation, while the gap bound for the Gibbs sampler corresponding to the full Hamiltonian is given in Theorem~\ref{BHexp:theorem}.

We start by recalling some of the results in \cite[Section~4.3.1]{BeckerRouzeSalzmannToAppearcmp} in order to verify the conditions in    
\cite[Theorem~4.12]{BeckerRouzeSalzmannToAppearcmp} \hskip0pt%
\hskip0pt
 on the bare jumps:
By definition we have
\begin{align}
\label{eq:NormTrunc}
    \left\|a^{\le M}\right\| =\left\|\left(a^{\le M}\right)^\dagger\right\| = \sqrt{M}.
\end{align}
Furthermore, throughout this section we consider the choice $\NA\equiv\Ntot = \sum_{i=1}^{n} a_i^\dagger a_i$. Since $[H_{\operatorname{BH}},\Ntot]=0$, we use the commuting branch of the approximation and circuit theorems in %
\cite{BeckerRouzeSalzmannToAppearcmp}\hskip0pt%
, which permits $0<\kappa\le1$. We see for $0<\kappa\le1$ and $l=1,2$
\begin{align}
\label{eq:2CondTruncJumps}
\nn\left\|e^{\Ntot^\kappa}a_ie^{-2\Ntot^\kappa}\right\|\,,\,\left\|e^{\Ntot^\kappa}\left(a_i\right)^\dagger e^{-2\Ntot^\kappa}\right\|\,,\,\left\|a_ie^{-\Ntot^\kappa}\right\|\,,\, \left\|e^{-\Ntot^\kappa}a_i\right\|&<\infty\\
\nn \left\|e^{\Ntot^\kappa}\, a^{\le M}_i e^{-\Ntot^\kappa} \right\| \ ,\   \left\|e^{\Ntot^\kappa}\, \left(a^{\le M}_i\right)^\dagger e^{-\Ntot^\kappa} \right\|& \lesssim \sqrt{M},\\
\left\|e^{(l-1)\Ntot^\kappa}\left(a_i- a^{\le M}_i\right)e^{-l\Ntot^\kappa}\right\|\ , \ \left\|e^{(l-1)\Ntot^\kappa}\left(a^\dagger_i- \left(a^{\le M}_i\right)^\dagger\right)e^{-l\Ntot^\kappa}\right\| &\lesssim \sqrt{M}e^{-M^\kappa},
\end{align}
where    $\lesssim$ hides constants depending only on $\kappa$  and where we used    \mbox{
\cite[Lemma~4.6]{BeckerRouzeSalzmannToAppearcmp} }\hskip0pt
 for the second line.  Taking adjoints also gives the bounds with $a_i$ replaced by $a_i^\dagger$ in the last two norms of the first line and with the signs in the exponential conjugations reversed in the second line.

Furthermore, in    
\cite[Theorem~4.12]{BeckerRouzeSalzmannToAppearcmp} \hskip0pt%
\hskip0pt
 the following conditions on the Hamiltonian and its truncation are needed: It is assumed that $H$ is well-approximated by its truncated analogue on states which are in the domain of $e^{\Ntot^\kappa}$, i.e.
\begin{align}
\label{eq:HamTruncApprox}
 \left\|\left(H-H_{\le M}\right)e^{-\Ntot^\kappa}\right\| \le p(n,M) e^{-M^\kappa}
\end{align}
for some function    $p$,  which is polynomially bounded in its entries. It is, further, assumed that the truncated Hamiltonian has operator norm bounded as
\begin{align}
\label{eq:TruncHamNorm}
    \left\|H_{\le M}\right\| \le p_1(n,M)
\end{align}
for some function    $p_1$,  which is polynomially bounded in its entries. Lastly, it is assumed that the Hamiltonian evolution with respect to the untruncated Hamiltonian $H$ is compatible with exponential energy constraints with respect to $\Ntot$, i.e.~for $k=2,4$ and $t\in\R$ we have

  \begin{align}
\label{eq:ExpHamBound}
    e^{-itH}e^{k\Ntot^\kappa}e^{itH} \le e^{r|t|^{2\kappa}}e^{k\Ntot^\kappa}
\end{align} 
for some    $r\ge0$.

In the following    subsection  we verify the above conditions for $H_{\operatorname{BH}}$ and obtain quantitative  Gibbs state preparation    bounds using Theorem~\ref{BHexp:theorem} and 
\cite[Corollary~4.25]{BeckerRouzeSalzmannToAppearcmp}\hskip0pt%
\hskip0pt
 .

\subsection{ Gibbs state preparation and resource bounds}

In this section we show that  the dynamics $e^{t\cL_{\sigma_E,\widehat f_\eta,H_{\operatorname{BH}}}}$ starting from the vacuum can be efficiently simulated on a qubit-based quantum computer, with resources polynomial in $n$, $t$ and $\log(1/\eps)$, under the access assumptions of Theorem~\ref{thm:CircuitGibbsSF}. With that we obtain a Gibbs state preparation algorithm by evolving the vacuum for a mixing time of order $\lambda_2^{-1}(n+\log(1/\eps))$, where $\lambda_2:=\operatorname{gap}(L_{\sigma_E,\widehat f_\eta,H_{\operatorname{BH}}})$ and the dependence on $n$ accounts for the initial-state factor bounded in Lemma~\ref{lem:VacuumBoundSF}. By Theorem~\ref{BHexp:theorem}, $\lambda_2\ge C e^{-cn}>0$, so this mixing time is finite and at most exponential in the number of modes for fixed accuracy. This is a worst-case upper bound. Moreover, improving the gap lower bound in Theorem~\ref{BHexp:theorem} to an inverse-polynomial or positive constant bound in $n$, for example for small hopping strength $|J|$, would yield an efficient end-to-end Gibbs state preparation algorithm under the same implementation assumptions. For that we first state and prove the following two supporting lemmas.

\begin{lemma}
\label{lem:VacuumBoundSF}
Let $\beta,U>0$, $\mu,J\in\R$ and consider the Bose-Hubbard Hamiltonian $H_{\operatorname{BH}}$ from Section~\ref{BHexp:section} . Then we have

  \begin{align}
\label{eq:VacuumSF}
  \kb{0}^{\otimes n} \le \mathcal{Z}_\beta\,\sigma_\beta(H_{\operatorname{BH}}),
\end{align} 
where we defined the partition function    $\mathcal{Z}_\beta := \Tr\left(e^{-\beta H_{\operatorname{BH}}}\right)$.  Furthermore, the partition function satisfies the upper bound

  \begin{align}
\label{eq:ZScalingSF}
\mathcal{Z}_\beta \le e^{Cn},
\end{align} 
   where $C\ge0$ denotes some constant depending only on $\beta,U,\mu,J,D$ .
\end{lemma}
\begin{proof}
First note that by definition    $a_i\ket{0}^{\otimes n}=0$ for all $i\in[n]$.  Therefore, we immediately see    $H_{\operatorname{BH}}\ket{0}^{\otimes n}=0$,  which implies \eqref{eq:VacuumSF} as

  \begin{align*}
    \bra{0}^{\otimes n}\sigma^{-1}_\beta(H_{\operatorname{BH}})\ket{0}^{\otimes n}=\mathcal Z_\beta.
\end{align*} 
We continue with bounding the  
  partition function. Using the hopping bound preceding \eqref{BHexp:energy} and completing the square gives
\begin{align}
\label{eq:BHImplementationLower}
H_{\operatorname{BH}}\ge\frac{U}{2}\sum_iN_i^2-b\Ntot
\ge\frac{U}{4}\sum_iN_i^2-\frac{b^2}{U}n\id,
\end{align} 
where we    denoted $b=U/2+|\mu|+2D|J|$ . By the min-max principle of eigenvalues 
  , we therefore have
\begin{align*}
\mathcal Z_\beta\le e^{\beta b^2n/U}\Tr\left(e^{-\beta U\sum_iN_i^2/4}\right)
=e^{\beta b^2n/U}\left(\sum_{m=0}^\infty e^{-\beta Um^2/4}\right)^n,
\end{align*} 

  which proves  \eqref{eq:ZScalingSF}.
\end{proof}

 \begin{lemma}
\label{lem:EGibbsSF}
Let $\beta,U>0$, $\mu,J\in\R$ and consider the Bose-Hubbard Hamiltonian $H_{\operatorname{BH}}$ from Section~\ref{BHexp:section} . Then we have for $\kappa\in(0,1]$

  \begin{align}
E_{\operatorname{Gibbs}}:=\Tr\left(e^{4\Ntot^\kappa}\sigma_\beta(H_{\operatorname{BH}})\right)\le e^{Cn},
\end{align} 
   where $C\ge0$ denotes some constant depending only on $\beta,U,\mu,J,D,\kappa$ .
\end{lemma}
\begin{proof}

  By number conservation we have
\begin{align}
\label{eq:H_SFNtotCommute}
[H_{\operatorname{BH}},\Ntot]=0.
\end{align} 

Using $x^\kappa\le1+x$ for $x\ge0$ and $0<\kappa\le1$, together with \eqref{eq:H_SFNtotCommute}, gives
\begin{align}
\label{eq:E_GibbsProofSF}
\Tr\left(e^{4\Ntot^\kappa}\sigma_\beta(H_{\operatorname{BH}})\right)
&\le e^4\frac{\mathcal Z_\beta(\mu+4/\beta)}{\mathcal Z_\beta(\mu)}
\le e^{Cn}.
\end{align}
Here $\mathcal Z_\beta(\mu)\ge1$ because the vacuum has energy zero, and Lemma~\ref{lem:VacuumBoundSF} applies with $\mu$ replaced by $\mu+4/\beta$. At $\kappa=1$ the first inequality is an equality without the factor $e^4$.
\end{proof}

The following result is a consequence of Theorem~\ref{BHexp:theorem} and %
\cite[Theorem~4.23 and Corollary~4.25]{BeckerRouzeSalzmannToAppearcmp}\hskip0pt%
.
\begin{theorem}
\label{thm:CircuitGibbsSF}
Let $\eps\in(0,1)$, $\beta,U>0$, $\mu,J\in\R$, and fix $\eta\in(0,1/2]$ and $\sigma_E\in(0,\infty)$. Consider the Bose-Hubbard Hamiltonian $H_{\operatorname{BH}}$ from Section~\ref{BHexp:section} and the filter $\widehat f_\eta$ from \eqref{BHexp:filter}. Assume coherent, index-controlled access to block encodings of the truncated jumps and their adjoints, controlled Hamiltonian simulation, the state-preparation oracles for the windowed jump and coherent-correction filters, and reversible evaluation of the Gauss--Hermite nodes and weights at the required precision, as specified in %
\cite[Corollary~4.20 and Theorem~4.23]{BeckerRouzeSalzmannToAppearcmp}\hskip0pt%
. All ancillary-qubit counts exclude workspace internal to the supplied oracles.

For $t\ge0$, the state $e^{t\cL_{\sigma_E,\widehat f_\eta,H_{\operatorname{BH}}}}(\kb{0}^{\otimes n})$ can be prepared within $\eps$-trace distance on a quantum computer with $\mathcal O\left(n\log\left(2+n+\log(2+t/\eps)\right)\right)$ system qubits, together with $\widetilde{\mathcal O}\left(\log^2\left(2+n(1+t)/\eps\right)\right)$ ancilla qubits, and with total Hamiltonian simulation time
\begin{align*}
\widetilde{\mathcal O}\left((1+t)\operatorname{poly}\left(n,\log(1/\eps)\right)\right).
\end{align*}
Consequently, denoting $\lambda_2:=\operatorname{gap}(L_{\sigma_E,\widehat f_\eta,H_{\operatorname{BH}}})>0$,
  the Gibbs state $\sigma_\beta(H_{\operatorname{BH}})$  can be prepared within $\eps$-trace distance on a quantum computer with  
  $\mathcal O\left(n\log\left(2+n+\log\left((1+\lambda_2^{-1})/\eps\right)\right)\right)$ system qubits, together with $\widetilde{\mathcal O}\left(\log^2\left(2+n(1+\lambda_2^{-1})/\eps\right)\right)$ ancilla qubits, and with total Hamiltonian simulation time
\begin{align*}
    \widetilde{\mathcal O}\left(\left(1+\frac{1}{\lambda_2}\right)\operatorname{poly}\left(n,\log(1/\eps)\right)\right),
\end{align*} 
where    $\widetilde{\mathcal O}$ hides constants depending on the fixed  parameters and additionally suppresses subdominant    polylogarithmic  factors in $t$ and $1/\lambda_2$, respectively.  Employing the lower bound $\lambda_2\ge C e^{-cn}$ from Theorem~\ref{BHexp:theorem}, we obtain a total Hamiltonian simulation time bounded by
\begin{align*}
\mathcal O\left(e^{cn}\operatorname{poly}\left(n,\log(1/\eps)\right)\right).
\end{align*}
 \end{theorem}

\begin{proof}

 We start by verifying all the required conditions in    
\cite[Theorem~4.12]{BeckerRouzeSalzmannToAppearcmp}\hskip0pt%
\hskip0pt
 .
For the bare jumps we have already verified all relevant conditions in \eqref{eq:NormTrunc} and \eqref{eq:2CondTruncJumps}.
We continue with   verifying the conditions on the Hamiltonian and its truncation, i.e.~\eqref{eq:HamTruncApprox}, \eqref{eq:TruncHamNorm} and \eqref{eq:ExpHamBound}.
From \eqref{eq:H_SFNtotCommute}, we see that \eqref{eq:ExpHamBound} is satisfied for $r=0$; we henceforth take $\kappa=1$. Only the first time moment of the integrable filter is then required, and it is finite for $\eta>0$.

We continue with the conditions on the truncated Hamiltonian. We show   for $M\in\N$ that

  \begin{align}
\label{eq:VerifyHamConditionSF}
\left\|\left(H_{\operatorname{BH}}-(H_{\operatorname{BH}})_{\le M}\right)e^{-\Ntot^\kappa}\right\|
\le C_\kappa(M+1)^2e^{-M^\kappa},
\end{align} 
 where $C_\kappa\ge0$ denotes some constant depending only on $U,\mu,J,D,\kappa$.  To see this 
  , note that $\left\|H_{\operatorname{BH},q}\right\|\le C(q+1)^2$ by $\sum_iN_i^2\le\Ntot^2$ and the hopping bound preceding \eqref{BHexp:energy}, where $C\ge0$ depends only on $U,\mu,J,D$. Since $\Ntot\ge M+1$ on the range of $\id-P_M$ and both $H_{\operatorname{BH}}$ and $P_M$ commute with $\Ntot$, we obtain
\begin{align*}
\left\|\left(H_{\operatorname{BH}}-P_MH_{\operatorname{BH}}P_M\right)e^{-\Ntot^\kappa}\right\|
&\le\left\|(\id-P_M)H_{\operatorname{BH}}e^{-\Ntot^\kappa}\right\|
+\left\|P_MH_{\operatorname{BH}}(\id-P_M)e^{-\Ntot^\kappa}\right\|\\
&\le2C\sup_{q\ge M+1}(q+1)^2e^{-q^\kappa}
\le C_\kappa(M+1)^2e^{-M^\kappa},
\end{align*} 

  which proves  \eqref{eq:VerifyHamConditionSF}. Similarly,    using $\Ntot\le nM$ on the range of $P_M$ gives $\left\|(H_{\operatorname{BH}})_{\le M}\right\|\le C(nM+1)^2$, proving \eqref{eq:TruncHamNorm}.

We prepare the Gibbs state    $\sigma_\beta(H_{\operatorname{BH}})$  using a circuit implementation of the Gibbs sampler generated by    $\cL_{\sigma_E,\widehat f_\eta,H_{\operatorname{BH}}}$,  where the relevant complexity analysis has been carried out in    
\cite[Theorem~4.23]{BeckerRouzeSalzmannToAppearcmp}\hskip0pt%
\hskip0pt
.
 We take as initial state of the circuit    the  multi-mode vacuum state    $\rho_{\operatorname{ini}}=\kb{0}^{\otimes n}$.  From Lemma~\ref{lem:VacuumBoundSF} and Lemma~\ref{lem:EGibbsSF} we see that the constants    $\mathfrak c$  and $E_{\operatorname{Gibbs}}$ in  

\cite[Theorem~4.23]{BeckerRouzeSalzmannToAppearcmp} \hskip0pt%
\hskip0pt
satisfy $\mathfrak c=\mathcal Z_\beta\le e^{Cn}$ and $E_{\operatorname{Gibbs}}\le e^{Cn}$, and the vacuum lies in every truncated system register.
 We take truncation level $M$ for simulation up to time $t$ as in  

\cite[Theorem~4.23]{BeckerRouzeSalzmannToAppearcmp}\hskip0pt%
\hskip0pt
, namely
\begin{align}
\label{eq:BH-dynamics-cutoff}
M=\Theta\left(n+\log(2+t/\eps)\right),
\end{align} 
   where the constant is chosen large enough to absorb the polynomial prefactors in the truncation error and the bounds $\mathfrak c,E_{\operatorname{Gibbs}}\le e^{Cn}$.
Using 
\cite[Remarks~4.21 and 4.22]{BeckerRouzeSalzmannToAppearcmp}\hskip0pt%
\hskip0pt
, the required  block encodings of    $a_i^{\le M}/\sqrt{M}$ and $\left(a_i^{\le M}\right)^\dagger/\sqrt{M}$ can be obtained by reversible arithmetic, and the Hamiltonian simulation can be implemented using the local polynomial structure of $(H_{\operatorname{BH}})_{\le M}$. The filter assumptions follow from 
\cite[Lemmas~3.5 and 4.19]{BeckerRouzeSalzmannToAppearcmp}\hskip0pt%
\hskip0pt
, while the coherent filter and quadrature access is assumed in the statement. 

We have now   checked all required assumptions of    
\cite[Theorem~4.23]{BeckerRouzeSalzmannToAppearcmp}\hskip0pt%
\hskip0pt
 . This allows us to prepare $e^{t\cL_{\sigma_E,\widehat f_\eta,H_{\operatorname{BH}}}}(\kb{0}^{\otimes n})$ within $\eps$-trace distance with $\mathcal O(n\log(M+1))$ system qubits and $\widetilde{\mathcal O}\left(\log^2\left(2+n(1+t)/\eps\right)\right)$ ancilla qubits, using total Hamiltonian simulation time
\begin{align*}
\widetilde{\mathcal O}\left((1+t)\operatorname{poly}\left(n,\log\left(2+\frac{\mathfrak c E_{\operatorname{Gibbs}}}{\eps}\right)\right)\right)
=\widetilde{\mathcal O}\left((1+t)\operatorname{poly}\left(n,\log(1/\eps)\right)\right).
\end{align*}
Substituting the bound on $M$ gives the claimed system-qubit count. For the ancilla count, we choose the frequency-window size $S=\max\{1,4\left\|(H_{\operatorname{BH}})_{\le M}\right\|\}$. Using $|\mathcal A|=2n$, $\left\|(H_{\operatorname{BH}})_{\le M}\right\|\le C(nM+1)^2$ and \eqref{eq:BH-dynamics-cutoff} in %
\cite[Corollary~4.20]{BeckerRouzeSalzmannToAppearcmp}\hskip0pt%
, the number of ancilla qubits is bounded by
\begin{align}
\label{eq:BH-dynamics-ancillas}
\widetilde{\mathcal O}\left(\log(2+nM)+\log^2\left(2+tn/\eps\right)\right)
=\widetilde{\mathcal O}\left(\log^2\left(2+n(1+t)/\eps\right)\right).
\end{align}

We now obtain Gibbs state preparation using %
\cite[Corollary~4.13]{BeckerRouzeSalzmannToAppearcmp}\hskip0pt%
, which shows that a time
\begin{align*}
t=\mathcal O\left(\frac{1}{\lambda_2}\log\left(\frac{2\mathfrak c}{\eps}\right)\right)
=\mathcal O\left(\frac{n+\log(1/\eps)}{\lambda_2}\right)
\end{align*}
suffices to bring the evolved vacuum within $\eps/2$ of $\sigma_\beta(H_{\operatorname{BH}})$ in trace norm. Simulating this evolution to accuracy $\eps/2$ and applying the triangle inequality gives the desired preparation accuracy. For this choice of $t$, \eqref{eq:BH-dynamics-cutoff} gives a cutoff $M$ polynomially bounded in $n$ and $\log((1+\lambda_2^{-1})/\eps)$, yielding $\mathcal O\left(n\log\left(2+n+\log\left((1+\lambda_2^{-1})/\eps\right)\right)\right)$ system qubits. Substituting the same time into \eqref{eq:BH-dynamics-ancillas} gives $\widetilde{\mathcal O}\left(\log^2\left(2+n(1+\lambda_2^{-1})/\eps\right)\right)$ ancilla qubits. Equivalently, applying %
\cite[Corollary~4.25]{BeckerRouzeSalzmannToAppearcmp} \hskip0pt%
gives the following resource bound.

We hence see that    $\sigma_\beta(H_{\operatorname{BH}})$  can be prepared within $\eps$-trace distance    with $\mathcal O(n\log(M+1))$ system qubits and $\widetilde{\mathcal O}\left(\log^2\left(2+n(1+\lambda_2^{-1})/\eps\right)\right)$ ancilla qubits,  using a total Hamiltonian simulation time  
  of order
\begin{align*}
\widetilde{\mathcal O}\left(\left(1+\frac{1}{\lambda_2}\right)\operatorname{poly}\left(n,\log\left(2+\frac{\mathfrak c E_{\operatorname{Gibbs}}}{\eps}\right)\right)\right)
=\widetilde{\mathcal O}\left(\left(1+\frac{1}{\lambda_2}\right)\operatorname{poly}\left(n,\log(1/\eps)\right)\right).
\end{align*} 
   Finally, substituting $\lambda_2\ge C e^{-cn}$ from Theorem~\ref{BHexp:theorem} gives the stated exponential bound on the total Hamiltonian simulation time, where the suppressed polylogarithmic factors in $1/\lambda_2$ are absorbed into the polynomial in $n$ and $\log(1/\eps)$ .
\end{proof}

\subsection{End-to-end simulation cost for free energies}\label{appendixforfreeenergy}

\noindent In this section, we consider the task of computing the difference in free energies of two Hamiltonians $(H_0,D(H_0))$ and $(H_1,D(H_1))$   over a Hilbert space    $\mathcal H$  at inverse temperature $\beta>0$:

  \begin{align*}
\Delta F(\beta,H):=F(\beta,H_1)-F(\beta,H_0),\qquad F(\beta,H):=-\beta^{-1}\log\operatorname{Tr}(e^{-\beta H}).
\end{align*} 
We consider a path $H(s):=(1-s)H_0+sH_1$,  
  and assume that the Gibbs states and the following derivatives are well-defined along the path, as holds for the Bose-Hubbard path below. Duhamel's formula and cyclicity of the trace give
\begin{align}
\Delta F(\beta,H)
&=\int_0^1\frac{d}{ds}F(\beta,H(s))\,ds\nn\\
&=\int_0^1\int_0^1\Tr\left(\sigma_\beta(H(s))^{1-u}(H_1-H_0)\sigma_\beta(H(s))^u\right)\,du\,ds\nn\\
&=\int_0^1\Tr\left(\sigma_\beta(H(s))(H_1-H_0)\right)\,ds.
\label{integralpartitions}\end{align} 
  From this we see that it suffices to estimate the average of the    observable  $V:=H_1-H_0$ in the Gibbs state at different values of $s$ to obtain a good estimate of the free energy difference.

\noindent To illustrate the method, we consider the    $n$ -mode  
  Bose-Hubbard Hamiltonian from \eqref{BHexp:hamiltonian}, with $\beta,U>0$  and    $\mu,J\in\R$, and define
\begin{align}
H_0=\frac{U}{2}\sum_iN_i(N_i-1)-\mu\Ntot,\qquad
V=-J\sum_{\langle i,j\rangle}(a_i^\dagger a_j+a_j^\dagger a_i),\qquad H_1=H_0+V=H_{\operatorname{BH}}.
\label{eq:BH-free-energy-path}\end{align}
Thus $H(s)=H_0+sV$ has hopping strength $sJ$ and fixed $U>0$, so Theorem~\ref{BHexp:theorem} applies along the whole path, including $s=0$. The reference Hamiltonian $H_0$ is a sum of independent single-site Hamiltonians, whose partition function we evaluate below. We first truncate $V$ to the  finite-rank    observable $V_M:=P_MVP_M$, where $P_M=\pi_M^{\otimes n}$ is the projection defined in \eqref{eq:DefTruncOperators}.

\begin{lemma}\label{lemtruncatingobservablesBH}
  Consider the path \eqref{eq:BH-free-energy-path}. For $s\in[0,1]$, we have
\begin{align}
\Tr\left(e^{\Ntot}\sigma_\beta(H(s))\right)\le e^{Cn},\qquad
\Tr\left(\Ntot^2\sigma_\beta(H(s))\right)\le Cn^2.
\label{eq:BH-path-moments}\end{align} 
   Furthermore, for $\eps\in(0,1)$,
\begin{align}
\sup_{s\in[0,1]}\left|\Tr\left(\sigma_\beta(H(s))(V-V_M)\right)\right|\le\eps
\label{eq:BH-observable-tail}\end{align}
whenever $M\ge C(n+\log(1/\eps))$, where $C\ge1$ denotes a constant depending only on $\beta,U,\mu,J,D$ .
\end{lemma}

\begin{proof}

  Since $[H(s),\Ntot]=0$ and the vacuum has energy zero, we have $\mathcal Z_\beta(H(s))\ge1$. Applying Lemma~\ref{lem:VacuumBoundSF} with chemical potential $\mu+\beta^{-1}$ gives
\begin{align*}
\Tr\left(e^{\Ntot}\sigma_\beta(H(s))\right)
=\frac{\Tr\left(e^{-\beta(H(s)-\beta^{-1}\Ntot)}\right)}{\mathcal Z_\beta(H(s))}
\le e^{Cn},
\end{align*} 

  where $C\ge1$ denotes a constant depending only on $\beta,U,\mu,J,D$, uniformly in $s$, since the proof of \eqref{eq:ZScalingSF} uses only $|sJ|\le|J|$. Hence the spectral tail of $\Ntot$ satisfies $\Tr\left(\sigma_\beta(H(s))\1_{\{\Ntot>r\}}\right)\le\min\{1,e^{Cn-r}\}$ for $r\ge0$, and
\begin{align*}
\Tr\left(\Ntot^2\sigma_\beta(H(s))\right)
\le\int_0^\infty 2r\min\{1,e^{Cn-r}\}\,dr
=(Cn)^2+2Cn+2,
\end{align*} 
 which proves \eqref{eq:BH-path-moments} after increasing $C$.

We now prove \eqref{eq:BH-observable-tail}. The hopping bound preceding \eqref{BHexp:energy} gives $|V|\le2D|J|\Ntot$. Since both $V$  and  
  $P_M$ commute with $\Ntot$ and $P_M$ is the identity on every sector with total particle number at most $M$, we have
\begin{align*}
|V-V_M|\le4D|J|\Ntot\1_{\{\Ntot>M\}}\le C e^{-M/2}e^{\Ntot}.
\end{align*} 

  Using \eqref{eq:BH-path-moments}, we therefore obtain
\begin{align*}
\left|\Tr\left(\sigma_\beta(H(s))(V-V_M)\right)\right|
\le C e^{-M/2}\Tr\left(e^{\Ntot}\sigma_\beta(H(s))\right)
\le C e^{Cn-M/2},
\end{align*} 

  where $C\ge1$ denotes a constant depending only on the fixed parameters. The stated choice of $M$ makes this at most $\eps$, after increasing $C$ once more .
\end{proof}

\noindent By \eqref{integralpartitions} and Lemma~\ref{lemtruncatingobservablesBH}, it remains to approximate the thermodynamic integral. Put $g(s)=\Tr(\sigma_\beta(H(s))V)$. Duhamel's formula gives
\begin{align}
 g'(s)&=-\beta\int_0^1\Tr\big(\sigma_\beta(H(s))^{1-u}\overline V_s\sigma_\beta(H(s))^u\overline V_s\big)\,du\le0,
 \label{eq:BH-free-energy-derivative}
\end{align}
where $\overline V_s=V-g(s)\id$. Moreover, $|g(s)|\le2D|J|\Tr(\sigma_\beta(H(s))\Ntot)\le Cn$ by \eqref{eq:BH-path-moments}. These identities follow first on total-particle sectors and then by the same moment bounds. Monotonicity gives
\[
0\le\frac{1}K\sum_{k=0}^{K-1}g(k/K)-\int_0^1g(s)\,ds
\le\frac{g(0)-g(1)}K\le\frac{Cn}K.
\]
Replacing each grid value by $\Tr(\sigma_\beta(H(k/K))V_M)$ adds at most $Ce^{Cn-M/2}$, uniformly in $K$, by \eqref{eq:BH-observable-tail}.
This implies the following:
\begin{proposition}
\label{prop:FreeEnergyEstimateSF}

  For the path \eqref{eq:BH-free-energy-path} and every $\eps\in(0,1)$, choosing $M\ge C(n+\log(1/\eps))$ and $K\ge Cn/\eps$ gives
\begin{align*}
\left|\Delta F(\beta,H)-\frac{1}K\sum_{k=0}^{K-1}\Tr\left(\sigma_\beta(H(k/K))V_M\right)\right|\le\eps,
\end{align*} 
where    $C\ge1$ denotes a constant depending only on $\beta,U,\mu,J,D$ .
\end{proposition}

It remains to estimate the number of samples of    $\sigma_\beta(H(s))$  for different values of $s\in[0,1]$    and the resources  for preparing each Gibbs state   .  This is obtained in the following theorem.
\begin{theorem}
\label{thm:FreeEnergyQuantumAlgorithm}
 Let $\beta,U>0$, $\mu,J\in\R$, $\eps,\delta\in(0,1)$, and fix $\eta\in(0,1/2]$ and $\sigma_E\in(0,\infty)$. Consider the path \eqref{eq:BH-free-energy-path} and the filter $\widehat f_\eta$ from \eqref{BHexp:filter}, under the same coherent-access assumptions and workspace convention as in Theorem~\ref{thm:CircuitGibbsSF}. Denote
\begin{align*}
\lambda_2(s):=\operatorname{gap}(L_{\sigma_E,\widehat f_\eta,H(s)}),\qquad
\lambda_2^{\min}:=\inf_{s\in[0,1]}\lambda_2(s)>0.
\end{align*}
 The free energy    $F(\beta,H_{\operatorname{BH}})$  can be estimated with accuracy    $\eps$  and probability of failure    at most $\delta$  on a quantum computer with  
  $\mathcal O\left(n\log\left(2+n+\log\left((1+(\lambda_2^{\min})^{-1})/\eps\right)\right)\right)$ system qubits and $\widetilde{\mathcal O}\left(\log^2\left(2+n(1+(\lambda_2^{\min})^{-1})/\eps\right)\right)$ ancilla qubits, with total Hamiltonian simulation time of order
\begin{align*}
\widetilde{\mathcal O}\left(\frac{1+(\lambda_2^{\min})^{-1}}{\eps^2}\log\left(\frac{2}{\delta}\right)\operatorname{poly}\left(n,\log(1/\eps)\right)\right).
\end{align*} 
   Here the constants depend only on the fixed physical and filter parameters, and $\widetilde{\mathcal O}$ additionally suppresses subdominant polylogarithmic factors in $1/\lambda_2^{\min}$, $n$  and    $1/\eps$. Employing the uniform bound $\lambda_2^{\min}\ge C e^{-cn}$ from Theorem~\ref{BHexp:theorem}, as justified in the proof, gives a total Hamiltonian simulation time bounded by
\begin{align*}
\mathcal O\left(\frac{e^{cn}}{\eps^2}\log\left(\frac{2}{\delta}\right)\operatorname{poly}\left(n,\log(1/\eps)\right)\right).
\end{align*}
 \end{theorem}

 \begin{proof}
 We first verify that the gap and preparation estimates hold uniformly along the path for $s\in[0,1]$. In the proof of Theorem~\ref{BHexp:theorem}, every estimate involving the hopping strength remains valid with $J$ replaced by $sJ$ and $|sJ|$ bounded by $|J|$. In particular, one can use the same constants in \eqref{BHexp:source-upper}, \eqref{BHexp:source-lower} and \eqref{BHexp:derivative-bound}, taking $C_D=U+|\mu|+2D|J|$ and $K_0=C_D+U$ in \eqref{BHexp:minorant}. Hence the gap constants are uniform in $s\in[0,1]$, and $\lambda_2^{\min}\ge C e^{-cn}>0$. The same observation applies to the partition-function, Gibbs-moment and truncation estimates used in Theorem~\ref{thm:CircuitGibbsSF}.

Note that  
  the reference free energy has the explicit expression
\begin{align}
F(\beta,H_0)=-\frac{n}{\beta}\log\left(\sum_{m=0}^\infty e^{-\beta(Um(m-1)/2-\mu m)}\right).
\label{eq:BH-reference-free-energy}\end{align} 
 The Gaussian decay of the summands allows us to evaluate this expression to accuracy $\eps/4$ using $\mathcal O(1+\sqrt{\log(n/\eps)})$ terms, with constants depending only on $\beta,U,\mu$.  Hence, in order to estimate    $F(\beta,H_{\operatorname{BH}})=\Delta F(\beta,H)+F(\beta,H_0)$,  it remains to estimate the difference in free energies.

Using Proposition~\ref{prop:FreeEnergyEstimateSF}, we approximate $\Delta F(\beta,H)$ to accuracy $\eps/4$ by
\begin{align}
\frac{1}K\sum_{k=0}^{K-1}\Tr\left(\sigma_\beta(H(k/K))V_M\right),
\label{eq:Riemann}\end{align} 
   with the commuting choice $\kappa=1$ and
\begin{align}
M=\Theta\left(n+\log\left(\frac{1+(\lambda_2^{\min})^{-1}}{\eps}\right)\right),\qquad K=\Theta(n/\eps),
\label{eq:ScalingLandM}\end{align}
where the implicit constants depend only on the fixed parameters and $\kappa$ and are chosen sufficiently large. We define $B_M:=1+2D|J|nM$, so that $\left\|V_M\right\|\le2D|J|nM\le B_M$. This bound controls the effect of the state-preparation error and also bounds the measurement outcomes used below. We use the same cutoff for state preparation with accuracy $\eps/(4B_M)$. This is sufficient by %
\cite[Corollary~4.13]{BeckerRouzeSalzmannToAppearcmp} \hskip0pt%
and the uniform bounds on $\mathfrak c$ and $E_{\operatorname{Gibbs}}$ established above, since the additional $\log B_M=\mathcal O(\log(n+1)+\log(M+1))$ in the preparation bound is absorbed by the choice in \eqref{eq:ScalingLandM}.

For each independent sample, choose $k$ uniformly from $\{0,\ldots,K-1\}$, prepare $\rho_k$ as below, and use the following measurement. Since the prepared states are supported on $P_M\cH$, we choose a lattice edge uniformly, measure the two-mode observable $-J((a_i^{\le M})^\dagger a_j^{\le M}+(a_j^{\le M})^\dagger a_i^{\le M})$, and multiply its outcome by the number of edges. Conditional on $k$, this gives an unbiased estimator of $\operatorname{Tr}(\rho_k V_M)$ with absolute value at most $B_M$. Averaging over the random grid point therefore estimates the entire sum in \eqref{eq:Riemann}. If the lattice has no edges, return zero. Each two-mode measurement acts on a space of dimension $(M+1)^2$ and can be compiled to the required precision with resources polynomial in $M$ and $\log(1/\eps)$.

By Theorem~\ref{thm:CircuitGibbsSF} , we can  

  prepare an approximation $\rho_k=P_M\rho_kP_M$ to $\sigma_\beta(H(k/K))$ with trace-norm error at most $\eps/(4B_M)$. This introduces an expectation error at most $\eps/4$, since
\begin{align}
\left|\Tr\left((\rho_k-\sigma_\beta(H(k/K)))V_M\right)\right|
\le\left\|\rho_k-\sigma_\beta(H(k/K))\right\|_1\left\|V_M\right\|
\le\eps/4.
\label{eq:BH-free-energy-preparation-error}\end{align} 
   Hoeffding's inequality estimates $K^{-1}\sum_k\operatorname{Tr}(\rho_kV_M)$ to accuracy $\eps/4$ with failure probability at most $\delta$ using a total of
\begin{align}
\mathcal O\left(\frac{B_M^2}{\eps^2}\log\left(\frac{2}{\delta}\right)\right)
\label{eq:SampleComplexitySigmaL}
\end{align}
   samples. Combining the $\eps/4$ errors from evaluating \eqref{eq:BH-reference-free-energy}, the Riemann-sum approximation \eqref{eq:Riemann}, the preparation bound \eqref{eq:BH-free-energy-preparation-error} and the sampling estimate obtained using \eqref{eq:SampleComplexitySigmaL} gives total error at most $\eps$. Finite measurement precision can be included by decreasing the constant error allocations above.

Using  Theorem~\ref{thm:CircuitGibbsSF}  
  with $\lambda_2(k/K)\ge\lambda_2^{\min}$, each preparation uses total Hamiltonian simulation time
\begin{align*}
\widetilde{\mathcal O}\left(\left(1+\frac{1}{\lambda_2^{\min}}\right)\operatorname{poly}\left(n,\log(B_M/\eps)\right)\right).
\end{align*} 
Multiplying this with   \eqref{eq:SampleComplexitySigmaL} and inserting   \eqref{eq:ScalingLandM} yields the claimed bound, with the additional powers of $\log(1+(\lambda_2^{\min})^{-1})$ from $B_M^2$ included in the $\widetilde{\mathcal O}$ notation. The common cutoff in \eqref{eq:ScalingLandM} uses $\mathcal O(n\log(M+1))$ system qubits, giving the stated count. Substituting the preparation accuracy $\eps/(4B_M)$ in the ancilla bound of Theorem~\ref{thm:CircuitGibbsSF} likewise gives the stated count. The registers are reused between samples. Finally, employing $\lambda_2^{\min}\ge C e^{-cn}$ gives the stated exponential bound on the total Hamiltonian simulation time.
 \end{proof}

% The supplied marked bibliography is included directly.
\makeatletter\let\write@bibliographystyle\relax\makeatother
% The tracked bibliography is embedded for reproducible review copies.


\begin{thebibliography}{81}%
\makeatletter
\providecommand \@ifxundefined [1]{%
 \@ifx{#1\undefined}
}%
\providecommand \@ifnum [1]{%
 \ifnum #1\expandafter \@firstoftwo
 \else \expandafter \@secondoftwo
 \fi
}%
\providecommand \@ifx [1]{%
 \ifx #1\expandafter \@firstoftwo
 \else \expandafter \@secondoftwo
 \fi
}%
\providecommand \natexlab [1]{#1}%
\providecommand \enquote  [1]{``#1''}%
\providecommand \bibnamefont  [1]{#1}%
\providecommand \bibfnamefont [1]{#1}%
\providecommand \citenamefont [1]{#1}%
\providecommand \href@noop [0]{\@secondoftwo}%
\providecommand \href [0]{\begingroup \@sanitize@url \@href}%
\providecommand \@href[1]{\@@startlink{#1}\@@href}%
\providecommand \@@href[1]{\endgroup#1\@@endlink}%
\providecommand \@sanitize@url [0]{\catcode `\\12\catcode `\$12\catcode
  `\&12\catcode `\#12\catcode `\^12\catcode `\_12\catcode `\%12\relax}%
\providecommand \@@startlink[1]{}%
\providecommand \@@endlink[0]{}%
\providecommand \url  [0]{\begingroup\@sanitize@url \@url }%
\providecommand \@url [1]{\endgroup\@href {#1}{\urlprefix }}%
\providecommand \urlprefix  [0]{URL }%
\providecommand \Eprint [0]{\href }%
\providecommand \doibase [0]{https://doi.org/}%
\providecommand \selectlanguage [0]{\@gobble}%
\providecommand \bibinfo  [0]{\@secondoftwo}%
\providecommand \bibfield  [0]{\@secondoftwo}%
\providecommand \translation [1]{[#1]}%
\providecommand \BibitemOpen [0]{}%
\providecommand \bibitemStop [0]{}%
\providecommand \bibitemNoStop [0]{.\EOS\space}%
\providecommand \EOS [0]{\spacefactor3000\relax}%
\providecommand \BibitemShut  [1]{\csname bibitem#1\endcsname}%
\let\auto@bib@innerbib\@empty
%</preamble>
\bibitem [{\citenamefont {Bravyi}\ \emph {et~al.}(2021)\citenamefont {Bravyi},
  \citenamefont {Chowdhury}, \citenamefont {Gosset},\ and\ \citenamefont
  {Wocjan}}]{bravyi2021complexity}%
  \BibitemOpen
  \bibfield  {author} {\bibinfo {author} {\bibfnamefont {S.}~\bibnamefont
  {Bravyi}}, \bibinfo {author} {\bibfnamefont {A.}~\bibnamefont {Chowdhury}},
  \bibinfo {author} {\bibfnamefont {D.}~\bibnamefont {Gosset}},\ and\ \bibinfo
  {author} {\bibfnamefont {P.}~\bibnamefont {Wocjan}},\ }\href@noop {}
  {\bibfield  {journal} {\bibinfo  {journal} {arXiv preprint arXiv:2110.15466}\
  } (\bibinfo {year} {2021})}\BibitemShut {NoStop}%
\bibitem [{\citenamefont {Anschuetz}\ \emph {et~al.}(2025)\citenamefont
  {Anschuetz}, \citenamefont {Chen}, \citenamefont {Kiani},\ and\ \citenamefont
  {King}}]{Anschuetz2025}%
  \BibitemOpen
  \bibfield  {author} {\bibinfo {author} {\bibfnamefont {E.~R.}\ \bibnamefont
  {Anschuetz}}, \bibinfo {author} {\bibfnamefont {C.-F.}\ \bibnamefont {Chen}},
  \bibinfo {author} {\bibfnamefont {B.~T.}\ \bibnamefont {Kiani}},\ and\
  \bibinfo {author} {\bibfnamefont {R.}~\bibnamefont {King}},\ }\bibfield
  {journal} {\bibinfo  {journal} {Physical Review Letters}\ }\textbf {\bibinfo
  {volume} {135}},\ \href {https://doi.org/10.1103/cbqf-d24r}
  {10.1103/cbqf-d24r} (\bibinfo {year} {2025})\BibitemShut {NoStop}%
\bibitem [{\citenamefont {Gily{\'e}n}\ \emph {et~al.}(2021)\citenamefont
  {Gily{\'e}n}, \citenamefont {Hastings},\ and\ \citenamefont
  {Vazirani}}]{gilyen2021sub}%
  \BibitemOpen
  \bibfield  {author} {\bibinfo {author} {\bibfnamefont {A.}~\bibnamefont
  {Gily{\'e}n}}, \bibinfo {author} {\bibfnamefont {M.~B.}\ \bibnamefont
  {Hastings}},\ and\ \bibinfo {author} {\bibfnamefont {U.}~\bibnamefont
  {Vazirani}},\ }in\ \href@noop {} {\emph {\bibinfo {booktitle} {Proceedings of
  the 53rd Annual ACM SIGACT Symposium on Theory of Computing}}}\ (\bibinfo
  {year} {2021})\ pp.\ \bibinfo {pages} {1357--1369}\BibitemShut {NoStop}%
\bibitem [{\citenamefont {Leng}\ \emph {et~al.}(2025)\citenamefont {Leng},
  \citenamefont {Wu}, \citenamefont {Wu},\ and\ \citenamefont
  {Zheng}}]{leng2025sub}%
  \BibitemOpen
  \bibfield  {author} {\bibinfo {author} {\bibfnamefont {J.}~\bibnamefont
  {Leng}}, \bibinfo {author} {\bibfnamefont {K.}~\bibnamefont {Wu}}, \bibinfo
  {author} {\bibfnamefont {X.}~\bibnamefont {Wu}},\ and\ \bibinfo {author}
  {\bibfnamefont {Y.}~\bibnamefont {Zheng}},\ }\href@noop {} {\bibfield
  {journal} {\bibinfo  {journal} {arXiv preprint arXiv:2504.14841}\ } (\bibinfo
  {year} {2025})}\BibitemShut {NoStop}%
\bibitem [{\citenamefont {Basso}\ \emph {et~al.}(2024)\citenamefont {Basso},
  \citenamefont {Chen},\ and\ \citenamefont {Dalzell}}]{basso2024optimizing}%
  \BibitemOpen
  \bibfield  {author} {\bibinfo {author} {\bibfnamefont {J.}~\bibnamefont
  {Basso}}, \bibinfo {author} {\bibfnamefont {C.-F.}\ \bibnamefont {Chen}},\
  and\ \bibinfo {author} {\bibfnamefont {A.~M.}\ \bibnamefont {Dalzell}},\
  }\href@noop {} {\bibfield  {journal} {\bibinfo  {journal} {arXiv preprint
  arXiv:2411.02578}\ } (\bibinfo {year} {2024})}\BibitemShut {NoStop}%
\bibitem [{\citenamefont {Mi}\ \emph {et~al.}(2024)\citenamefont {Mi},
  \citenamefont {Michailidis}, \citenamefont {Shabani}, \citenamefont {Miao},
  \citenamefont {Klimov}, \citenamefont {Lloyd}, \citenamefont {Rosenberg},
  \citenamefont {Acharya}, \citenamefont {Aleiner}, \citenamefont {Andersen},
  \citenamefont {Ansmann}, \citenamefont {Arute}, \citenamefont {Arya},
  \citenamefont {Asfaw}, \citenamefont {Atalaya}, \citenamefont {Bardin},
  \citenamefont {Bengtsson}, \citenamefont {Bortoli}, \citenamefont {Bourassa},
  \citenamefont {Bovaird}, \citenamefont {Brill}, \citenamefont {Broughton},
  \citenamefont {Buckley}, \citenamefont {Buell}, \citenamefont {Burger},
  \citenamefont {Burkett}, \citenamefont {Bushnell}, \citenamefont {Chen},
  \citenamefont {Chiaro}, \citenamefont {Chik}, \citenamefont {Chou},
  \citenamefont {Cogan}, \citenamefont {Collins}, \citenamefont {Conner},
  \citenamefont {Courtney}, \citenamefont {Crook}, \citenamefont {Curtin},
  \citenamefont {Dau}, \citenamefont {Debroy}, \citenamefont {Del Toro~Barba},
  \citenamefont {Demura}, \citenamefont {Di~Paolo}, \citenamefont {Drozdov},
  \citenamefont {Dunsworth}, \citenamefont {Erickson}, \citenamefont {Faoro},
  \citenamefont {Farhi}, \citenamefont {Fatemi}, \citenamefont {Ferreira},
  \citenamefont {Burgos}, \citenamefont {Forati}, \citenamefont {Fowler},
  \citenamefont {Foxen}, \citenamefont {Genois}, \citenamefont {Giang},
  \citenamefont {Gidney}, \citenamefont {Gilboa}, \citenamefont {Giustina},
  \citenamefont {Gosula}, \citenamefont {Gross}, \citenamefont {Habegger},
  \citenamefont {Hamilton}, \citenamefont {Hansen}, \citenamefont {Harrigan},
  \citenamefont {Harrington}, \citenamefont {Heu}, \citenamefont {Hoffmann},
  \citenamefont {Hong}, \citenamefont {Huang}, \citenamefont {Huff},
  \citenamefont {Huggins}, \citenamefont {Ioffe}, \citenamefont {Isakov},
  \citenamefont {Iveland}, \citenamefont {Jeffrey}, \citenamefont {Jiang},
  \citenamefont {Jones}, \citenamefont {Juhas}, \citenamefont {Kafri},
  \citenamefont {Kechedzhi}, \citenamefont {Khattar}, \citenamefont {Khezri},
  \citenamefont {Kieferov{\'a}}, \citenamefont {Kim}, \citenamefont {Kitaev},
  \citenamefont {Klots}, \citenamefont {Korotkov}, \citenamefont {Kostritsa},
  \citenamefont {Kreikebaum}, \citenamefont {Landhuis}, \citenamefont {Laptev},
  \citenamefont {Lau}, \citenamefont {Laws}, \citenamefont {Lee}, \citenamefont
  {Lee}, \citenamefont {Lensky}, \citenamefont {Lester}, \citenamefont {Lill},
  \citenamefont {Liu}, \citenamefont {Locharla}, \citenamefont {Malone},
  \citenamefont {Martin}, \citenamefont {McClean}, \citenamefont {McEwen},
  \citenamefont {Mieszala}, \citenamefont {Montazeri}, \citenamefont {Morvan},
  \citenamefont {Movassagh}, \citenamefont {Mruczkiewicz}, \citenamefont
  {Neeley}, \citenamefont {Neill}, \citenamefont {Nersisyan}, \citenamefont
  {Newman}, \citenamefont {Ng}, \citenamefont {Nguyen}, \citenamefont {Nguyen},
  \citenamefont {Niu}, \citenamefont {{O'Brien}}, \citenamefont {Opremcak},
  \citenamefont {Petukhov}, \citenamefont {Potter}, \citenamefont {Pryadko},
  \citenamefont {Quintana}, \citenamefont {Rocque}, \citenamefont {Rubin},
  \citenamefont {Saei}, \citenamefont {Sank}, \citenamefont {Sankaragomathi},
  \citenamefont {Satzinger}, \citenamefont {Schurkus}, \citenamefont
  {Schuster}, \citenamefont {Shearn}, \citenamefont {Shorter}, \citenamefont
  {Shutty}, \citenamefont {Shvarts}, \citenamefont {Skruzny}, \citenamefont
  {Smith}, \citenamefont {Somma}, \citenamefont {Sterling}, \citenamefont
  {Strain}, \citenamefont {Szalay}, \citenamefont {Torres}, \citenamefont
  {Vidal}, \citenamefont {Villalonga}, \citenamefont {Heidweiller},
  \citenamefont {White}, \citenamefont {Woo}, \citenamefont {Xing},
  \citenamefont {Yao}, \citenamefont {Yeh}, \citenamefont {Yoo}, \citenamefont
  {Young}, \citenamefont {Zalcman}, \citenamefont {Zhang}, \citenamefont {Zhu},
  \citenamefont {Zobrist}, \citenamefont {Neven}, \citenamefont {Babbush},
  \citenamefont {Bacon}, \citenamefont {Boixo}, \citenamefont {Hilton},
  \citenamefont {Lucero}, \citenamefont {Megrant}, \citenamefont {Kelly},
  \citenamefont {Chen}, \citenamefont {Roushan}, \citenamefont {Smelyanskiy},\
  and\ \citenamefont {Abanin}}]{Mi2024}%
  \BibitemOpen
  \bibfield  {author} {\bibinfo {author} {\bibfnamefont {X.}~\bibnamefont
  {Mi}}, \bibinfo {author} {\bibfnamefont {A.~A.}\ \bibnamefont {Michailidis}},
  \bibinfo {author} {\bibfnamefont {S.}~\bibnamefont {Shabani}}, \bibinfo
  {author} {\bibfnamefont {K.~C.}\ \bibnamefont {Miao}}, \bibinfo {author}
  {\bibfnamefont {P.~V.}\ \bibnamefont {Klimov}}, \bibinfo {author}
  {\bibfnamefont {J.}~\bibnamefont {Lloyd}}, \bibinfo {author} {\bibfnamefont
  {E.}~\bibnamefont {Rosenberg}}, \bibinfo {author} {\bibfnamefont
  {R.}~\bibnamefont {Acharya}}, \bibinfo {author} {\bibfnamefont
  {I.}~\bibnamefont {Aleiner}}, \bibinfo {author} {\bibfnamefont {T.~I.}\
  \bibnamefont {Andersen}}, \bibinfo {author} {\bibfnamefont {M.}~\bibnamefont
  {Ansmann}}, \bibinfo {author} {\bibfnamefont {F.}~\bibnamefont {Arute}},
  \bibinfo {author} {\bibfnamefont {K.}~\bibnamefont {Arya}}, \bibinfo {author}
  {\bibfnamefont {A.}~\bibnamefont {Asfaw}}, \bibinfo {author} {\bibfnamefont
  {J.}~\bibnamefont {Atalaya}}, \bibinfo {author} {\bibfnamefont {J.~C.}\
  \bibnamefont {Bardin}}, \bibinfo {author} {\bibfnamefont {A.}~\bibnamefont
  {Bengtsson}}, \bibinfo {author} {\bibfnamefont {G.}~\bibnamefont {Bortoli}},
  \bibinfo {author} {\bibfnamefont {A.}~\bibnamefont {Bourassa}}, \bibinfo
  {author} {\bibfnamefont {J.}~\bibnamefont {Bovaird}}, \bibinfo {author}
  {\bibfnamefont {L.}~\bibnamefont {Brill}}, \bibinfo {author} {\bibfnamefont
  {M.}~\bibnamefont {Broughton}}, \bibinfo {author} {\bibfnamefont {B.~B.}\
  \bibnamefont {Buckley}}, \bibinfo {author} {\bibfnamefont {D.~A.}\
  \bibnamefont {Buell}}, \bibinfo {author} {\bibfnamefont {T.}~\bibnamefont
  {Burger}}, \bibinfo {author} {\bibfnamefont {B.}~\bibnamefont {Burkett}},
  \bibinfo {author} {\bibfnamefont {N.}~\bibnamefont {Bushnell}}, \bibinfo
  {author} {\bibfnamefont {Z.}~\bibnamefont {Chen}}, \bibinfo {author}
  {\bibfnamefont {B.}~\bibnamefont {Chiaro}}, \bibinfo {author} {\bibfnamefont
  {D.}~\bibnamefont {Chik}}, \bibinfo {author} {\bibfnamefont {C.}~\bibnamefont
  {Chou}}, \bibinfo {author} {\bibfnamefont {J.}~\bibnamefont {Cogan}},
  \bibinfo {author} {\bibfnamefont {R.}~\bibnamefont {Collins}}, \bibinfo
  {author} {\bibfnamefont {P.}~\bibnamefont {Conner}}, \bibinfo {author}
  {\bibfnamefont {W.}~\bibnamefont {Courtney}}, \bibinfo {author}
  {\bibfnamefont {A.~L.}\ \bibnamefont {Crook}}, \bibinfo {author}
  {\bibfnamefont {B.}~\bibnamefont {Curtin}}, \bibinfo {author} {\bibfnamefont
  {A.~G.}\ \bibnamefont {Dau}}, \bibinfo {author} {\bibfnamefont {D.~M.}\
  \bibnamefont {Debroy}}, \bibinfo {author} {\bibfnamefont {A.}~\bibnamefont
  {Del Toro~Barba}}, \bibinfo {author} {\bibfnamefont {S.}~\bibnamefont
  {Demura}}, \bibinfo {author} {\bibfnamefont {A.}~\bibnamefont {Di~Paolo}},
  \bibinfo {author} {\bibfnamefont {I.~K.}\ \bibnamefont {Drozdov}}, \bibinfo
  {author} {\bibfnamefont {A.}~\bibnamefont {Dunsworth}}, \bibinfo {author}
  {\bibfnamefont {C.}~\bibnamefont {Erickson}}, \bibinfo {author}
  {\bibfnamefont {L.}~\bibnamefont {Faoro}}, \bibinfo {author} {\bibfnamefont
  {E.}~\bibnamefont {Farhi}}, \bibinfo {author} {\bibfnamefont
  {R.}~\bibnamefont {Fatemi}}, \bibinfo {author} {\bibfnamefont {V.~S.}\
  \bibnamefont {Ferreira}}, \bibinfo {author} {\bibfnamefont {L.~F.}\
  \bibnamefont {Burgos}}, \bibinfo {author} {\bibfnamefont {E.}~\bibnamefont
  {Forati}}, \bibinfo {author} {\bibfnamefont {A.~G.}\ \bibnamefont {Fowler}},
  \bibinfo {author} {\bibfnamefont {B.}~\bibnamefont {Foxen}}, \bibinfo
  {author} {\bibfnamefont {{\'E}.}~\bibnamefont {Genois}}, \bibinfo {author}
  {\bibfnamefont {W.}~\bibnamefont {Giang}}, \bibinfo {author} {\bibfnamefont
  {C.}~\bibnamefont {Gidney}}, \bibinfo {author} {\bibfnamefont
  {D.}~\bibnamefont {Gilboa}}, \bibinfo {author} {\bibfnamefont
  {M.}~\bibnamefont {Giustina}}, \bibinfo {author} {\bibfnamefont
  {R.}~\bibnamefont {Gosula}}, \bibinfo {author} {\bibfnamefont {J.~A.}\
  \bibnamefont {Gross}}, \bibinfo {author} {\bibfnamefont {S.}~\bibnamefont
  {Habegger}}, \bibinfo {author} {\bibfnamefont {M.~C.}\ \bibnamefont
  {Hamilton}}, \bibinfo {author} {\bibfnamefont {M.}~\bibnamefont {Hansen}},
  \bibinfo {author} {\bibfnamefont {M.~P.}\ \bibnamefont {Harrigan}}, \bibinfo
  {author} {\bibfnamefont {S.~D.}\ \bibnamefont {Harrington}}, \bibinfo
  {author} {\bibfnamefont {P.}~\bibnamefont {Heu}}, \bibinfo {author}
  {\bibfnamefont {M.~R.}\ \bibnamefont {Hoffmann}}, \bibinfo {author}
  {\bibfnamefont {S.}~\bibnamefont {Hong}}, \bibinfo {author} {\bibfnamefont
  {T.}~\bibnamefont {Huang}}, \bibinfo {author} {\bibfnamefont
  {A.}~\bibnamefont {Huff}}, \bibinfo {author} {\bibfnamefont {W.~J.}\
  \bibnamefont {Huggins}}, \bibinfo {author} {\bibfnamefont {L.~B.}\
  \bibnamefont {Ioffe}}, \bibinfo {author} {\bibfnamefont {S.~V.}\ \bibnamefont
  {Isakov}}, \bibinfo {author} {\bibfnamefont {J.}~\bibnamefont {Iveland}},
  \bibinfo {author} {\bibfnamefont {E.}~\bibnamefont {Jeffrey}}, \bibinfo
  {author} {\bibfnamefont {Z.}~\bibnamefont {Jiang}}, \bibinfo {author}
  {\bibfnamefont {C.}~\bibnamefont {Jones}}, \bibinfo {author} {\bibfnamefont
  {P.}~\bibnamefont {Juhas}}, \bibinfo {author} {\bibfnamefont
  {D.}~\bibnamefont {Kafri}}, \bibinfo {author} {\bibfnamefont
  {K.}~\bibnamefont {Kechedzhi}}, \bibinfo {author} {\bibfnamefont
  {T.}~\bibnamefont {Khattar}}, \bibinfo {author} {\bibfnamefont
  {M.}~\bibnamefont {Khezri}}, \bibinfo {author} {\bibfnamefont
  {M.}~\bibnamefont {Kieferov{\'a}}}, \bibinfo {author} {\bibfnamefont
  {S.}~\bibnamefont {Kim}}, \bibinfo {author} {\bibfnamefont {A.}~\bibnamefont
  {Kitaev}}, \bibinfo {author} {\bibfnamefont {A.~R.}\ \bibnamefont {Klots}},
  \bibinfo {author} {\bibfnamefont {A.~N.}\ \bibnamefont {Korotkov}}, \bibinfo
  {author} {\bibfnamefont {F.}~\bibnamefont {Kostritsa}}, \bibinfo {author}
  {\bibfnamefont {J.~M.}\ \bibnamefont {Kreikebaum}}, \bibinfo {author}
  {\bibfnamefont {D.}~\bibnamefont {Landhuis}}, \bibinfo {author}
  {\bibfnamefont {P.}~\bibnamefont {Laptev}}, \bibinfo {author} {\bibfnamefont
  {K.-M.}\ \bibnamefont {Lau}}, \bibinfo {author} {\bibfnamefont
  {L.}~\bibnamefont {Laws}}, \bibinfo {author} {\bibfnamefont {J.}~\bibnamefont
  {Lee}}, \bibinfo {author} {\bibfnamefont {K.~W.}\ \bibnamefont {Lee}},
  \bibinfo {author} {\bibfnamefont {Y.~D.}\ \bibnamefont {Lensky}}, \bibinfo
  {author} {\bibfnamefont {B.~J.}\ \bibnamefont {Lester}}, \bibinfo {author}
  {\bibfnamefont {A.~T.}\ \bibnamefont {Lill}}, \bibinfo {author}
  {\bibfnamefont {W.}~\bibnamefont {Liu}}, \bibinfo {author} {\bibfnamefont
  {A.}~\bibnamefont {Locharla}}, \bibinfo {author} {\bibfnamefont {F.~D.}\
  \bibnamefont {Malone}}, \bibinfo {author} {\bibfnamefont {O.}~\bibnamefont
  {Martin}}, \bibinfo {author} {\bibfnamefont {J.~R.}\ \bibnamefont {McClean}},
  \bibinfo {author} {\bibfnamefont {M.}~\bibnamefont {McEwen}}, \bibinfo
  {author} {\bibfnamefont {A.}~\bibnamefont {Mieszala}}, \bibinfo {author}
  {\bibfnamefont {S.}~\bibnamefont {Montazeri}}, \bibinfo {author}
  {\bibfnamefont {A.}~\bibnamefont {Morvan}}, \bibinfo {author} {\bibfnamefont
  {R.}~\bibnamefont {Movassagh}}, \bibinfo {author} {\bibfnamefont
  {W.}~\bibnamefont {Mruczkiewicz}}, \bibinfo {author} {\bibfnamefont
  {M.}~\bibnamefont {Neeley}}, \bibinfo {author} {\bibfnamefont
  {C.}~\bibnamefont {Neill}}, \bibinfo {author} {\bibfnamefont
  {A.}~\bibnamefont {Nersisyan}}, \bibinfo {author} {\bibfnamefont
  {M.}~\bibnamefont {Newman}}, \bibinfo {author} {\bibfnamefont {J.~H.}\
  \bibnamefont {Ng}}, \bibinfo {author} {\bibfnamefont {A.}~\bibnamefont
  {Nguyen}}, \bibinfo {author} {\bibfnamefont {M.}~\bibnamefont {Nguyen}},
  \bibinfo {author} {\bibfnamefont {M.~Y.}\ \bibnamefont {Niu}}, \bibinfo
  {author} {\bibfnamefont {T.~E.}\ \bibnamefont {{O'Brien}}}, \bibinfo {author}
  {\bibfnamefont {A.}~\bibnamefont {Opremcak}}, \bibinfo {author}
  {\bibfnamefont {A.}~\bibnamefont {Petukhov}}, \bibinfo {author}
  {\bibfnamefont {R.}~\bibnamefont {Potter}}, \bibinfo {author} {\bibfnamefont
  {L.~P.}\ \bibnamefont {Pryadko}}, \bibinfo {author} {\bibfnamefont
  {C.}~\bibnamefont {Quintana}}, \bibinfo {author} {\bibfnamefont
  {C.}~\bibnamefont {Rocque}}, \bibinfo {author} {\bibfnamefont {N.~C.}\
  \bibnamefont {Rubin}}, \bibinfo {author} {\bibfnamefont {N.}~\bibnamefont
  {Saei}}, \bibinfo {author} {\bibfnamefont {D.}~\bibnamefont {Sank}}, \bibinfo
  {author} {\bibfnamefont {K.}~\bibnamefont {Sankaragomathi}}, \bibinfo
  {author} {\bibfnamefont {K.~J.}\ \bibnamefont {Satzinger}}, \bibinfo {author}
  {\bibfnamefont {H.~F.}\ \bibnamefont {Schurkus}}, \bibinfo {author}
  {\bibfnamefont {C.}~\bibnamefont {Schuster}}, \bibinfo {author}
  {\bibfnamefont {M.~J.}\ \bibnamefont {Shearn}}, \bibinfo {author}
  {\bibfnamefont {A.}~\bibnamefont {Shorter}}, \bibinfo {author} {\bibfnamefont
  {N.}~\bibnamefont {Shutty}}, \bibinfo {author} {\bibfnamefont
  {V.}~\bibnamefont {Shvarts}}, \bibinfo {author} {\bibfnamefont
  {J.}~\bibnamefont {Skruzny}}, \bibinfo {author} {\bibfnamefont {W.~C.}\
  \bibnamefont {Smith}}, \bibinfo {author} {\bibfnamefont {R.}~\bibnamefont
  {Somma}}, \bibinfo {author} {\bibfnamefont {G.}~\bibnamefont {Sterling}},
  \bibinfo {author} {\bibfnamefont {D.}~\bibnamefont {Strain}}, \bibinfo
  {author} {\bibfnamefont {M.}~\bibnamefont {Szalay}}, \bibinfo {author}
  {\bibfnamefont {A.}~\bibnamefont {Torres}}, \bibinfo {author} {\bibfnamefont
  {G.}~\bibnamefont {Vidal}}, \bibinfo {author} {\bibfnamefont
  {B.}~\bibnamefont {Villalonga}}, \bibinfo {author} {\bibfnamefont {C.~V.}\
  \bibnamefont {Heidweiller}}, \bibinfo {author} {\bibfnamefont
  {T.}~\bibnamefont {White}}, \bibinfo {author} {\bibfnamefont {B.~W.~K.}\
  \bibnamefont {Woo}}, \bibinfo {author} {\bibfnamefont {C.}~\bibnamefont
  {Xing}}, \bibinfo {author} {\bibfnamefont {Z.~J.}\ \bibnamefont {Yao}},
  \bibinfo {author} {\bibfnamefont {P.}~\bibnamefont {Yeh}}, \bibinfo {author}
  {\bibfnamefont {J.}~\bibnamefont {Yoo}}, \bibinfo {author} {\bibfnamefont
  {G.}~\bibnamefont {Young}}, \bibinfo {author} {\bibfnamefont
  {A.}~\bibnamefont {Zalcman}}, \bibinfo {author} {\bibfnamefont
  {Y.}~\bibnamefont {Zhang}}, \bibinfo {author} {\bibfnamefont
  {N.}~\bibnamefont {Zhu}}, \bibinfo {author} {\bibfnamefont {N.}~\bibnamefont
  {Zobrist}}, \bibinfo {author} {\bibfnamefont {H.}~\bibnamefont {Neven}},
  \bibinfo {author} {\bibfnamefont {R.}~\bibnamefont {Babbush}}, \bibinfo
  {author} {\bibfnamefont {D.}~\bibnamefont {Bacon}}, \bibinfo {author}
  {\bibfnamefont {S.}~\bibnamefont {Boixo}}, \bibinfo {author} {\bibfnamefont
  {J.}~\bibnamefont {Hilton}}, \bibinfo {author} {\bibfnamefont
  {E.}~\bibnamefont {Lucero}}, \bibinfo {author} {\bibfnamefont
  {A.}~\bibnamefont {Megrant}}, \bibinfo {author} {\bibfnamefont
  {J.}~\bibnamefont {Kelly}}, \bibinfo {author} {\bibfnamefont
  {Y.}~\bibnamefont {Chen}}, \bibinfo {author} {\bibfnamefont {P.}~\bibnamefont
  {Roushan}}, \bibinfo {author} {\bibfnamefont {V.}~\bibnamefont
  {Smelyanskiy}},\ and\ \bibinfo {author} {\bibfnamefont {D.~A.}\ \bibnamefont
  {Abanin}},\ }\href {https://doi.org/10.1126/science.adh9932} {\bibfield
  {journal} {\bibinfo  {journal} {Science}\ }\textbf {\bibinfo {volume}
  {383}},\ \bibinfo {pages} {1332} (\bibinfo {year} {2024})}\BibitemShut
  {NoStop}%
\bibitem [{\citenamefont {Chen}\ \emph
  {et~al.}(2023{\natexlab{a}})\citenamefont {Chen}, \citenamefont {Kastoryano},
  \citenamefont {Brand{\~a}o},\ and\ \citenamefont
  {Gily{\'e}n}}]{chen2023quantum}%
  \BibitemOpen
  \bibfield  {author} {\bibinfo {author} {\bibfnamefont {C.-F.}\ \bibnamefont
  {Chen}}, \bibinfo {author} {\bibfnamefont {M.~J.}\ \bibnamefont
  {Kastoryano}}, \bibinfo {author} {\bibfnamefont {F.~G.}\ \bibnamefont
  {Brand{\~a}o}},\ and\ \bibinfo {author} {\bibfnamefont {A.}~\bibnamefont
  {Gily{\'e}n}},\ }\href@noop {} {\bibfield  {journal} {\bibinfo  {journal}
  {arXiv preprint arXiv:2303.18224}\ } (\bibinfo {year}
  {2023}{\natexlab{a}})}\BibitemShut {NoStop}%
\bibitem [{\citenamefont {Gily{\'e}n}\ \emph {et~al.}(2024)\citenamefont
  {Gily{\'e}n}, \citenamefont {Chen}, \citenamefont {Doriguello},\ and\
  \citenamefont {Kastoryano}}]{gilyen2024quantum}%
  \BibitemOpen
  \bibfield  {author} {\bibinfo {author} {\bibfnamefont {A.}~\bibnamefont
  {Gily{\'e}n}}, \bibinfo {author} {\bibfnamefont {C.-F.}\ \bibnamefont
  {Chen}}, \bibinfo {author} {\bibfnamefont {J.~F.}\ \bibnamefont
  {Doriguello}},\ and\ \bibinfo {author} {\bibfnamefont {M.~J.}\ \bibnamefont
  {Kastoryano}},\ }\href@noop {} {\bibfield  {journal} {\bibinfo  {journal}
  {arXiv preprint arXiv:2405.20322}\ } (\bibinfo {year} {2024})}\BibitemShut
  {NoStop}%
\bibitem [{\citenamefont {Chen}\ \emph
  {et~al.}(2023{\natexlab{b}})\citenamefont {Chen}, \citenamefont
  {Kastoryano},\ and\ \citenamefont {Gily{\'e}n}}]{chen2023efficient}%
  \BibitemOpen
  \bibfield  {author} {\bibinfo {author} {\bibfnamefont {C.-F.}\ \bibnamefont
  {Chen}}, \bibinfo {author} {\bibfnamefont {M.~J.}\ \bibnamefont
  {Kastoryano}},\ and\ \bibinfo {author} {\bibfnamefont {A.}~\bibnamefont
  {Gily{\'e}n}},\ }\href@noop {} {\bibfield  {journal} {\bibinfo  {journal}
  {arXiv preprint arXiv:2311.09207}\ } (\bibinfo {year}
  {2023}{\natexlab{b}})}\BibitemShut {NoStop}%
\bibitem [{\citenamefont {Ding}\ \emph
  {et~al.}(2025{\natexlab{a}})\citenamefont {Ding}, \citenamefont {Li},\ and\
  \citenamefont {Lin}}]{ding2025efficient}%
  \BibitemOpen
  \bibfield  {author} {\bibinfo {author} {\bibfnamefont {Z.}~\bibnamefont
  {Ding}}, \bibinfo {author} {\bibfnamefont {B.}~\bibnamefont {Li}},\ and\
  \bibinfo {author} {\bibfnamefont {L.}~\bibnamefont {Lin}},\ }\href
  {https://doi.org/10.1007/s00220-025-05235-3} {\bibfield  {journal} {\bibinfo
  {journal} {Communications in Mathematical Physics}\ }\textbf {\bibinfo
  {volume} {406}},\ \bibinfo {pages} {67} (\bibinfo {year}
  {2025}{\natexlab{a}})}\BibitemShut {NoStop}%
\bibitem [{\citenamefont {Rouz\'{e}}\ \emph {et~al.}(2026)\citenamefont
  {Rouz\'{e}}, \citenamefont {Stilck~Fran\c{c}a},\ and\ \citenamefont
  {Alhambra}}]{rouze2411optimal}%
  \BibitemOpen
  \bibfield  {author} {\bibinfo {author} {\bibfnamefont {C.}~\bibnamefont
  {Rouz\'{e}}}, \bibinfo {author} {\bibfnamefont {D.}~\bibnamefont
  {Stilck~Fran\c{c}a}},\ and\ \bibinfo {author} {\bibfnamefont {A.~M.}\
  \bibnamefont {Alhambra}},\ }\bibfield  {journal} {\bibinfo  {journal}
  {Physical Review Letters}\ }\textbf {\bibinfo {volume} {136}},\ \href
  {https://doi.org/10.1103/lhht-svmn} {10.1103/lhht-svmn} (\bibinfo {year}
  {2026})\BibitemShut {NoStop}%
\bibitem [{\citenamefont {Rouz\'{e}}\ \emph {et~al.}(2025)\citenamefont
  {Rouz\'{e}}, \citenamefont {Fran\c{c}a},\ and\ \citenamefont
  {Alhambra}}]{rouze2024efficient}%
  \BibitemOpen
  \bibfield  {author} {\bibinfo {author} {\bibfnamefont {C.}~\bibnamefont
  {Rouz\'{e}}}, \bibinfo {author} {\bibfnamefont {D.~S.}\ \bibnamefont
  {Fran\c{c}a}},\ and\ \bibinfo {author} {\bibfnamefont {A.~M.}\ \bibnamefont
  {Alhambra}},\ }in\ \href {https://doi.org/10.1145/3717823.3718268} {\emph
  {\bibinfo {booktitle} {Proceedings of the 57th Annual ACM Symposium on Theory
  of Computing}}},\ \bibinfo {series and number} {STOC ’25}\ (\bibinfo
  {publisher} {ACM},\ \bibinfo {year} {2025})\ p.\ \bibinfo {pages}
  {1488–1495}\BibitemShut {NoStop}%
\bibitem [{\citenamefont {{\v S}m{\'\i}d}\ \emph {et~al.}(2025)\citenamefont
  {{\v S}m{\'\i}d}, \citenamefont {Meister}, \citenamefont {Berta},\ and\
  \citenamefont {Bondesan}}]{vsmid2025polynomial}%
  \BibitemOpen
  \bibfield  {author} {\bibinfo {author} {\bibfnamefont {{\v S}.}~\bibnamefont
  {{\v S}m{\'\i}d}}, \bibinfo {author} {\bibfnamefont {R.}~\bibnamefont
  {Meister}}, \bibinfo {author} {\bibfnamefont {M.}~\bibnamefont {Berta}},\
  and\ \bibinfo {author} {\bibfnamefont {R.}~\bibnamefont {Bondesan}},\ }\href
  {https://doi.org/10.1038/s41467-025-65765-1} {\bibfield  {journal} {\bibinfo
  {journal} {Nature Communications}\ }\textbf {\bibinfo {volume} {16}},\
  \bibinfo {pages} {10736} (\bibinfo {year} {2025})}\BibitemShut {NoStop}%
\bibitem [{\citenamefont {{\v{S}}m{\'\i}d}\ \emph {et~al.}(2025)\citenamefont
  {{\v{S}}m{\'\i}d}, \citenamefont {Meister}, \citenamefont {Berta},\ and\
  \citenamefont {Bondesan}}]{vsmid2025rapid}%
  \BibitemOpen
  \bibfield  {author} {\bibinfo {author} {\bibfnamefont {{\v{S}}.}~\bibnamefont
  {{\v{S}}m{\'\i}d}}, \bibinfo {author} {\bibfnamefont {R.}~\bibnamefont
  {Meister}}, \bibinfo {author} {\bibfnamefont {M.}~\bibnamefont {Berta}},\
  and\ \bibinfo {author} {\bibfnamefont {R.}~\bibnamefont {Bondesan}},\
  }\href@noop {} {\bibfield  {journal} {\bibinfo  {journal} {arXiv preprint
  arXiv:2510.04954}\ } (\bibinfo {year} {2025})}\BibitemShut {NoStop}%
\bibitem [{\citenamefont {Tong}\ and\ \citenamefont
  {Zhan}(2025)}]{tong2025fast}%
  \BibitemOpen
  \bibfield  {author} {\bibinfo {author} {\bibfnamefont {Y.}~\bibnamefont
  {Tong}}\ and\ \bibinfo {author} {\bibfnamefont {Y.}~\bibnamefont {Zhan}},\
  }\href@noop {} {\bibfield  {journal} {\bibinfo  {journal} {PRX Quantum}\
  }\textbf {\bibinfo {volume} {6}},\ \bibinfo {pages} {030301} (\bibinfo {year}
  {2025})}\BibitemShut {NoStop}%
\bibitem [{\citenamefont {Ding}\ \emph
  {et~al.}(2025{\natexlab{b}})\citenamefont {Ding}, \citenamefont {Zhan},
  \citenamefont {Preskill},\ and\ \citenamefont {Lin}}]{ding2025end}%
  \BibitemOpen
  \bibfield  {author} {\bibinfo {author} {\bibfnamefont {Z.}~\bibnamefont
  {Ding}}, \bibinfo {author} {\bibfnamefont {Y.}~\bibnamefont {Zhan}}, \bibinfo
  {author} {\bibfnamefont {J.}~\bibnamefont {Preskill}},\ and\ \bibinfo
  {author} {\bibfnamefont {L.}~\bibnamefont {Lin}},\ }\href@noop {} {\bibfield
  {journal} {\bibinfo  {journal} {arXiv preprint arXiv:2508.05703}\ } (\bibinfo
  {year} {2025}{\natexlab{b}})}\BibitemShut {NoStop}%
\bibitem [{\citenamefont {Hahn}\ \emph {et~al.}(2025)\citenamefont {Hahn},
  \citenamefont {Sweke}, \citenamefont {Deshpande},\ and\ \citenamefont
  {Shtanko}}]{hahn2025efficient}%
  \BibitemOpen
  \bibfield  {author} {\bibinfo {author} {\bibfnamefont {D.}~\bibnamefont
  {Hahn}}, \bibinfo {author} {\bibfnamefont {R.}~\bibnamefont {Sweke}},
  \bibinfo {author} {\bibfnamefont {A.}~\bibnamefont {Deshpande}},\ and\
  \bibinfo {author} {\bibfnamefont {O.}~\bibnamefont {Shtanko}},\ }\href@noop
  {} {\bibfield  {journal} {\bibinfo  {journal} {arXiv preprint
  arXiv:2506.04321}\ } (\bibinfo {year} {2025})}\BibitemShut {NoStop}%
\bibitem [{\citenamefont {Bakshi}\ \emph {et~al.}(2025)\citenamefont {Bakshi},
  \citenamefont {Liu}, \citenamefont {Moitra},\ and\ \citenamefont
  {Tang}}]{bakshi2025dobrushin}%
  \BibitemOpen
  \bibfield  {author} {\bibinfo {author} {\bibfnamefont {A.}~\bibnamefont
  {Bakshi}}, \bibinfo {author} {\bibfnamefont {A.}~\bibnamefont {Liu}},
  \bibinfo {author} {\bibfnamefont {A.}~\bibnamefont {Moitra}},\ and\ \bibinfo
  {author} {\bibfnamefont {E.}~\bibnamefont {Tang}},\ }\href@noop {} {\bibfield
   {journal} {\bibinfo  {journal} {arXiv preprint arXiv:2510.08542}\ }
  (\bibinfo {year} {2025})}\BibitemShut {NoStop}%
\bibitem [{\citenamefont {Bergamaschi}\ and\ \citenamefont
  {Chen}(2025)}]{bergamaschi2025quantum}%
  \BibitemOpen
  \bibfield  {author} {\bibinfo {author} {\bibfnamefont {T.}~\bibnamefont
  {Bergamaschi}}\ and\ \bibinfo {author} {\bibfnamefont {C.-F.}\ \bibnamefont
  {Chen}},\ }\href@noop {} {\bibfield  {journal} {\bibinfo  {journal} {arXiv
  preprint arXiv:2510.08533}\ } (\bibinfo {year} {2025})}\BibitemShut {NoStop}%
\bibitem [{\citenamefont {Zhan}\ \emph {et~al.}(2025)\citenamefont {Zhan},
  \citenamefont {Ding}, \citenamefont {Huhn}, \citenamefont {Gray},
  \citenamefont {Preskill}, \citenamefont {Chan},\ and\ \citenamefont
  {Lin}}]{zhan2025rapid}%
  \BibitemOpen
  \bibfield  {author} {\bibinfo {author} {\bibfnamefont {Y.}~\bibnamefont
  {Zhan}}, \bibinfo {author} {\bibfnamefont {Z.}~\bibnamefont {Ding}}, \bibinfo
  {author} {\bibfnamefont {J.}~\bibnamefont {Huhn}}, \bibinfo {author}
  {\bibfnamefont {J.}~\bibnamefont {Gray}}, \bibinfo {author} {\bibfnamefont
  {J.}~\bibnamefont {Preskill}}, \bibinfo {author} {\bibfnamefont {G.~K.}\
  \bibnamefont {Chan}},\ and\ \bibinfo {author} {\bibfnamefont
  {L.}~\bibnamefont {Lin}},\ }\href@noop {} {\bibfield  {journal} {\bibinfo
  {journal} {arXiv preprint arXiv:2503.15827}\ } (\bibinfo {year}
  {2025})}\BibitemShut {NoStop}%
\bibitem [{\citenamefont {Bakshi}\ \emph {et~al.}(2024)\citenamefont {Bakshi},
  \citenamefont {Liu}, \citenamefont {Moitra},\ and\ \citenamefont
  {Tang}}]{bakshi2024high}%
  \BibitemOpen
  \bibfield  {author} {\bibinfo {author} {\bibfnamefont {A.}~\bibnamefont
  {Bakshi}}, \bibinfo {author} {\bibfnamefont {A.}~\bibnamefont {Liu}},
  \bibinfo {author} {\bibfnamefont {A.}~\bibnamefont {Moitra}},\ and\ \bibinfo
  {author} {\bibfnamefont {E.}~\bibnamefont {Tang}},\ }in\ \href@noop {} {\emph
  {\bibinfo {booktitle} {2024 IEEE 65th Annual Symposium on Foundations of
  Computer Science (FOCS)}}}\ (\bibinfo {organization} {IEEE},\ \bibinfo {year}
  {2024})\ pp.\ \bibinfo {pages} {1027--1036}\BibitemShut {NoStop}%
\bibitem [{\citenamefont {Mann}\ and\ \citenamefont
  {Helmuth}(2021)}]{mann2021efficient}%
  \BibitemOpen
  \bibfield  {author} {\bibinfo {author} {\bibfnamefont {R.~L.}\ \bibnamefont
  {Mann}}\ and\ \bibinfo {author} {\bibfnamefont {T.}~\bibnamefont {Helmuth}},\
  }\href@noop {} {\bibfield  {journal} {\bibinfo  {journal} {Journal of
  Mathematical Physics}\ }\textbf {\bibinfo {volume} {62}},\ \bibinfo {pages}
  {{022201}} (\bibinfo {year} {2021})}\BibitemShut {NoStop}%
\bibitem [{\citenamefont {Helmuth}\ and\ \citenamefont
  {Mann}(2023)}]{helmuth2023efficient}%
  \BibitemOpen
  \bibfield  {author} {\bibinfo {author} {\bibfnamefont {T.}~\bibnamefont
  {Helmuth}}\ and\ \bibinfo {author} {\bibfnamefont {R.~L.}\ \bibnamefont
  {Mann}},\ }\href@noop {} {\bibfield  {journal} {\bibinfo  {journal}
  {Quantum}\ }\textbf {\bibinfo {volume} {7}},\ \bibinfo {pages} {1155}
  (\bibinfo {year} {2023})}\BibitemShut {NoStop}%
\bibitem [{\citenamefont {Mann}\ and\ \citenamefont {Minko}(2024)}]{Mann2024}%
  \BibitemOpen
  \bibfield  {author} {\bibinfo {author} {\bibfnamefont {R.~L.}\ \bibnamefont
  {Mann}}\ and\ \bibinfo {author} {\bibfnamefont {R.~M.}\ \bibnamefont
  {Minko}},\ }\bibfield  {journal} {\bibinfo  {journal} {PRX Quantum}\ }\textbf
  {\bibinfo {volume} {5}},\ \href {https://doi.org/10.1103/prxquantum.5.010305}
  {10.1103/prxquantum.5.010305} (\bibinfo {year} {2024})\BibitemShut {NoStop}%
\bibitem [{\citenamefont {Chen}\ \emph {et~al.}(2025)\citenamefont {Chen},
  \citenamefont {Rouz{\'e}}, \citenamefont {Chen}, \citenamefont {Jiang},
  \citenamefont {Scalet}, \citenamefont {Zhan}, \citenamefont {Chan},
  \citenamefont {Ying},\ and\ \citenamefont {Tong}}]{chen2025convergence}%
  \BibitemOpen
  \bibfield  {author} {\bibinfo {author} {\bibfnamefont {H.}~\bibnamefont
  {Chen}}, \bibinfo {author} {\bibfnamefont {C.}~\bibnamefont {Rouz{\'e}}},
  \bibinfo {author} {\bibfnamefont {J.}~\bibnamefont {Chen}}, \bibinfo {author}
  {\bibfnamefont {J.}~\bibnamefont {Jiang}}, \bibinfo {author} {\bibfnamefont
  {S.~O.}\ \bibnamefont {Scalet}}, \bibinfo {author} {\bibfnamefont
  {Y.}~\bibnamefont {Zhan}}, \bibinfo {author} {\bibfnamefont {G.~K.}\
  \bibnamefont {Chan}}, \bibinfo {author} {\bibfnamefont {L.}~\bibnamefont
  {Ying}},\ and\ \bibinfo {author} {\bibfnamefont {Y.}~\bibnamefont {Tong}},\
  }\href@noop {} {\bibfield  {journal} {\bibinfo  {journal} {arXiv preprint
  arXiv:2512.12010}\ } (\bibinfo {year} {2025})}\BibitemShut {NoStop}%
\bibitem [{\citenamefont {Braunstein}\ and\ \citenamefont
  {Van~Loock}(2005)}]{braunstein2005quantum}%
  \BibitemOpen
  \bibfield  {author} {\bibinfo {author} {\bibfnamefont {S.~L.}\ \bibnamefont
  {Braunstein}}\ and\ \bibinfo {author} {\bibfnamefont {P.}~\bibnamefont
  {Van~Loock}},\ }\href@noop {} {\bibfield  {journal} {\bibinfo  {journal}
  {Reviews of modern physics}\ }\textbf {\bibinfo {volume} {77}},\ \bibinfo
  {pages} {513} (\bibinfo {year} {2005})}\BibitemShut {NoStop}%
\bibitem [{\citenamefont {Cerf}\ \emph {et~al.}(2007)\citenamefont {Cerf},
  \citenamefont {Leuchs},\ and\ \citenamefont {Polzik}}]{cerf2007quantum}%
  \BibitemOpen
  \bibfield  {author} {\bibinfo {author} {\bibfnamefont {N.~J.}\ \bibnamefont
  {Cerf}}, \bibinfo {author} {\bibfnamefont {G.}~\bibnamefont {Leuchs}},\ and\
  \bibinfo {author} {\bibfnamefont {E.~S.}\ \bibnamefont {Polzik}},\
  }\href@noop {} {\emph {\bibinfo {title} {Quantum information with continuous
  variables of atoms and light}}}\ (\bibinfo  {publisher} {World Scientific},\
  \bibinfo {year} {2007})\BibitemShut {NoStop}%
\bibitem [{\citenamefont {Weedbrook}\ \emph {et~al.}(2012)\citenamefont
  {Weedbrook}, \citenamefont {Pirandola}, \citenamefont
  {Garc{\'\i}a-Patr{\'o}n}, \citenamefont {Cerf}, \citenamefont {Ralph},
  \citenamefont {Shapiro},\ and\ \citenamefont
  {Lloyd}}]{weedbrook2012gaussian}%
  \BibitemOpen
  \bibfield  {author} {\bibinfo {author} {\bibfnamefont {C.}~\bibnamefont
  {Weedbrook}}, \bibinfo {author} {\bibfnamefont {S.}~\bibnamefont
  {Pirandola}}, \bibinfo {author} {\bibfnamefont {R.}~\bibnamefont
  {Garc{\'\i}a-Patr{\'o}n}}, \bibinfo {author} {\bibfnamefont {N.~J.}\
  \bibnamefont {Cerf}}, \bibinfo {author} {\bibfnamefont {T.~C.}\ \bibnamefont
  {Ralph}}, \bibinfo {author} {\bibfnamefont {J.~H.}\ \bibnamefont {Shapiro}},\
  and\ \bibinfo {author} {\bibfnamefont {S.}~\bibnamefont {Lloyd}},\
  }\href@noop {} {\bibfield  {journal} {\bibinfo  {journal} {Reviews of Modern
  Physics}\ }\textbf {\bibinfo {volume} {84}},\ \bibinfo {pages} {621}
  (\bibinfo {year} {2012})}\BibitemShut {NoStop}%
\bibitem [{\citenamefont {Chabaud}\ \emph {et~al.}(2024)\citenamefont
  {Chabaud}, \citenamefont {Joseph}, \citenamefont {Mehraban},\ and\
  \citenamefont {Motamedi}}]{chabaud2024bosonic}%
  \BibitemOpen
  \bibfield  {author} {\bibinfo {author} {\bibfnamefont {U.}~\bibnamefont
  {Chabaud}}, \bibinfo {author} {\bibfnamefont {M.}~\bibnamefont {Joseph}},
  \bibinfo {author} {\bibfnamefont {S.}~\bibnamefont {Mehraban}},\ and\
  \bibinfo {author} {\bibfnamefont {A.}~\bibnamefont {Motamedi}},\ }\href@noop
  {} {\bibfield  {journal} {\bibinfo  {journal} {arXiv preprint
  arXiv:2410.04274}\ } (\bibinfo {year} {2024})}\BibitemShut {NoStop}%
\bibitem [{\citenamefont {Chabaud}\ \emph {et~al.}(2025)\citenamefont
  {Chabaud}, \citenamefont {Gharibian}, \citenamefont {Mehraban}, \citenamefont
  {Motamedi}, \citenamefont {Naeij}, \citenamefont {Rudolph},\ and\
  \citenamefont {Sambrani}}]{chabaud2025energy}%
  \BibitemOpen
  \bibfield  {author} {\bibinfo {author} {\bibfnamefont {U.}~\bibnamefont
  {Chabaud}}, \bibinfo {author} {\bibfnamefont {S.}~\bibnamefont {Gharibian}},
  \bibinfo {author} {\bibfnamefont {S.}~\bibnamefont {Mehraban}}, \bibinfo
  {author} {\bibfnamefont {A.}~\bibnamefont {Motamedi}}, \bibinfo {author}
  {\bibfnamefont {H.~R.}\ \bibnamefont {Naeij}}, \bibinfo {author}
  {\bibfnamefont {D.}~\bibnamefont {Rudolph}},\ and\ \bibinfo {author}
  {\bibfnamefont {D.}~\bibnamefont {Sambrani}},\ }\href@noop {} {\bibfield
  {journal} {\bibinfo  {journal} {arXiv preprint arXiv:2510.08545}\ } (\bibinfo
  {year} {2025})}\BibitemShut {NoStop}%
\bibitem [{\citenamefont {Cazalilla}\ \emph {et~al.}(2011)\citenamefont
  {Cazalilla}, \citenamefont {Citro}, \citenamefont {Giamarchi}, \citenamefont
  {Orignac},\ and\ \citenamefont {Rigol}}]{cazalilla2011one}%
  \BibitemOpen
  \bibfield  {author} {\bibinfo {author} {\bibfnamefont {M.~A.}\ \bibnamefont
  {Cazalilla}}, \bibinfo {author} {\bibfnamefont {R.}~\bibnamefont {Citro}},
  \bibinfo {author} {\bibfnamefont {T.}~\bibnamefont {Giamarchi}}, \bibinfo
  {author} {\bibfnamefont {E.}~\bibnamefont {Orignac}},\ and\ \bibinfo {author}
  {\bibfnamefont {M.}~\bibnamefont {Rigol}},\ }\href@noop {} {\bibfield
  {journal} {\bibinfo  {journal} {Reviews of Modern Physics}\ }\textbf
  {\bibinfo {volume} {83}},\ \bibinfo {pages} {1405} (\bibinfo {year}
  {2011})}\BibitemShut {NoStop}%
\bibitem [{\citenamefont {Griffin}\ \emph {et~al.}(1996)\citenamefont
  {Griffin}, \citenamefont {Snoke},\ and\ \citenamefont
  {Stringari}}]{griffin1996bose}%
  \BibitemOpen
  \bibfield  {author} {\bibinfo {author} {\bibfnamefont {A.}~\bibnamefont
  {Griffin}}, \bibinfo {author} {\bibfnamefont {D.~W.}\ \bibnamefont {Snoke}},\
  and\ \bibinfo {author} {\bibfnamefont {S.}~\bibnamefont {Stringari}},\
  }\href@noop {} {\emph {\bibinfo {title} {Bose-einstein condensation}}}\
  (\bibinfo  {publisher} {Cambridge University Press},\ \bibinfo {year}
  {1996})\BibitemShut {NoStop}%
\bibitem [{\citenamefont {Dalfovo}\ \emph {et~al.}(1999)\citenamefont
  {Dalfovo}, \citenamefont {Giorgini}, \citenamefont {Pitaevskii},\ and\
  \citenamefont {Stringari}}]{dalfovo1999theory}%
  \BibitemOpen
  \bibfield  {author} {\bibinfo {author} {\bibfnamefont {F.}~\bibnamefont
  {Dalfovo}}, \bibinfo {author} {\bibfnamefont {S.}~\bibnamefont {Giorgini}},
  \bibinfo {author} {\bibfnamefont {L.~P.}\ \bibnamefont {Pitaevskii}},\ and\
  \bibinfo {author} {\bibfnamefont {S.}~\bibnamefont {Stringari}},\ }\href@noop
  {} {\bibfield  {journal} {\bibinfo  {journal} {Reviews of modern physics}\
  }\textbf {\bibinfo {volume} {71}},\ \bibinfo {pages} {463} (\bibinfo {year}
  {1999})}\BibitemShut {NoStop}%
\bibitem [{\citenamefont {T{\"o}rm{\"a}}\ \emph {et~al.}(2022)\citenamefont
  {T{\"o}rm{\"a}}, \citenamefont {Peotta},\ and\ \citenamefont
  {Bernevig}}]{torma2022superconductivity}%
  \BibitemOpen
  \bibfield  {author} {\bibinfo {author} {\bibfnamefont {P.}~\bibnamefont
  {T{\"o}rm{\"a}}}, \bibinfo {author} {\bibfnamefont {S.}~\bibnamefont
  {Peotta}},\ and\ \bibinfo {author} {\bibfnamefont {B.~A.}\ \bibnamefont
  {Bernevig}},\ }\href@noop {} {\bibfield  {journal} {\bibinfo  {journal}
  {Nature Reviews Physics}\ }\textbf {\bibinfo {volume} {4}},\ \bibinfo {pages}
  {528} (\bibinfo {year} {2022})}\BibitemShut {NoStop}%
\bibitem [{\citenamefont {Fabre}\ and\ \citenamefont
  {Treps}(2020)}]{fabre2020modes}%
  \BibitemOpen
  \bibfield  {author} {\bibinfo {author} {\bibfnamefont {C.}~\bibnamefont
  {Fabre}}\ and\ \bibinfo {author} {\bibfnamefont {N.}~\bibnamefont {Treps}},\
  }\href@noop {} {\bibfield  {journal} {\bibinfo  {journal} {Reviews of Modern
  Physics}\ }\textbf {\bibinfo {volume} {92}},\ \bibinfo {pages} {035005}
  (\bibinfo {year} {2020})}\BibitemShut {NoStop}%
\bibitem [{\citenamefont {Brennecke}\ \emph {et~al.}(2007)\citenamefont
  {Brennecke}, \citenamefont {Donner}, \citenamefont {Ritter}, \citenamefont
  {Bourdel}, \citenamefont {K{\"o}hl},\ and\ \citenamefont
  {Esslinger}}]{brennecke2007cavity}%
  \BibitemOpen
  \bibfield  {author} {\bibinfo {author} {\bibfnamefont {F.}~\bibnamefont
  {Brennecke}}, \bibinfo {author} {\bibfnamefont {T.}~\bibnamefont {Donner}},
  \bibinfo {author} {\bibfnamefont {S.}~\bibnamefont {Ritter}}, \bibinfo
  {author} {\bibfnamefont {T.}~\bibnamefont {Bourdel}}, \bibinfo {author}
  {\bibfnamefont {M.}~\bibnamefont {K{\"o}hl}},\ and\ \bibinfo {author}
  {\bibfnamefont {T.}~\bibnamefont {Esslinger}},\ }\href@noop {} {\bibfield
  {journal} {\bibinfo  {journal} {nature}\ }\textbf {\bibinfo {volume} {450}},\
  \bibinfo {pages} {268} (\bibinfo {year} {2007})}\BibitemShut {NoStop}%
\bibitem [{\citenamefont {Mivehvar}\ \emph {et~al.}(2021)\citenamefont
  {Mivehvar}, \citenamefont {Piazza}, \citenamefont {Donner},\ and\
  \citenamefont {Ritsch}}]{mivehvar2021cavity}%
  \BibitemOpen
  \bibfield  {author} {\bibinfo {author} {\bibfnamefont {F.}~\bibnamefont
  {Mivehvar}}, \bibinfo {author} {\bibfnamefont {F.}~\bibnamefont {Piazza}},
  \bibinfo {author} {\bibfnamefont {T.}~\bibnamefont {Donner}},\ and\ \bibinfo
  {author} {\bibfnamefont {H.}~\bibnamefont {Ritsch}},\ }\href@noop {}
  {\bibfield  {journal} {\bibinfo  {journal} {Advances in Physics}\ }\textbf
  {\bibinfo {volume} {70}},\ \bibinfo {pages} {1} (\bibinfo {year}
  {2021})}\BibitemShut {NoStop}%
\bibitem [{\citenamefont {Gersch}\ and\ \citenamefont
  {Knollman}(1963)}]{gersch1963quantum}%
  \BibitemOpen
  \bibfield  {author} {\bibinfo {author} {\bibfnamefont {H.~A.}\ \bibnamefont
  {Gersch}}\ and\ \bibinfo {author} {\bibfnamefont {G.~C.}\ \bibnamefont
  {Knollman}},\ }\href@noop {} {\bibfield  {journal} {\bibinfo  {journal}
  {Physical Review}\ }\textbf {\bibinfo {volume} {129}},\ \bibinfo {pages}
  {959} (\bibinfo {year} {1963})}\BibitemShut {NoStop}%
\bibitem [{\citenamefont {Guo}\ \emph {et~al.}(2012)\citenamefont {Guo},
  \citenamefont {Weichselbaum}, \citenamefont {von Delft},\ and\ \citenamefont
  {Vojta}}]{guo2012critical}%
  \BibitemOpen
  \bibfield  {author} {\bibinfo {author} {\bibfnamefont {C.}~\bibnamefont
  {Guo}}, \bibinfo {author} {\bibfnamefont {A.}~\bibnamefont {Weichselbaum}},
  \bibinfo {author} {\bibfnamefont {J.}~\bibnamefont {von Delft}},\ and\
  \bibinfo {author} {\bibfnamefont {M.}~\bibnamefont {Vojta}},\ }\href@noop {}
  {\bibfield  {journal} {\bibinfo  {journal} {Physical review letters}\
  }\textbf {\bibinfo {volume} {108}},\ \bibinfo {pages} {160401} (\bibinfo
  {year} {2012})}\BibitemShut {NoStop}%
\bibitem [{\citenamefont {Del~Pino}\ \emph {et~al.}(2018)\citenamefont
  {Del~Pino}, \citenamefont {Schr{\"o}der}, \citenamefont {Chin}, \citenamefont
  {Feist},\ and\ \citenamefont {Garcia-Vidal}}]{del2018tensor}%
  \BibitemOpen
  \bibfield  {author} {\bibinfo {author} {\bibfnamefont {J.}~\bibnamefont
  {Del~Pino}}, \bibinfo {author} {\bibfnamefont {F.~A.}\ \bibnamefont
  {Schr{\"o}der}}, \bibinfo {author} {\bibfnamefont {A.~W.}\ \bibnamefont
  {Chin}}, \bibinfo {author} {\bibfnamefont {J.}~\bibnamefont {Feist}},\ and\
  \bibinfo {author} {\bibfnamefont {F.~J.}\ \bibnamefont {Garcia-Vidal}},\
  }\href@noop {} {\bibfield  {journal} {\bibinfo  {journal} {Physical review
  letters}\ }\textbf {\bibinfo {volume} {121}},\ \bibinfo {pages} {227401}
  (\bibinfo {year} {2018})}\BibitemShut {NoStop}%
\bibitem [{\citenamefont {Kloss}\ \emph {et~al.}(2019)\citenamefont {Kloss},
  \citenamefont {Reichman},\ and\ \citenamefont
  {Tempelaar}}]{kloss2019multiset}%
  \BibitemOpen
  \bibfield  {author} {\bibinfo {author} {\bibfnamefont {B.}~\bibnamefont
  {Kloss}}, \bibinfo {author} {\bibfnamefont {D.~R.}\ \bibnamefont
  {Reichman}},\ and\ \bibinfo {author} {\bibfnamefont {R.}~\bibnamefont
  {Tempelaar}},\ }\href@noop {} {\bibfield  {journal} {\bibinfo  {journal}
  {Physical review letters}\ }\textbf {\bibinfo {volume} {123}},\ \bibinfo
  {pages} {126601} (\bibinfo {year} {2019})}\BibitemShut {NoStop}%
\bibitem [{\citenamefont {Macridin}\ \emph
  {et~al.}(2018{\natexlab{a}})\citenamefont {Macridin}, \citenamefont
  {Spentzouris}, \citenamefont {Amundson},\ and\ \citenamefont
  {Harnik}}]{macridin2018digital}%
  \BibitemOpen
  \bibfield  {author} {\bibinfo {author} {\bibfnamefont {A.}~\bibnamefont
  {Macridin}}, \bibinfo {author} {\bibfnamefont {P.}~\bibnamefont
  {Spentzouris}}, \bibinfo {author} {\bibfnamefont {J.}~\bibnamefont
  {Amundson}},\ and\ \bibinfo {author} {\bibfnamefont {R.}~\bibnamefont
  {Harnik}},\ }\href@noop {} {\bibfield  {journal} {\bibinfo  {journal}
  {Physical Review A}\ }\textbf {\bibinfo {volume} {98}},\ \bibinfo {pages}
  {042312} (\bibinfo {year} {2018}{\natexlab{a}})}\BibitemShut {NoStop}%
\bibitem [{\citenamefont {Macridin}\ \emph
  {et~al.}(2018{\natexlab{b}})\citenamefont {Macridin}, \citenamefont
  {Spentzouris}, \citenamefont {Amundson},\ and\ \citenamefont
  {Harnik}}]{macridin2018electron}%
  \BibitemOpen
  \bibfield  {author} {\bibinfo {author} {\bibfnamefont {A.}~\bibnamefont
  {Macridin}}, \bibinfo {author} {\bibfnamefont {P.}~\bibnamefont
  {Spentzouris}}, \bibinfo {author} {\bibfnamefont {J.}~\bibnamefont
  {Amundson}},\ and\ \bibinfo {author} {\bibfnamefont {R.}~\bibnamefont
  {Harnik}},\ }\href@noop {} {\bibfield  {journal} {\bibinfo  {journal}
  {Physical review letters}\ }\textbf {\bibinfo {volume} {121}},\ \bibinfo
  {pages} {110504} (\bibinfo {year} {2018}{\natexlab{b}})}\BibitemShut
  {NoStop}%
\bibitem [{\citenamefont {Reinhard}\ \emph {et~al.}(2019)\citenamefont
  {Reinhard}, \citenamefont {Mordovina}, \citenamefont {Hubig}, \citenamefont
  {Kretchmer}, \citenamefont {Schollwöck}, \citenamefont {Appel},
  \citenamefont {Sentef},\ and\ \citenamefont {Rubio}}]{reinhard2019density}%
  \BibitemOpen
  \bibfield  {author} {\bibinfo {author} {\bibfnamefont {T.~E.}\ \bibnamefont
  {Reinhard}}, \bibinfo {author} {\bibfnamefont {U.}~\bibnamefont {Mordovina}},
  \bibinfo {author} {\bibfnamefont {C.}~\bibnamefont {Hubig}}, \bibinfo
  {author} {\bibfnamefont {J.~S.}\ \bibnamefont {Kretchmer}}, \bibinfo {author}
  {\bibfnamefont {U.}~\bibnamefont {Schollwöck}}, \bibinfo {author}
  {\bibfnamefont {H.}~\bibnamefont {Appel}}, \bibinfo {author} {\bibfnamefont
  {M.~A.}\ \bibnamefont {Sentef}},\ and\ \bibinfo {author} {\bibfnamefont
  {A.}~\bibnamefont {Rubio}},\ }\href@noop {} {\bibfield  {journal} {\bibinfo
  {journal} {Journal of chemical theory and computation}\ }\textbf {\bibinfo
  {volume} {15}},\ \bibinfo {pages} {2221} (\bibinfo {year}
  {2019})}\BibitemShut {NoStop}%
\bibitem [{\citenamefont {Sandhoefer}\ and\ \citenamefont
  {Chan}(2016)}]{sandhoefer2016density}%
  \BibitemOpen
  \bibfield  {author} {\bibinfo {author} {\bibfnamefont {B.}~\bibnamefont
  {Sandhoefer}}\ and\ \bibinfo {author} {\bibfnamefont {G.~K.-L.}\ \bibnamefont
  {Chan}},\ }\href@noop {} {\bibfield  {journal} {\bibinfo  {journal} {Physical
  Review B}\ }\textbf {\bibinfo {volume} {94}},\ \bibinfo {pages} {085115}
  (\bibinfo {year} {2016})}\BibitemShut {NoStop}%
\bibitem [{\citenamefont {Schr{\"o}der}\ and\ \citenamefont
  {Chin}(2016)}]{schroder2016simulating}%
  \BibitemOpen
  \bibfield  {author} {\bibinfo {author} {\bibfnamefont {F.~A.}\ \bibnamefont
  {Schr{\"o}der}}\ and\ \bibinfo {author} {\bibfnamefont {A.~W.}\ \bibnamefont
  {Chin}},\ }\href@noop {} {\bibfield  {journal} {\bibinfo  {journal} {Physical
  Review B}\ }\textbf {\bibinfo {volume} {93}},\ \bibinfo {pages} {075105}
  (\bibinfo {year} {2016})}\BibitemShut {NoStop}%
\bibitem [{\citenamefont {Woods}\ \emph {et~al.}(2015)\citenamefont {Woods},
  \citenamefont {Cramer},\ and\ \citenamefont {Plenio}}]{woods2015simulating}%
  \BibitemOpen
  \bibfield  {author} {\bibinfo {author} {\bibfnamefont {M.~P.}\ \bibnamefont
  {Woods}}, \bibinfo {author} {\bibfnamefont {M.}~\bibnamefont {Cramer}},\ and\
  \bibinfo {author} {\bibfnamefont {M.~B.}\ \bibnamefont {Plenio}},\
  }\href@noop {} {\bibfield  {journal} {\bibinfo  {journal} {Physical Review
  Letters}\ }\textbf {\bibinfo {volume} {115}},\ \bibinfo {pages} {130401}
  (\bibinfo {year} {2015})}\BibitemShut {NoStop}%
\bibitem [{\citenamefont {Kuwahara}\ \emph {et~al.}(2024)\citenamefont
  {Kuwahara}, \citenamefont {Vu},\ and\ \citenamefont
  {Saito}}]{kuwahara2024effective}%
  \BibitemOpen
  \bibfield  {author} {\bibinfo {author} {\bibfnamefont {T.}~\bibnamefont
  {Kuwahara}}, \bibinfo {author} {\bibfnamefont {T.~V.}\ \bibnamefont {Vu}},\
  and\ \bibinfo {author} {\bibfnamefont {K.}~\bibnamefont {Saito}},\
  }\href@noop {} {\bibfield  {journal} {\bibinfo  {journal} {Nature
  Communications}\ }\textbf {\bibinfo {volume} {15}},\ \bibinfo {pages} {2520}
  (\bibinfo {year} {2024})}\BibitemShut {NoStop}%
\bibitem [{\citenamefont {Schuch}\ \emph {et~al.}(2006)\citenamefont {Schuch},
  \citenamefont {Cirac},\ and\ \citenamefont {Wolf}}]{schuch2006quantum}%
  \BibitemOpen
  \bibfield  {author} {\bibinfo {author} {\bibfnamefont {N.}~\bibnamefont
  {Schuch}}, \bibinfo {author} {\bibfnamefont {J.~I.}\ \bibnamefont {Cirac}},\
  and\ \bibinfo {author} {\bibfnamefont {M.~M.}\ \bibnamefont {Wolf}},\
  }\href@noop {} {\bibfield  {journal} {\bibinfo  {journal} {Communications in
  mathematical physics}\ }\textbf {\bibinfo {volume} {267}},\ \bibinfo {pages}
  {65} (\bibinfo {year} {2006})}\BibitemShut {NoStop}%
\bibitem [{\citenamefont {Nachtergaele}\ \emph {et~al.}(2007)\citenamefont
  {Nachtergaele}, \citenamefont {Raz}, \citenamefont {Schlein},\ and\
  \citenamefont {Sims}}]{nachtergaele2007lieb}%
  \BibitemOpen
  \bibfield  {author} {\bibinfo {author} {\bibfnamefont {B.}~\bibnamefont
  {Nachtergaele}}, \bibinfo {author} {\bibfnamefont {H.}~\bibnamefont {Raz}},
  \bibinfo {author} {\bibfnamefont {B.}~\bibnamefont {Schlein}},\ and\ \bibinfo
  {author} {\bibfnamefont {R.}~\bibnamefont {Sims}},\ }\href@noop {} {\bibfield
   {journal} {\bibinfo  {journal} {arXiv preprint arXiv:0712.3820}\ } (\bibinfo
  {year} {2007})}\BibitemShut {NoStop}%
\bibitem [{\citenamefont {Bartlett}\ \emph {et~al.}(2002)\citenamefont
  {Bartlett}, \citenamefont {Sanders}, \citenamefont {Braunstein},\ and\
  \citenamefont {Nemoto}}]{bartlett2002efficient}%
  \BibitemOpen
  \bibfield  {author} {\bibinfo {author} {\bibfnamefont {S.~D.}\ \bibnamefont
  {Bartlett}}, \bibinfo {author} {\bibfnamefont {B.~C.}\ \bibnamefont
  {Sanders}}, \bibinfo {author} {\bibfnamefont {S.~L.}\ \bibnamefont
  {Braunstein}},\ and\ \bibinfo {author} {\bibfnamefont {K.}~\bibnamefont
  {Nemoto}},\ }\href@noop {} {\bibfield  {journal} {\bibinfo  {journal}
  {Physical Review Letters}\ }\textbf {\bibinfo {volume} {88}},\ \bibinfo
  {pages} {097904} (\bibinfo {year} {2002})}\BibitemShut {NoStop}%
\bibitem [{\citenamefont {Goemans}\ and\ \citenamefont
  {Williamson}(1995)}]{Goemans1995}%
  \BibitemOpen
  \bibfield  {author} {\bibinfo {author} {\bibfnamefont {M.~X.}\ \bibnamefont
  {Goemans}}\ and\ \bibinfo {author} {\bibfnamefont {D.~P.}\ \bibnamefont
  {Williamson}},\ }\href {https://doi.org/10.1145/227683.227684} {\bibfield
  {journal} {\bibinfo  {journal} {Journal of the ACM}\ }\textbf {\bibinfo
  {volume} {42}},\ \bibinfo {pages} {1115–1145} (\bibinfo {year}
  {1995})}\BibitemShut {NoStop}%
\bibitem [{\citenamefont {Hastings}\ and\ \citenamefont
  {O’Donnell}(2022)}]{Hastings2022}%
  \BibitemOpen
  \bibfield  {author} {\bibinfo {author} {\bibfnamefont {M.~B.}\ \bibnamefont
  {Hastings}}\ and\ \bibinfo {author} {\bibfnamefont {R.}~\bibnamefont
  {O’Donnell}},\ }in\ \href {https://doi.org/10.1145/3519935.3519960} {\emph
  {\bibinfo {booktitle} {Proceedings of the 54th Annual ACM SIGACT Symposium on
  Theory of Computing}}},\ \bibinfo {series and number} {STOC ’22}\ (\bibinfo
   {publisher} {ACM},\ \bibinfo {year} {2022})\ p.\ \bibinfo {pages}
  {776–789}\BibitemShut {NoStop}%
\bibitem [{\citenamefont {Navascués}\ \emph {et~al.}(2013)\citenamefont
  {Navascués}, \citenamefont {García-Sáez}, \citenamefont {Acín},
  \citenamefont {Pironio},\ and\ \citenamefont {Plenio}}]{Navascus2013}%
  \BibitemOpen
  \bibfield  {author} {\bibinfo {author} {\bibfnamefont {M.}~\bibnamefont
  {Navascués}}, \bibinfo {author} {\bibfnamefont {A.}~\bibnamefont
  {García-Sáez}}, \bibinfo {author} {\bibfnamefont {A.}~\bibnamefont
  {Acín}}, \bibinfo {author} {\bibfnamefont {S.}~\bibnamefont {Pironio}},\
  and\ \bibinfo {author} {\bibfnamefont {M.~B.}\ \bibnamefont {Plenio}},\
  }\href {https://doi.org/10.1088/1367-2630/15/2/023026} {\bibfield  {journal}
  {\bibinfo  {journal} {New Journal of Physics}\ }\textbf {\bibinfo {volume}
  {15}},\ \bibinfo {pages} {023026} (\bibinfo {year} {2013})}\BibitemShut
  {NoStop}%
\bibitem [{\citenamefont {Tong}\ and\ \citenamefont
  {Kuwahara}(2025)}]{kuwahara25}%
  \BibitemOpen
  \bibfield  {author} {\bibinfo {author} {\bibfnamefont {X.-H.}\ \bibnamefont
  {Tong}}\ and\ \bibinfo {author} {\bibfnamefont {T.}~\bibnamefont
  {Kuwahara}},\ }\href {https://doi.org/10.48550/ARXIV.2509.25572} {\bibinfo
  {title} {Long-range bosonic systems at thermal equilibrium: Computational
  complexity and clustering of correlations}} (\bibinfo {year}
  {2025})\BibitemShut {NoStop}%
\bibitem [{\citenamefont {Han}\ \emph {et~al.}(2026)\citenamefont {Han},
  \citenamefont {Huang}, \citenamefont {Komargodski}, \citenamefont {Lucas},\
  and\ \citenamefont {Popov}}]{han2025entropic}
%
  \BibitemOpen
  \bibfield  {author} {\bibinfo {author} {\bibfnamefont {Y.}~\bibnamefont
  {Han}}, \bibinfo {author} {\bibfnamefont {X.}~\bibnamefont {Huang}}, \bibinfo
  {author} {\bibfnamefont {Z.}~\bibnamefont {Komargodski}}, \bibinfo {author}
  {\bibfnamefont {A.}~\bibnamefont {Lucas}},\ and\ \bibinfo {author}
  {\bibfnamefont {F.~K.}\ \bibnamefont {Popov}},\ }\href@noop {} {\bibfield
  {journal} {\bibinfo  {journal} {Nature Communications}\ }\textbf {\bibinfo
  {volume} {{17}}},\ \bibinfo {pages} {{87}} (\bibinfo {year}
  {{2026}})}\BibitemShut {NoStop}%
\bibitem [{\citenamefont {Aaronson}\ and\ \citenamefont
  {Arkhipov}(2011)}]{Aaronson2011}%
  \BibitemOpen
  \bibfield  {author} {\bibinfo {author} {\bibfnamefont {S.}~\bibnamefont
  {Aaronson}}\ and\ \bibinfo {author} {\bibfnamefont {A.}~\bibnamefont
  {Arkhipov}},\ }in\ \href {https://doi.org/10.1145/1993636.1993682} {\emph
  {\bibinfo {booktitle} {Proceedings of the forty-third annual ACM symposium on
  Theory of computing}}},\ \bibinfo {series and number} {STOC’11}\ (\bibinfo
  {publisher} {ACM},\ \bibinfo {year} {2011})\ p.\ \bibinfo {pages}
  {333–342}\BibitemShut {NoStop}%
\bibitem [{\citenamefont {Tillmann}\ \emph {et~al.}(2013)\citenamefont
  {Tillmann}, \citenamefont {Daki{\'c}}, \citenamefont {Heilmann},
  \citenamefont {Nolte}, \citenamefont {Szameit},\ and\ \citenamefont
  {Walther}}]{tillmann2013experimental}%
  \BibitemOpen
  \bibfield  {author} {\bibinfo {author} {\bibfnamefont {M.}~\bibnamefont
  {Tillmann}}, \bibinfo {author} {\bibfnamefont {B.}~\bibnamefont {Daki{\'c}}},
  \bibinfo {author} {\bibfnamefont {R.}~\bibnamefont {Heilmann}}, \bibinfo
  {author} {\bibfnamefont {S.}~\bibnamefont {Nolte}}, \bibinfo {author}
  {\bibfnamefont {A.}~\bibnamefont {Szameit}},\ and\ \bibinfo {author}
  {\bibfnamefont {P.}~\bibnamefont {Walther}},\ }\href@noop {} {\bibfield
  {journal} {\bibinfo  {journal} {Nature photonics}\ }\textbf {\bibinfo
  {volume} {7}},\ \bibinfo {pages} {540} (\bibinfo {year} {2013})}\BibitemShut
  {NoStop}%
\bibitem [{\citenamefont {Spring}\ \emph {et~al.}(2013)\citenamefont {Spring},
  \citenamefont {Metcalf}, \citenamefont {Humphreys}, \citenamefont
  {Kolthammer}, \citenamefont {Jin}, \citenamefont {Barbieri}, \citenamefont
  {Datta}, \citenamefont {Thomas-Peter}, \citenamefont {Langford},
  \citenamefont {Kundys} \emph {et~al.}}]{spring2013boson}%
  \BibitemOpen
  \bibfield  {author} {\bibinfo {author} {\bibfnamefont {J.~B.}\ \bibnamefont
  {Spring}}, \bibinfo {author} {\bibfnamefont {B.~J.}\ \bibnamefont {Metcalf}},
  \bibinfo {author} {\bibfnamefont {P.~C.}\ \bibnamefont {Humphreys}}, \bibinfo
  {author} {\bibfnamefont {W.~S.}\ \bibnamefont {Kolthammer}}, \bibinfo
  {author} {\bibfnamefont {X.-M.}\ \bibnamefont {Jin}}, \bibinfo {author}
  {\bibfnamefont {M.}~\bibnamefont {Barbieri}}, \bibinfo {author}
  {\bibfnamefont {A.}~\bibnamefont {Datta}}, \bibinfo {author} {\bibfnamefont
  {N.}~\bibnamefont {Thomas-Peter}}, \bibinfo {author} {\bibfnamefont {N.~K.}\
  \bibnamefont {Langford}}, \bibinfo {author} {\bibfnamefont {D.}~\bibnamefont
  {Kundys}}, \emph {et~al.},\ }\href@noop {} {\bibfield  {journal} {\bibinfo
  {journal} {Science}\ }\textbf {\bibinfo {volume} {339}},\ \bibinfo {pages}
  {798} (\bibinfo {year} {2013})}\BibitemShut {NoStop}%
\bibitem [{\citenamefont {Hamilton}\ \emph {et~al.}(2017)\citenamefont
  {Hamilton}, \citenamefont {Kruse}, \citenamefont {Sansoni}, \citenamefont
  {Barkhofen}, \citenamefont {Silberhorn},\ and\ \citenamefont
  {Jex}}]{hamilton2017gaussian}%
  \BibitemOpen
  \bibfield  {author} {\bibinfo {author} {\bibfnamefont {C.~S.}\ \bibnamefont
  {Hamilton}}, \bibinfo {author} {\bibfnamefont {R.}~\bibnamefont {Kruse}},
  \bibinfo {author} {\bibfnamefont {L.}~\bibnamefont {Sansoni}}, \bibinfo
  {author} {\bibfnamefont {S.}~\bibnamefont {Barkhofen}}, \bibinfo {author}
  {\bibfnamefont {C.}~\bibnamefont {Silberhorn}},\ and\ \bibinfo {author}
  {\bibfnamefont {I.}~\bibnamefont {Jex}},\ }\href@noop {} {\bibfield
  {journal} {\bibinfo  {journal} {Physical review letters}\ }\textbf {\bibinfo
  {volume} {119}},\ \bibinfo {pages} {170501} (\bibinfo {year}
  {2017})}\BibitemShut {NoStop}%
\bibitem [{\citenamefont {Bloch}\ \emph {et~al.}(2008)\citenamefont {Bloch},
  \citenamefont {Dalibard},\ and\ \citenamefont {Zwerger}}]{bloch2008many}%
  \BibitemOpen
  \bibfield  {author} {\bibinfo {author} {\bibfnamefont {I.}~\bibnamefont
  {Bloch}}, \bibinfo {author} {\bibfnamefont {J.}~\bibnamefont {Dalibard}},\
  and\ \bibinfo {author} {\bibfnamefont {W.}~\bibnamefont {Zwerger}},\
  }\href@noop {} {\bibfield  {journal} {\bibinfo  {journal} {Reviews of modern
  physics}\ }\textbf {\bibinfo {volume} {80}},\ \bibinfo {pages} {885}
  (\bibinfo {year} {2008})}\BibitemShut {NoStop}%
\bibitem [{\citenamefont {Kollath}\ \emph {et~al.}(2007)\citenamefont
  {Kollath}, \citenamefont {L{\"a}uchli},\ and\ \citenamefont
  {Altman}}]{kollath2007quench}%
  \BibitemOpen
  \bibfield  {author} {\bibinfo {author} {\bibfnamefont {C.}~\bibnamefont
  {Kollath}}, \bibinfo {author} {\bibfnamefont {A.~M.}\ \bibnamefont
  {L{\"a}uchli}},\ and\ \bibinfo {author} {\bibfnamefont {E.}~\bibnamefont
  {Altman}},\ }\href@noop {} {\bibfield  {journal} {\bibinfo  {journal}
  {Physical review letters}\ }\textbf {\bibinfo {volume} {98}},\ \bibinfo
  {pages} {180601} (\bibinfo {year} {2007})}\BibitemShut {NoStop}%
\bibitem [{\citenamefont {Fisher}\ \emph {et~al.}(1989)\citenamefont {Fisher},
  \citenamefont {Weichman}, \citenamefont {Grinstein},\ and\ \citenamefont
  {Fisher}}]{Fisher1989}%
  \BibitemOpen
  \bibfield  {author} {\bibinfo {author} {\bibfnamefont {M.~P.~A.}\
  \bibnamefont {Fisher}}, \bibinfo {author} {\bibfnamefont {P.~B.}\
  \bibnamefont {Weichman}}, \bibinfo {author} {\bibfnamefont {G.}~\bibnamefont
  {Grinstein}},\ and\ \bibinfo {author} {\bibfnamefont {D.~S.}\ \bibnamefont
  {Fisher}},\ }\href {https://doi.org/10.1103/physrevb.40.546} {\bibfield
  {journal} {\bibinfo  {journal} {Physical Review B}\ }\textbf {\bibinfo
  {volume} {40}},\ \bibinfo {pages} {546–570} (\bibinfo {year}
  {1989})}\BibitemShut {NoStop}%
\bibitem [{\citenamefont {Freericks}\ and\ \citenamefont
  {Monien}(1996)}]{freericks1996strong}%
  \BibitemOpen
  \bibfield  {author} {\bibinfo {author} {\bibfnamefont {J.}~\bibnamefont
  {Freericks}}\ and\ \bibinfo {author} {\bibfnamefont {H.}~\bibnamefont
  {Monien}},\ }\href@noop {} {\bibfield  {journal} {\bibinfo  {journal}
  {Physical Review B}\ }\textbf {\bibinfo {volume} {53}},\ \bibinfo {pages}
  {2691} (\bibinfo {year} {1996})}\BibitemShut {NoStop}%
\bibitem [{\citenamefont {Jaksch}\ \emph {et~al.}(1998)\citenamefont {Jaksch},
  \citenamefont {Bruder}, \citenamefont {Cirac}, \citenamefont {Gardiner},\
  and\ \citenamefont {Zoller}}]{jaksch1998cold}%
  \BibitemOpen
  \bibfield  {author} {\bibinfo {author} {\bibfnamefont {D.}~\bibnamefont
  {Jaksch}}, \bibinfo {author} {\bibfnamefont {C.}~\bibnamefont {Bruder}},
  \bibinfo {author} {\bibfnamefont {J.~I.}\ \bibnamefont {Cirac}}, \bibinfo
  {author} {\bibfnamefont {C.~W.}\ \bibnamefont {Gardiner}},\ and\ \bibinfo
  {author} {\bibfnamefont {P.}~\bibnamefont {Zoller}},\ }\href@noop {}
  {\bibfield  {journal} {\bibinfo  {journal} {Physical Review Letters}\
  }\textbf {\bibinfo {volume} {81}},\ \bibinfo {pages} {3108} (\bibinfo {year}
  {1998})}\BibitemShut {NoStop}%
\bibitem [{\citenamefont {Greiner}\ \emph {et~al.}(2002)\citenamefont
  {Greiner}, \citenamefont {Mandel}, \citenamefont {Esslinger}, \citenamefont
  {H{\"a}nsch},\ and\ \citenamefont {Bloch}}]{greiner2002quantum}%
  \BibitemOpen
  \bibfield  {author} {\bibinfo {author} {\bibfnamefont {M.}~\bibnamefont
  {Greiner}}, \bibinfo {author} {\bibfnamefont {O.}~\bibnamefont {Mandel}},
  \bibinfo {author} {\bibfnamefont {T.}~\bibnamefont {Esslinger}}, \bibinfo
  {author} {\bibfnamefont {T.~W.}\ \bibnamefont {H{\"a}nsch}},\ and\ \bibinfo
  {author} {\bibfnamefont {I.}~\bibnamefont {Bloch}},\ }\href@noop {}
  {\bibfield  {journal} {\bibinfo  {journal} {nature}\ }\textbf {\bibinfo
  {volume} {415}},\ \bibinfo {pages} {39} (\bibinfo {year} {2002})}\BibitemShut
  {NoStop}%
\bibitem [{\citenamefont {St{\"o}ferle}\ \emph {et~al.}(2004)\citenamefont
  {St{\"o}ferle}, \citenamefont {Moritz}, \citenamefont {Schori}, \citenamefont
  {K{\"o}hl},\ and\ \citenamefont {Esslinger}}]{stoferle2004transition}%
  \BibitemOpen
  \bibfield  {author} {\bibinfo {author} {\bibfnamefont {T.}~\bibnamefont
  {St{\"o}ferle}}, \bibinfo {author} {\bibfnamefont {H.}~\bibnamefont
  {Moritz}}, \bibinfo {author} {\bibfnamefont {C.}~\bibnamefont {Schori}},
  \bibinfo {author} {\bibfnamefont {M.}~\bibnamefont {K{\"o}hl}},\ and\
  \bibinfo {author} {\bibfnamefont {T.}~\bibnamefont {Esslinger}},\ }\href@noop
  {} {\bibfield  {journal} {\bibinfo  {journal} {Physical review letters}\
  }\textbf {\bibinfo {volume} {92}},\ \bibinfo {pages} {130403} (\bibinfo
  {year} {2004})}\BibitemShut {NoStop}%
\bibitem [{\citenamefont {Bakr}\ \emph {et~al.}(2009)\citenamefont {Bakr},
  \citenamefont {Gillen}, \citenamefont {Peng}, \citenamefont {F{\"o}lling},\
  and\ \citenamefont {Greiner}}]{bakr2009quantum}%
  \BibitemOpen
  \bibfield  {author} {\bibinfo {author} {\bibfnamefont {W.~S.}\ \bibnamefont
  {Bakr}}, \bibinfo {author} {\bibfnamefont {J.~I.}\ \bibnamefont {Gillen}},
  \bibinfo {author} {\bibfnamefont {A.}~\bibnamefont {Peng}}, \bibinfo {author}
  {\bibfnamefont {S.}~\bibnamefont {F{\"o}lling}},\ and\ \bibinfo {author}
  {\bibfnamefont {M.}~\bibnamefont {Greiner}},\ }\href@noop {} {\bibfield
  {journal} {\bibinfo  {journal} {Nature}\ }\textbf {\bibinfo {volume} {462}},\
  \bibinfo {pages} {74} (\bibinfo {year} {2009})}\BibitemShut {NoStop}%
\bibitem [{\citenamefont {Fallani}\ \emph {et~al.}(2007)\citenamefont
  {Fallani}, \citenamefont {Lye}, \citenamefont {Guarrera}, \citenamefont
  {Fort},\ and\ \citenamefont {Inguscio}}]{fallani2007ultracold}%
  \BibitemOpen
  \bibfield  {author} {\bibinfo {author} {\bibfnamefont {L.}~\bibnamefont
  {Fallani}}, \bibinfo {author} {\bibfnamefont {J.}~\bibnamefont {Lye}},
  \bibinfo {author} {\bibfnamefont {V.}~\bibnamefont {Guarrera}}, \bibinfo
  {author} {\bibfnamefont {C.}~\bibnamefont {Fort}},\ and\ \bibinfo {author}
  {\bibfnamefont {M.}~\bibnamefont {Inguscio}},\ }\href@noop {} {\bibfield
  {journal} {\bibinfo  {journal} {Physical review letters}\ }\textbf {\bibinfo
  {volume} {98}},\ \bibinfo {pages} {130404} (\bibinfo {year}
  {2007})}\BibitemShut {NoStop}%
\bibitem [{\citenamefont {Bakr}\ \emph {et~al.}(2010)\citenamefont {Bakr},
  \citenamefont {Peng}, \citenamefont {Tai}, \citenamefont {Ma}, \citenamefont
  {Simon}, \citenamefont {Gillen}, \citenamefont {Foelling}, \citenamefont
  {Pollet},\ and\ \citenamefont {Greiner}}]{bakr2010probing}%
  \BibitemOpen
  \bibfield  {author} {\bibinfo {author} {\bibfnamefont {W.~S.}\ \bibnamefont
  {Bakr}}, \bibinfo {author} {\bibfnamefont {A.}~\bibnamefont {Peng}}, \bibinfo
  {author} {\bibfnamefont {M.~E.}\ \bibnamefont {Tai}}, \bibinfo {author}
  {\bibfnamefont {R.}~\bibnamefont {Ma}}, \bibinfo {author} {\bibfnamefont
  {J.}~\bibnamefont {Simon}}, \bibinfo {author} {\bibfnamefont {J.~I.}\
  \bibnamefont {Gillen}}, \bibinfo {author} {\bibfnamefont {S.}~\bibnamefont
  {Foelling}}, \bibinfo {author} {\bibfnamefont {L.}~\bibnamefont {Pollet}},\
  and\ \bibinfo {author} {\bibfnamefont {M.}~\bibnamefont {Greiner}},\
  }\href@noop {} {\bibfield  {journal} {\bibinfo  {journal} {Science}\ }\textbf
  {\bibinfo {volume} {329}},\ \bibinfo {pages} {547} (\bibinfo {year}
  {2010})}\BibitemShut {NoStop}%
\bibitem [{\citenamefont {Becker}\ \emph
  {et~al.}(2026{\natexlab{a}})\citenamefont {Becker}, \citenamefont {Rouzé},\
  and\ \citenamefont {Salzmann}}]{BeckerRouzeSalzmannToAppearcmp}%
  \BibitemOpen
  \bibfield  {author} {\bibinfo {author} {\bibfnamefont {S.}~\bibnamefont
  {Becker}}, \bibinfo {author} {\bibfnamefont {C.}~\bibnamefont {Rouzé}},\
  and\ \bibinfo {author} {\bibfnamefont {R.}~\bibnamefont {Salzmann}},\ }\href
  {https://arxiv.org/abs/2604.01192} {\bibinfo {title} {Quantum {G}ibbs
  sampling in infinite dimensions: Generation, mixing times and circuit
  implementation}},\ \bibinfo {howpublished}
  {{arXiv:2604.01192}} (\bibinfo
  {year} {2026}{\natexlab{a}}),\ \bibinfo {note} {{Revised companion
  manuscript, October 2026; numbered statements refer to this revision}},\
  \Eprint {https://arxiv.org/abs/2604.01192} {arXiv:2604.01192 [quant-ph]}
  \BibitemShut {NoStop}%
\bibitem [{\citenamefont {Becker}\ \emph
  {et~al.}(2026{\natexlab{b}})\citenamefont {Becker}, \citenamefont
  {Rouz{\'e}},\ and\ \citenamefont
  {Salzmann}}]{BeckerRouzeSalzmannSchroedinger}%
  \BibitemOpen
  \bibfield  {author} {\bibinfo {author} {\bibfnamefont {S.}~\bibnamefont
  {Becker}}, \bibinfo {author} {\bibfnamefont {C.}~\bibnamefont {Rouz{\'e}}},\
  and\ \bibinfo {author} {\bibfnamefont {R.}~\bibnamefont {Salzmann}},\ }\href
  {https://arxiv.org/abs/2604.15263} {\bibinfo {title} {{Computing the free energy of quantum {Coulomb} gases
  via quantum {Gibbs} sampling}}},\ \bibinfo {howpublished}
  {{arXiv:2604.15263}} (\bibinfo
  {year} {2026}{\natexlab{b}}),\ \bibinfo {note} {{Revised companion
  manuscript, October 2026; numbered statements refer to this revision}},\
  \Eprint {https://arxiv.org/abs/2604.15263} {arXiv:2604.15263 [quant-ph]}
  \BibitemShut {NoStop}%
\bibitem [{\citenamefont {Becker}\ \emph
  {et~al.}(2026{\natexlab{c}})\citenamefont {Becker}, \citenamefont
  {Rouz{\'e}},\ and\ \citenamefont
  {Salzmann}}]{BeckerRouzeSalzmannHypercontractivity2026}%
  \BibitemOpen
  \bibfield  {author} {\bibinfo {author} {\bibfnamefont {S.}~\bibnamefont
  {Becker}}, \bibinfo {author} {\bibfnamefont {C.}~\bibnamefont {Rouz{\'e}}},\
  and\ \bibinfo {author} {\bibfnamefont {R.}~\bibnamefont {Salzmann}},\
  }\href@noop {} {\bibinfo {title} {Hypercontractivity and entropy decay for
  infinite-dimensional quantum {Gibbs} samplers}} (\bibinfo {year}
  {2026}{\natexlab{c}}),\ \bibinfo {note} {companion manuscript}\BibitemShut
  {NoStop}%
\bibitem [{\citenamefont {Childs}\ \emph {et~al.}(2014)\citenamefont {Childs},
  \citenamefont {Gosset},\ and\ \citenamefont {Webb}}]{childs2014bose}%
  \BibitemOpen
  \bibfield  {author} {\bibinfo {author} {\bibfnamefont {A.~M.}\ \bibnamefont
  {Childs}}, \bibinfo {author} {\bibfnamefont {D.}~\bibnamefont {Gosset}},\
  and\ \bibinfo {author} {\bibfnamefont {Z.}~\bibnamefont {Webb}},\ }in\
  \href@noop {} {\emph {\bibinfo {booktitle} {International Colloquium on
  Automata, Languages, and Programming}}}\ (\bibinfo {organization}
  {Springer},\ \bibinfo {year} {2014})\ pp.\ \bibinfo {pages}
  {308--319}\BibitemShut {NoStop}%
\bibitem [{Note1()}]{Note1}%
  \BibitemOpen
  \bibinfo {note} {For $\beta >0$ and $\eta \in [0,1/2]$, the
  finite-dimensional Gibbs sampler with Hamiltonian $H$ on $n$ qubits, the $3n$
  single-site Pauli bare jumps $A^\alpha $ and the filters $\protect \widehat
  f_\eta $ from \protect \eqref {eq:BH-filter-main} satisfies $\protect
  \operatorname {gap}(L_{\protect \widehat f_\eta ,H})\ge Cn^{-c}e^{-c\beta
  W}$, where $W=\max \protect \operatorname {Sp}(H)-\min \protect \operatorname
  {Sp}(H)$ and $C,c>0$ denote constants independent of $n,\beta ,H,\eta $. This
  includes $\protect \hatfM $ at $\eta =0$ and gives an exponential lower bound
  in $n$ when $W=\protect \mathcal O(n)$ at fixed $\beta $. To see this, denote
  $\sigma _\beta =e^{-\beta H}/\Tr (e^{-\beta H})$ and $w_\beta (t)=[\beta
  \cosh (2\pi t/\beta )]^{-1}$, and define $B^\alpha (t)=\DOTSB \sum@ \slimits@
  _\nu e^{it\nu +\beta \nu /4}\protect \widehat f_\eta (\nu )A^\alpha _\nu $.
  For $X=\sigma _\beta ^{1/4}Y\sigma _\beta ^{1/4}$, the Dirichlet form is
  \begin {align*} \protect \mathcal {E}_{\protect \widehat f_\eta ,H}(X)=\DOTSB
  \sum@ \slimits@ _\alpha \DOTSI \intop \ilimits@ _\protect \mathbb {R}w_\beta
  (t) \left \|\sigma _\beta ^{1/4}[B^\alpha (t),Y]\sigma _\beta ^{1/4}\right
  \|_2^2\protect \,dt. \end {align*} The multiplier $e^{x/4}\protect \widehat
  f_\eta (x/\beta )$ is a positive mixture of Gaussians of total weight
  $e^{-1/4}$, giving $\left \|B^\alpha (t)\right \|_\infty \le 1$. Its
  reciprocal is analytic for $|\protect \operatorname {Im}x|<1$ and grows only
  exponentially. To recover the bare jumps, we multiply this reciprocal by
  $(1+e^{-b}\cosh x)^{-1}$, with $b=\beta W+\log (16\protect \sqrt {3n})$.
  After substituting $x=\beta \nu $, Fourier inversion, a contour shift and
  truncation give a kernel $k$ supported on $[-T,T]$ such that \begin {align*}
  \left \|A^\alpha -\DOTSI \intop \ilimits@ _{-T}^{T}k(t)B^\alpha (t)\protect
  \,dt\right \|_\infty &\le \protect \frac {1}{4\protect \sqrt {3n}},\\
  K^2:=\DOTSI \intop \ilimits@ _{-T}^{T}\protect \frac {|k(t)|^2}{w_\beta
  (t)}\protect \,dt&\le Cn^c e^{c\beta W}, \end {align*} where $T=\protect
  \mathcal O(\beta [1+\beta W+\log n])$. Subtracting a scalar from $Y$ leaves
  the form unchanged, so assume $\Tr Y=0$. The depolarising gap $4$, the
  triangle inequality and Cauchy-Schwarz give \begin {align*} \protect \sqrt
  8\left \|Y\right \|_2 &\le K\left \|\sigma _\beta ^{-1}\right \|_\infty
  ^{1/2}\protect \mathcal {E}_{\protect \widehat f_\eta ,H}(X)^{1/2} +\protect
  \tfrac 12\left \|Y\right \|_2. \end {align*} Absorbing the last term and
  using $\min _z\left \|X-z\protect \sqrt {\sigma _\beta }\right \|_2^2\le
  \left \|\sigma _\beta \right \|_\infty \left \|Y\right \|_2^2$ yields \begin
  {align*} \protect \operatorname {gap}(L_{\protect \widehat f_\eta ,H}) \ge
  4K^{-2}e^{-\beta W}\ge Cn^{-c}e^{-c\beta W}, \end {align*} after renaming
  $C,c$, since $\left \|\sigma _\beta \right \|_\infty \left \|\sigma _\beta
  ^{-1}\right \|_\infty =e^{\beta W}$}\BibitemShut {NoStop}%
\bibitem [{\citenamefont {Gondolf}\ \emph {et~al.}(2024)\citenamefont
  {Gondolf}, \citenamefont {M{\"o}bus},\ and\ \citenamefont
  {Rouz{\'e}}}]{gondolf2024energy}%
  \BibitemOpen
  \bibfield  {author} {\bibinfo {author} {\bibfnamefont {P.}~\bibnamefont
  {Gondolf}}, \bibinfo {author} {\bibfnamefont {T.}~\bibnamefont {M{\"o}bus}},\
  and\ \bibinfo {author} {\bibfnamefont {C.}~\bibnamefont {Rouz{\'e}}},\
  }\href@noop {} {\bibfield  {journal} {\bibinfo  {journal} {Quantum}\ }\textbf
  {\bibinfo {volume} {8}},\ \bibinfo {pages} {1551} (\bibinfo {year}
  {2024})}\BibitemShut {NoStop}%
\bibitem [{\citenamefont {{Cipriani}}\ \emph {et~al.}(2000)\citenamefont
  {{Cipriani}}, \citenamefont {{Fagnola}},\ and\ \citenamefont
  {{Lindsay}}}]{OU}%
  \BibitemOpen
  \bibfield  {author} {\bibinfo {author} {\bibfnamefont {F.}~\bibnamefont
  {{Cipriani}}}, \bibinfo {author} {\bibfnamefont {F.}~\bibnamefont
  {{Fagnola}}},\ and\ \bibinfo {author} {\bibfnamefont {J.~M.}\ \bibnamefont
  {{Lindsay}}},\ }\href {https://doi.org/10.1007/s002200050773} {\bibfield
  {journal} {\bibinfo  {journal} {Communications in Mathematical Physics}\
  }\textbf {\bibinfo {volume} {210}},\ \bibinfo {pages} {85} (\bibinfo {year}
  {2000})}\BibitemShut {NoStop}%
\bibitem [{Note2()}]{Note2}%
  \BibitemOpen
  \bibinfo {note} {For some orthonormal basis $(\ket {i})_i$ and operator $B=
  \DOTSB \sum@ \slimits@ _{i,j}b_{ij}\ket {i}\protect \!\bra {j}$ satisfying
  $M:=\max \{\sup _i\DOTSB \sum@ \slimits@ _j|b_{ij}|,\sup _j\DOTSB \sum@
  \slimits@ _i|b_{ij}|\}<\infty ,$ Schur's test gives that $B$ is a bounded
  operator with $\|B\| \le M.$}\BibitemShut {NoStop}%
\bibitem [{Note3()}]{Note3}%
  \BibitemOpen
  \bibinfo {note} {Note that this $\delta $ satisfies all requirements we have
  needed so far, namely $\delta < \protect \frac {U}{16}$ and $\protect \frac
  {|\psi |}{\delta }\le 1-\kappa $ for some $\kappa \in (0,1)$.}\BibitemShut
  {Stop}%
\bibitem [{Note4()}]{Note4}%
  \BibitemOpen
  \bibinfo {note} {To see the hopping bound, note that for each neighbouring
  pair, positivity of $(a_i\pm a_j)^\dagger (a_i\pm a_j)$ gives $\pm
  (a_i^\dagger a_j+a_j^\dagger a_i)\le N_i+N_j$. Summing over edges and using
  that every site has at most $2D$ neighbours gives the claimed bound.
  Moreover, to obtain \protect \eqref {BHexp:energy}, note that on $\protect
  \mathcal {H}_q$ the occupation numbers satisfy $\DOTSB \sum@ \slimits@ _i
  m_i^2\ge q^2/n$ by Cauchy-Schwarz, so $\DOTSB \sum@ \slimits@ _iN_i(N_i-1)\ge
  (q^2/n-q)\protect \mathbb {I}$. Combining this with the hopping bound and
  $-\mu q\ge -|\mu |q$ gives \protect \eqref {BHexp:energy}.}\BibitemShut
  {Stop}%
\bibitem [{Note5()}]{Note5}%
  \BibitemOpen
  \bibinfo {note} {To see this define first \protect \[ \begin {aligned}
  \protect \mathcal C&=i[1/\beta ,\infty )\cup (-i[1/\beta ,\infty )),\\
  \protect \mathcal D_{L,\delta } &=\left \{z\in \protect \mathbb {C}:|z|<L,\
  \inf _{w\in \protect \mathcal C}|z-w|>\delta \right \},\\ \Gamma _{L,\delta
  }&=\partial \protect \mathcal D_{L,\delta } \hskip 1em\relax \protect \text
  {with positive orientation}. \end {aligned} \protect \] For $s\in \protect
  \mathbb {R}$ with $L>|s|$ and $0<\delta <1/\beta $, we have $s\in \protect
  \mathcal D_{L,\delta }$ and $g$ is analytic on a neighbourhood of $\protect
  \overline {\protect \mathcal D_{L,\delta }}$. Hence \protect \[ g(s)=\protect
  \frac 1{2\pi i}\DOTSI \ointop \ilimits@ _{\Gamma _{L,\delta }}\protect \frac
  {g(z)}{z-s}\protect \,dz. \protect \] As $L\to \infty $ and $\delta
  \downarrow 0$, the outer-circle contribution is $\protect \mathcal
  O(L^{-1/2})$, and the contour portions within distance $2\delta $ of $i/\beta
  ,-i/\beta ,i\Omega ,-i\Omega $ contribute at most $\protect \mathcal O(\delta
  ^{3/4})$, uniformly in $\eta $. Only the two sides of each cut remain,
  traversed upwards on the left and downwards on the right. Parametrising the
  upper cut by $z=ir$ and the lower cut by $z=-ir$, with $r\ge 1/\beta $, gives
  \protect \eqref {BHexp:slit}.}\BibitemShut {Stop}%
\end{thebibliography}
\end{document}